\documentclass[superscriptaddress, twocolumn,
notitlepage, amsmath, amssymb, aps,
pra, letterpaper]{revtex4-1}

\usepackage{}

\usepackage{graphicx,graphics,epsfig,subfigure,times,bm,bbm,amssymb,amsmath,amsfonts,amsthm,mathrsfs,MnSymbol}
\usepackage[matrix,frame,arrow]{xypic}
\usepackage[normalem]{ulem}
\usepackage{slashed}
\usepackage{color}
\usepackage[usenames,dvipsnames,svgnames,table]{xcolor}
\definecolor{darkblue}{rgb}{0.0,0.0,0.3}
\usepackage[colorlinks=true,
            linkcolor=red,
            urlcolor= darkblue,
            citecolor=blue]{hyperref}

\usepackage{enumerate}
\usepackage{epstopdf}
\usepackage{relsize}
\newtheorem{definitionenv}{Definition}
\newtheorem{remarkenv}[definitionenv]{Remark}
\newenvironment{remark}{\begin{remarkenv}\rm}{\end{remarkenv}}

\newtheorem{exampleenv}{Example}
\newenvironment{example}{\begin{exampleenv}\rm}{\end{exampleenv}}

\newtheorem{thm}{Theorem}
\newtheorem{mydef}{Definition}

\newtheorem{mylemma}{Lemma}

\newcommand{\mcal}{\mathcal}

\newcommand{\bes} {\begin{subequations}}
\newcommand{\ees} {\end{subequations}}
\newcommand{\bea} {\begin{eqnarray}}
\newcommand{\eea} {\end{eqnarray}}

\newcommand{\beq}{\begin{equation}}

\newcommand{\eeq}{\end{equation}}

\newcommand{\upi}{\mathrm{i}}
\newcommand{\ignore}[1]{}

\def\>{\rangle}
\def\<{\langle}

\def\Pr{\mathrm{Pr}}

\newcommand{\ket}[1]{|#1\rangle}

\newcommand{\sfA}{\textsf{A}}

\newcommand{\sfG}{\textsf{G}}
\newcommand{\sfH}{\textsf{H}}
\newcommand{\sfI}{\textsf{I}}
\newcommand{\sfM}{\textsf{M}}

\newcommand{\cC}{\mathcal{C}}

\newcommand{\cQ}{\mathcal{Q}}

\newcommand{\step}{\tau} 

\begin{document}
\title{Efficient preparation of large block code ancilla states for fault-tolerant quantum computation}
\author{Yi-Cong Zheng}

\email{
zheng.yicong@quantumlah.org}
\affiliation{Centre for Quantum Technologies, National University of Singapore, Singapore, 117543}
\affiliation
{Yale-NUS College, Singapore, 138527}

\author{Ching-Yi Lai}
\affiliation{Institute of Information Science, Academia Sinica, Taipei 11529, Taiwan}
\author{Todd A. Brun}
\affiliation{Ming Hsieh Department of Electrical Engineering, Center for Quantum Information Science \& Technology, University of Southern California, Los Angeles, California 90089, USA\\}

\begin{abstract}
Fault-tolerant quantum computation (FTQC) schemes that use multi-qubit large block codes can potentially reduce the resource overhead to a great extent. A major obstacle is the requirement of a large number of clean ancilla states of different types without correlated errors inside each block. These ancilla states are usually logical stabilizer states of the data code blocks, which are generally difficult to prepare if the code size is large. Previously we have proposed an ancilla distillation protocol for Calderbank-Shor-Steane (CSS) codes by classical error-correcting codes.
It was assumed that the  quantum gates in the distillation circuit were perfect; however, in reality, noisy quantum gates may introduce correlated errors that are not treatable by the protocol.
In this paper, we show that additional postselection by another classical error-detecting code can be applied to remove almost all correlated errors.
Consequently, the revised protocol is fully fault-tolerant and capable of preparing a large set of stabilizer states sufficient for FTQC using large block codes. At the same time, the yield rate can be boosted from $O(t^{-2})$ to $O(1)$ in practice for an $[[n,k,d=2t+1]]$ CSS code. Ancilla preparation for the $[[23,1,7]]$ quantum Golay code is numerically studied in detail through Monte Carlo simulation. The results support the validity of the protocol when the gate failure rate is reasonably low. To the best of our knowledge, this approach is the first attempt to prepare general large block stabilizer states free of correlated errors for FTQC in a fault-tolerant and efficient manner.

\end{abstract}
\maketitle


\section{Introduction}
Quantum computers are difficult to build, since quantum states suffer from decoherence and quantum gates are imperfect~\cite{Nielsen:2000:CambridgeUniversityPress}. However, the theory of quantum error-correcting codes (QECCs)~\cite{Shor:1995:R2493,Steane:1996:793,Gaitan:2008:CRC,QECbook:2013} and fault-tolerant quantum computation (FTQC) ~\cite{Shor:1996:56,Aharonov:1997:176,
Gottesman:9705052,Kitaev:2003:2,
DivencenzoFTPhysRevLett.77.3260, Terhal:2005:012336,KnillFTNature,Aharonov:2006:050504,Folwer2012PhysRevA.86.032324,
QECbook:2013} have shown that an arbitrarily large-scale quantum computation can be achieved if errors are not strongly correlated and their rates are small enough to fall below a threshold~\cite{Aharonov:1997:176,KnillFTNature,
Terhal:2005:012336,Aharonov:2006:050504,Aliferis:2006:97,cross2007comparative,
Aliferis:2008:181}. This is achieved by encoding the state of the computation in a quantum error-correcting code.
Decoherence cannot corrupt an encoded state 
without affecting a large number of qubits.

Larger QECCs, with higher distances and rates, can potentially outperform smaller codes. Various FTQC schemes using multi-qubit large block codes have been proposed to achieve significantly higher code
rates for comparable error protection ability~~\cite{steane1999efficient_Nature, Steane:2003:042322, Steane:300bits, gottesman2013overhead,brun2015teleportation}.
For example, in Ref.~\cite{brun2015teleportation}, multiple logical
qubits are encoded into high-rate large block codes and fault-tolerant (FT) logical gates are implemented by measuring logical states and teleporting logical qubits states back and forth between different code blocks to achieve universality, so that magic state distillation~\cite{Bravyi:2005:022316,Bravyi_Haah_PhysRevA.86.052329} can be avoided. This can be done by preparation of different stabilizer ancilla states, transversal circuits, and bitwise qubit measurements.

High code rates and fault-tolerant non-Clifford gates without magic state distillation are two features that make the FTQC scheme based on large block codes promising. It potentially offers small resource overhead (see also~\cite{paetznick2013universal,
jochym2014using,anderson2014fault,bravyi2015doubled_color,
yoder2016universal} for other schemes with these features).
However, ancilla states are required for logical teleportation and error correction. Straightforward preparation of these ancillas  (e.g., through encoding circuits), is not fault-tolerant:
a single circuit failure can result in a \emph{high-weight} error in
the output state with high probability. This is called a \emph{correlated} error. As a result, the ancilla states need to be verified before they can be used.

The complexity of verifying prepared ancillas (in a straightforward way) grows quickly with code distance. For large codes, verification of encoded logical states such as $|0\>^{\otimes k}_L$ and $|+\>^{\otimes k}_L$ is accomplished by certain identically-prepared auxiliary ancillas. Errors from the
target ancilla are copied into the auxiliary ancillas, which are then bitwise measured. If the measurement outcomes suggest that no error occurs, the ancillas are accepted.
Otherwise,  all the ancillas are discarded and the entire process restarts.
Note that the auxiliary ancillas may also contain correlated errors after their preparation. These errors may propagate back to the target ancilla and the verification may fail even if the target ancilla is clean. Thus, the auxiliary ancillas must also be verified by more auxiliary ancillas. These recursive verifications will eventually dominate the overall resource overhead of quantum computation. In general, for an $[[n,k,d=2t+1]]$ CSS code, it requires $O(t^2)$ noisy ancilla states to generate a qualified one.

More efficient ways to prepare ancillas have been explored in the literature. Previously, verification of the ancilla state $|0\>_L$ has been studied for some quantum codes, such as the Steane code and the quantum Golay  code~\cite{steane2002fast,Steane:2003:042322,reichardt2004improved,divincenzo_aliferis2007,paetznick_ben2011fault}.
In Ref.~\cite{steane2002fast}, Steane proposed a method to filter out correct $|0\>_L^{\otimes k}$ and $|+\>_L^{\otimes k}$ for arbitrary Calderbank-Shor-Steane (CSS) codes~\cite{Calderbank:1996:1098,Steane:1996:793} using only single-round stabilizer measurements with postselection. However, it is highly likely that the overall rejection rate is too high to be of practical value when the code length is large. Moreover, that method can only remove one type (either $X$ or $Z$) of correlated errors, which is good for the scheme in Ref.~\cite{steane1999efficient_Nature} but not sufficient for the one in Ref.~\cite{brun2015teleportation}. In Ref.~\cite{paetznick_ben2011fault},  the permutation symmetry of the quantum Golay code is exploited by permuting the qubits of different code blocks so that the correlation of errors between code blocks can be suppressed before verification.
This simple operation can reduce the resource overhead of ancilla preparation by a factor of four.
In Ref.~\cite{goto2016minimizing}, an efficient method was proposed to fault-tolerantly prepare $|0\>_L$ of the Steane code using only one extra qubit for verification. However, it is unclear how to generalize the above methods to arbitrary  large block codes.

In Ref.~\cite{Ancilla_distillation_1}, we proposed a general distillation protocol to prepare CSS stabilizer states, such as $|0\>^{\otimes k}_L$ and $|+\>^{\otimes k}_L$ states, for arbitrary CSS codes by classical error-correcting codes.
The idea of the distillation protocol  is as follows. Many ancilla states are identically prepared, each of which typically may contain highly correlated errors. They are fed into a distillation circuit, and certain check blocks are measured bitwise. The error syndromes of the remaining blocks can be extracted when the distillation circuit is based on the structure of a classical code. From the measurements, one can estimate the error syndromes of the target blocks by classical decoding.
The protocol works correctly when the distillation circuit is perfect. However, if the CNOT gates in the distillation circuits are noisy, they can introduce correlated errors between different ancilla blocks so that classical decoding may not be able to recover these errors.
Thus, the protocol is not completely fault-tolerant: it was assumed that the distillation circuit is perfect in  Ref.~\cite{Ancilla_distillation_1}, while in reality, noisy gates and measurements may lead to correlated errors  after distillation.

In this paper we will modify the distillation protocol of Ref.~\cite{Ancilla_distillation_1} so that it works \emph{fault-tolerantly} and \emph{efficiently} when the distillation circuit is imperfect.
The goal is for each block after fault-tolerant preparation to contain \emph{no} correlated errors.
Identifying the correct error syndromes for the ancillas is a key step in our distillation protocol.
It is possible to introduce parity checks on the error syndromes by using another classical error-correcting code or an appropriate quantum stabilizer code~\cite{ashikhmin_chingyi2014robust,fujiwara2014ability,ashikhmin2016correction}, which greatly simplifies the process of syndrome verification.
We apply a similar idea here to compute additional syndrome bits encoded by classical \emph{error-detecting} codes.
The estimated syndromes are checked to see whether they are compatible with the additional syndrome bits.
If they agree, we accept the output; otherwise, they are discarded.
After that, quantum error correction can be applied to the remaining blocks to remove both $X$ and $Z$ errors based on the estimated syndromes.
Consequently, cleaner ancillas with no correlated errors can be obtained. The yield of the process depends on the rates of the chosen classical codes and the rejection rate of the output blocks. If an appropriate family of classical codes is chosen for postselection, the yield rate can be boosted from $O(t^{-2})$ to $O(1)$ in practice for an $[[n,k,d=2t+1]]$ CSS code.

Moreover, we prove that the output ancillas from this modified distillation protocol contain no correlated errors if sufficiently powerful classical codes are chosen.
We prove a no-go theorem (Theorem~\ref{thm:nogo}) that a noisy distillation circuit will have correlated errors with high probability,
which suggests that additional treatment is required when the distillation circuit is noisy. Our method of postselection by classical error-detecting codes
is proved to be effective here (Theorem~\ref{thm:main}) when the underlying quantum code is \emph{sparse} (see Def.~\ref{def:epsilon-sparse} and Lemma~\ref{lemma:valid_to_good})
and the chosen classical code has error-correcting ability at least as good as the sparse quantum code (Theorems~\ref{lemma:postselection_nogo} and~\ref{thm:3}).
This may even be true for general quantum codes (not necessarily sparse) and we demonstrate this by simulations of the distillation of the $[[23,1,7]]$ quantum Golay code in Sec.~\ref{sec:numerical}.

In addition, we also propose protocols to distill various entangled ancilla states necessary for FTQC, such as the encoded EPR pairs between two code blocks. (The details are postponed to Appendix~\ref{appendix:ancilla_states}.) In particular, we show how to do error correction and simultaneously measure a logical $Y$ by Steane syndrome extraction.

The paper is organized as follows. We provide preliminary materials in Sec.~\ref{sec:prim}, including the basic ideas of quantum error-correcting codes, CSS codes and Steane syndrome extraction.
In Sec.~\ref{sec:noiseless_prep}, we review the protocol of Ref.~\cite{Ancilla_distillation_1} and formulate it for general multiple-qubit CSS codes.
In Sec.~\ref{sec:FTpreparation}, we propose a \emph{fully} fault-tolerant protocol for  these stabilizer ancilla states.
In Sec.~\ref{sec:numerical}, we numerically study the procedure for preparing $|0\>_L$ for the quantum Golay code by estimating the error weight distribution of the output ancillas for different combinations of small-sized classical codes.
Discussion and suggestions for further improvements are presented in Sec.~\ref{sec:discussion}.

\section{Preliminaries}\label{sec:prim}
Both classical and quantum codes play essential roles in this paper. We start with a brief review of classical and quantum error-correcting codes and introduce our notation.

\subsection{Classical Linear Codes}
Classical error-correcting codes protect digital information by introducing redundancy. The encoded information strings---\emph{codewords}--- need to satisfy a set of \emph{parity checks}, that is, linear constraints.
Errors that cause violations of the parity checks can be detected and (hopefully) be corrected.

Let $\sfH$ be an $(n-k)\times n$ binary matrix. An $[n,k,d]$ classical binary code $\mcal{C}$ associated with \emph{parity-check matrix} $\sfH$ is a $k$-dimensional subspace of all binary $n-$tuples in $\mathbb{Z}_2^n$ such that
\beq
\sfH v^T=0,
\eeq
for all $v\in\mcal{C}$, where $v^T$ is the transpose of $v$ and addition is modulo $2$. Such vectors $v$ are the codewords of $\mcal{C}$. The parameter $d$ is called the \emph{minimum distance} of $\mcal{C}$ such that any two codewords in $\mcal{C}$ differ in at least $d$ bits. This code can correct arbitrary $\lfloor \frac{d-1}{2}\rfloor$-bit errors.

For a received string $\tilde{v}\in \mathbb{Z}_2^n$, if $\sfH \tilde{v}^T\neq 0$, we know that some error occurred.
Hence the rows of $\sfH$ are called \emph{parity checks} of $\mathcal{C}$, and $\sfH \tilde{v}^T$ is called the \emph{error syndrome} of $\tilde{v}$. A useful property of a linear code is that its parity-check matrix can be written in the systematic form:
\beq
\sfH=\left[\sfI_{n-k}|\sfA_{(n-k)\times k}\right].
\eeq
$\mcal{C}$ can also be defined as the row space of a $k\times n$ generator matrix $\sfG$, which satisfies
\beq
\sfH\sfG^T=0.
\eeq
The dual code $\mcal{C}^{\perp}$ of $\mcal{C}$  is the row space of $\sfH$.
More properties of classical codes can be found in~\cite{MacWilliams:1983:NorthHolland}.

\subsection{Stabilizer Codes and CSS Codes}
The Hilbert space of a single qubit is the two-dimensional complex vector space $\mathbb{C}^2$ with an orthonormal basis $\{|0\>, |1\>\}$. Then the Hilbert space of $n-$qubit state is $\mathbb{C}^{2^n}$.
Let $\mathcal{P}_n=\mathcal{P}_1^{\otimes n}$ denote the $n$-fold Pauli group, where
\beq
\mathcal{P}_1=\{\pm I, \pm \upi I, \pm X, \pm \upi X, \pm Y, \pm \upi Y, \pm Z, \pm \upi Z\},
\eeq
and
$$
I=\left[\begin{array}{cc}
1 & 0 \\
0 & 1
\end{array}\right]
, \ X=\left[
\begin{array}{cc}
0 & 1 \\
1 & 0 \\
\end{array}
\right]
,\ \ Y=
\left[
  \begin{array}{cc}
    0 & -\upi \\
    \upi & 0 \\
  \end{array}
\right]
, \ \
 Z=\left[
 \begin{array}{cc}
 1 & 0 \\
 0 & -1 \\
 \end{array}
 \right].
$$

For simplicity, we use the notation $X_j$ to denote an $X$ on qubit $j$, $I^{\otimes j-1}\otimes X\otimes I^{\otimes n-j}$, where $n$ is the number of qubits. $Y_j$ and $Z_j$ are defined similarly. We also introduce the notation $X^e$ for $e=e_1\cdots e_n\in \mathbb{Z}_2^n$, to denote the operator $\otimes_{i=1}^n X^{e_i}$. An $n$-fold Pauli operator can be expressed as
\beq\label{eq:general_error}
\bigotimes_{i=1}^n X^{e_i}Z^{f_i}=X^eZ^f, \qquad e,f\in \mathbb{Z}^n_2.
\eeq
up to an overall phase. Here $(e,f)$ is called the
\emph{binary representation} of the Pauli operator $X^eZ^f$. We define the weight of $E$, $\text{wt}(E)$, as the number of terms in the tensor product which are not equal to the identity.

Suppose  $\mathcal{G}$ is an Abelian subgroup of $\mathcal{P}_n$
with a set of $q$ independent and commuting generators $\{G_1,\dots, G_{q}\}$, and $\mathcal{G}$ does not include $-I^{\otimes n}$.
Every element in $\mcal{P}_n$ has eigenvalues $\pm 1$ or $\pm i$.
An $[[n,k=n-q]]$ quantum stabilizer code $C(\mathcal{G})$
is defined as the $2^{k}$-dimensional subspace of the $n$-qubit state space ($\mathbb{C}^{2^n}$) fixed by  $\mathcal{G}$,
which is the joint-$(+1)$ eigenspace of $G_1, \dots, G_{n-k}$.
Then for a codeword $\ket{\psi}\in C(\mcal{G})$, $$G\ket{\psi}=\ket{\psi}$$ for all $G\in \mathcal{G}$. When $k=0$, $C(\mathcal{G})$ has only one eigenstate up to a global phase, and this state is called the \emph{stabilizer state} of $\mathcal{G}$. Special stabilizer states play crucial roles in FTQC using large block codes~\cite{steane1999efficient_Nature,Steane:2003:042322,
brun2015teleportation}, and the main object of this paper is to prepare these states fault-tolerantly.

If a Pauli error $E$ corrupts $\ket{\psi}$, some eigenvalues of $G_1,\dots, G_{q}$ will be flipped.
Consequently, we gain error information by measuring the stabilizer generators $G_1,\dots, G_{q}$,
and the corresponding measurement outcomes (in bits),
$g_1,\dots, g_{q}$, are called the \emph{error syndrome} of $E$.
(The eigenvalue of a stabilizer generator is $+1$ or $-1$, and its corresponding syndrome bit is then $0$ or $1$, respectively. In the rest of the paper, the measurement outcomes are always represented in binary form.)
A quantum decoder has to choose a good recovery operation based on the measured error syndromes.

CSS codes are an important class of stabilizer codes for FTQC. Their generators consist of tensor products of the identity and either $X$ or $Z$ operators (but not both)~\cite{Calderbank:1996:1098, Steane:1996:793}. More formally, consider two classical codes, $\mathcal{C}_Z$ and $\mathcal{C}_X$ with parameters $[n, k_Z, d_Z]$ and $[n, k_X, d_X]$, respectively, such that $\mathcal{C}_X^\perp \subset \mathcal{C}_Z$. The corresponding parity-check matrices are $\sfH_Z$ $(r_Z\times n)$ and $\sfH_X$ $(r_X\times n)$ with full rank $r_Z=n-k_Z$ and $r_X=n-k_X$.
One can form an $[[n, k=k_X+k_Z-n, d]]$ CSS code $\mathcal{Q}$, where $d\geq\min\{d_Z,d_X\}$. For the special case that $\cC_{X}=\cC_{Z}$, we call such code a \emph{symmetric} CSS code. In general a logical state of $\mathcal{Q}$ can be represented as:
\beq\label{eq:CSS_codeword}
|u\>_L=\sum_{x\in \mathcal{C}_{X}^\perp}|x+u  D\>,
\eeq
where $u\in \mathbb{Z}_2^k$ and $D$ is a $k\times n$ binary  matrix, whose rows are the coset leaders of $\mathcal{C}_Z/\mathcal{C}_X^\perp$.
The $Z$ and $X$ stabilizer generators of $\mathcal{Q}$ are
\beq
G_{Z\,|\,i}=\bigotimes_{j=1}^n Z^{[\sfH_Z]_{i,j}}, \qquad i=1, \dots, r_Z,
\eeq
and
\beq
G_{X\,|\,i}=\bigotimes_{j=1}^n X^{[\sfH_X]_{i,j}}, \qquad i=1, \dots, r_X.
\eeq
Here $[\sfM]_{i,j}$ represents the $(i,j)$ entry of a matrix $\sfM$. The \emph{check matrix} for $\mathcal{Q}$ is defined as 
\beq\label{eq:quantum_parity}
\sfH = \left[
         \begin{array}{cc}
           \sfH_Z & 0 \\
           0 & \sfH_X \\
         \end{array}
       \right].
\eeq

The error syndrome of a Pauli error is a binary string of outcomes of measuring the stabilizers $G_1, \dots ,  G_{n-k}$. According to Eq.~(\ref{eq:general_error}), error correction can be done by treating   $X$ and $Z$ errors separately. For CSS codes, the eigenvalues of $G_{Z\,|\,1}, \dots, G_{Z\,|\,r_Z}$ $(G_{Z\,|\,1},\dots, G_{Z\,|\,r_X})$ correspond to the error syndrome of $X$ $(Z)$ errors. Then its $X$ and $Z$ error syndromes $g_Z\in \mathbb{Z}_2^{r_Z}$, $g_X  \in \mathbb{Z}_2^{r_X}$ are given by
\beq
g_Z^T=\sfH_Z e^T, \ \ \  g_X^T =\sfH_X f^T.
\eeq

The Pauli operators that are not in $\mathcal{G}$ but commute with all the stabilizer generators are called the \emph{logical Pauli operators}. For the CSS codes defined in Eq.~(\ref{eq:CSS_codeword}), it is always possible to find a subset of logical operators involving only tensor product of identity and either $X$ or $Z$ operators. Explicitly, the logical  $X$ operators for the $j$th logical qubit $\bar{X}_j$ can be written as
\beq
\bar{X}_j=\bigotimes_{i=0}^n X_i^{[D]_{j,i}}.
\eeq
The logical $Z$ operators for the $j$th logical qubit can be constructed as~\cite{Steane:1996:793,steane1999efficient_Nature}:
\beq
\bar{Z}_j = \bigotimes Z_i^{[(DD^T)^{-1}D]_{j,i}}.
\eeq

\subsection{Steane Syndrome Extraction}\label{sec:ancilla_states}
For a CSS code $\mathcal{Q}$, Steane suggested a method to extract syndromes as shown in Fig.~\ref{fig:steane_type}~\cite{steane1997active}. This method needs two clean ancilla states, which are called the $X$ and $Z$ ancillas, since they will be bitwise measured in the $X$ and $Z$ bases, respectively. If the $X$ and $Z$ ancillas are in the logical
\begin{figure}[!ht]
\centering\includegraphics[width=80mm]{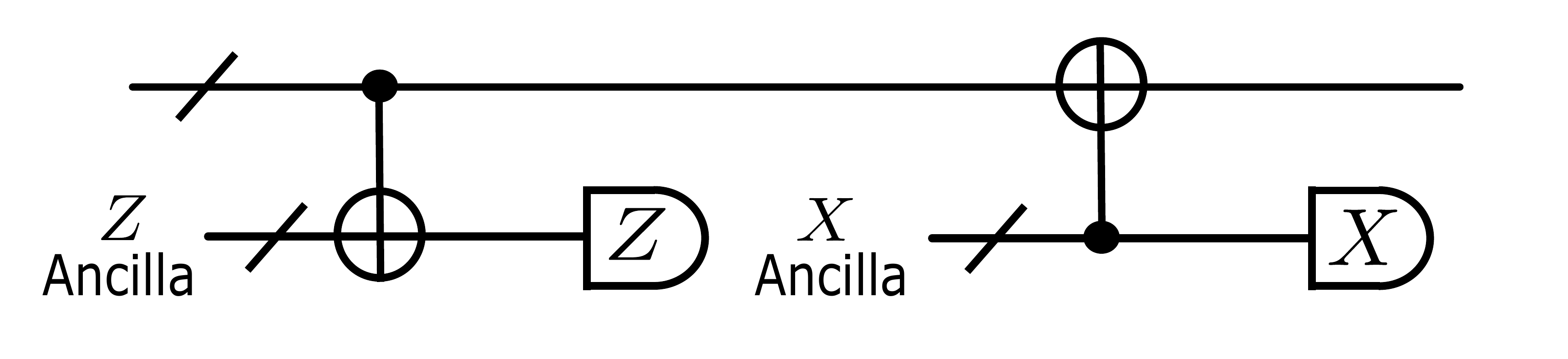}
\caption{\label{fig:steane_type} Quantum circuit for Steane syndrome extraction. }
\end{figure}
states $|0\>_L$ and $|+\>_L=\frac{1}{\sqrt{2}}\left(|0\>_L+|1\>_L\right)$ of $\mathcal{Q}$, the circuit extracts the error syndromes without disturbing the encoded quantum information. If other ancilla states are used, it is possible also to measure a logical Pauli operator on the data block. Since a logical gate can be implemented by measuring a sequence of logical operators~\cite{brun2015teleportation}, this leads to a fault-tolerant implementation of logical gate operation and syndrome extraction, simultaneously. As we can see in Appnedix.~\ref{appendix:ancilla_state_1}, these ancillas are stabilizer states of $\mathcal{Q}$. Namely, they are stabilized by both the stabilizer generators of $\mathcal{G}$ and some logical Pauli operators. For a logical ancilla of a single code block, its corresponding set of stabilizers is:
\beq
S=\{G_{X\,|\,1},\dots , G_{X\,|\,r_X},G_{Z\,|\,1},\dots , G_{Z\,|\,r_Z}, L_1,\cdots, L_k \}.
\eeq
where the $L_i$ are logical Pauli operators that commute with each other. Note that we may also need ancilla states that are entangled among $m > 1$ different blocks. The stabilizer generators and logical Pauli operators of such a logical ancilla can be represented by:
\beq
\begin{split}
S=\{&G^{(1)}_{X\,|\,1},\dots , G^{(1)}_{X\,|\,r_X},G^{(1)}_{Z\,|\,1},\dots , G^{(1)}_{Z\,|\,r_Z}, \dots \\
&G^{(m)}_{X\,|\,1},\dots , G^{(m)}_{X\, | \, r_X}, G^{(m)}_{Z\,|\,1},\dots , G^{(m)}_{Z \,|\, r_Z}, L_1,\dots, L_{mk}\}
\end{split}
\eeq
Two simple examples are the logical $|0\>^{\otimes k}_L$ and $|+\>^{\otimes k}_L$ states, where all logical qubits are in the states $|0\>_L$ and $|+\>_L$, respectively. These two states are stabilized by
\beq\label{eq:stab_0}
S=\{G_{X\,|\,1},\dots, G_{X\,|\,r_X}, G_{Z\,|\,1},\dots , G_{Z\,|\,r_Z}, \bar{Z}_1, \dots,  \bar{Z}_k\},
\eeq
and
\beq
S=\{G_{X\,|\,1},\dots, G_{X\,|\,r_X}, G_{Z\,|\,1},\dots , G_{Z\,|\,r_Z}, \bar{X}_1, \dots, \bar{X}_k\},
\eeq
respectively. In the rest of the paper, these stabilizer generators and logical Pauli operators will be called \emph{stabilizers} for short. The measurement outcomes of their eigenvalues in binary form are called \emph{generalized syndromes} for simplicity.

\subsection{Failures and Errors}


For clarification, we need to distinguish between \emph{failures} and \emph{errors}. Failures are physical processes which cause imperfection on certain qubits at particular time in the circuit. The resulting effect of failures on the quantum state are called \emph{errors}. In general, a single-qubit failure in the circuit can result in multiple-qubit errors.


A basic assumption here is that failures in different places and times of a quantum circuit are \emph{independent}. This assumption can be relaxed, but for simplicity, we retain it for this paper. In this paper, we consider only Pauli failures. Assume that at each time step, every physical qubit independently
undergoes an $X$, $Y$ or $Z$ error with probability $\epsilon/3$. Such failure is called a \emph{memory failure}. For a CNOT gate, failure is modeled as a perfect gate followed by one of the 15 possible 
failures from $IX$, $IY$ , $IZ$, $XI$, $XX$, $XY$ , $XZ$, $YI$, $YX$, $YY$ , $YZ$, $ZI$, $ZX$, $ZY$, and $ZZ$ with equal probability $p_g/15$. We call this a \emph{gate failure}. Also, a measurement of a single physical qubit suffers a classical bit-flip error with probability $p_m$ (an $X$ or $Z$ error preceding
a measurement in the $Z$ or $X$ basis, respectively). We call this a \emph{measurement failure}. In the rest of the paper, we assume that the gate failure rate is equal to the measurement failure rate---i.e., $p_g=q_m=p$. Also, we consider the scenario where the memory failure rate is much smaller than the gate failure rate, that is, $\epsilon \ll p$, so that it can be safely neglected. This is usually the case for a physical system with fast gates and long coherence time, like ion trap~\cite{ion-trap_08} or superconducting qubits~\cite{barends_Martinis2014superconducting}.

Correlated errors are usually very harmful for FTQC. The following definition is based on  Ref.~\cite{steane2002fast}:
\begin{mydef}\label{def:uncorrelatation}
Consider a quantum code $\mathcal{Q}$ which can correct any Pauli error on $t$ qubits. We say that an error $E$ is \emph{uncorrelated} if the probability of $E$ after the preparation of a codeword of $\mathcal{Q}$ scales like
\beq
\Pr(E)\sim O(p^s):
\begin{cases}
&\text{for some }s\geq\text{wt}(E),\ \text{if wt}(E)\leq t; \\
&\text{for some }s \geq t,\ \text{if wt}(E)> t.
\end{cases}
\eeq
\end{mydef}
\noindent
Otherwise, $E$ is said to be \emph{correlated}. The parameter $s$ is called the \emph{order} of error $E$. Note that the coefficients behind the big $O$ notation should not be unreasonably large. Note also that, for uncorrelated errors, the probability distribution of their weights is somewhat like a binomial distribution for error order less than $t$.

To realize the FTQC protocol in~\cite{steane1999efficient_Nature,Steane:2003:042322,brun2015teleportation}, the ancillas should contain no correlated errors, so that correlated errors will not propagate from the ancillas to the data block, introducing uncorrectable errors. Even if errors on an ancilla are correctable, if their order is not strictly greater than or equal to their weights, then they can propagate to the data blocks and cause quantum error correction to fail in the coming Steane syndrome extraction with much higher probability then if they were uncorrelated.
We say an ancilla is \emph{qualified} if it is free of correlated errors.
If both $Z$ and $X$ ancillas are qualified, then single-shot Steane syndrome extraction is sufficient to implement a logical gate operation or a logical state measurement~\cite{brun2015teleportation}.

\section{Ancilla preparation using classical codes with noiseless quantum circuits}\label{sec:noiseless_prep}
Preparing these states with high quality and high throughput is crucial. Here is a naive preparation procedure of these ancillas, through verification using an additional identically prepared ancilla state. Errors are copied from the initial ancilla to an auxiliary ancilla, which are then bitwise measured. If errors are detected, one discards both ancillas. Note that the auxiliary ancillas may also contain correlated errors after their preparation, and need further checks by more auxiliary ancillas. This procedure ultimately would will involve tremendous numbers of ancilla states.

For example, the verification circuit for a qualified ancilla $|0\>_L$ of the $[[23,1,7]]$ Golay code is shown in Fig.~\ref{fig:naive}~\cite{paetznick_ben2011fault}.
In general, to verify a $t-$error-correcting quantum code, one needs a distillation process using at least $t^2 + t$ identically prepared blocks
~\footnote{Overall, it requires more than $2+4+,\dots,+2t=t^2+t$ ancilla blocks for logical 0 state verification. But for general stabilizer ancillas, it requires at least $(t+1)^2$ ancilla blocks.}
, which can be very large for the codes of interest. Denote the overall rejection rate of naive verification by $R_{\text{naive}}(p)$, which is $O(p)$, since any single failure will lead to a rejection. The average yield rate of qualified ancillas is:
\beq\label{eq:naive_yield}
\text{Yield}=\frac{1 - R_{\text{naive}}(p)}{t^2 + t}.
\eeq

\begin{figure}[!ht]
\centering\includegraphics[width=85mm]{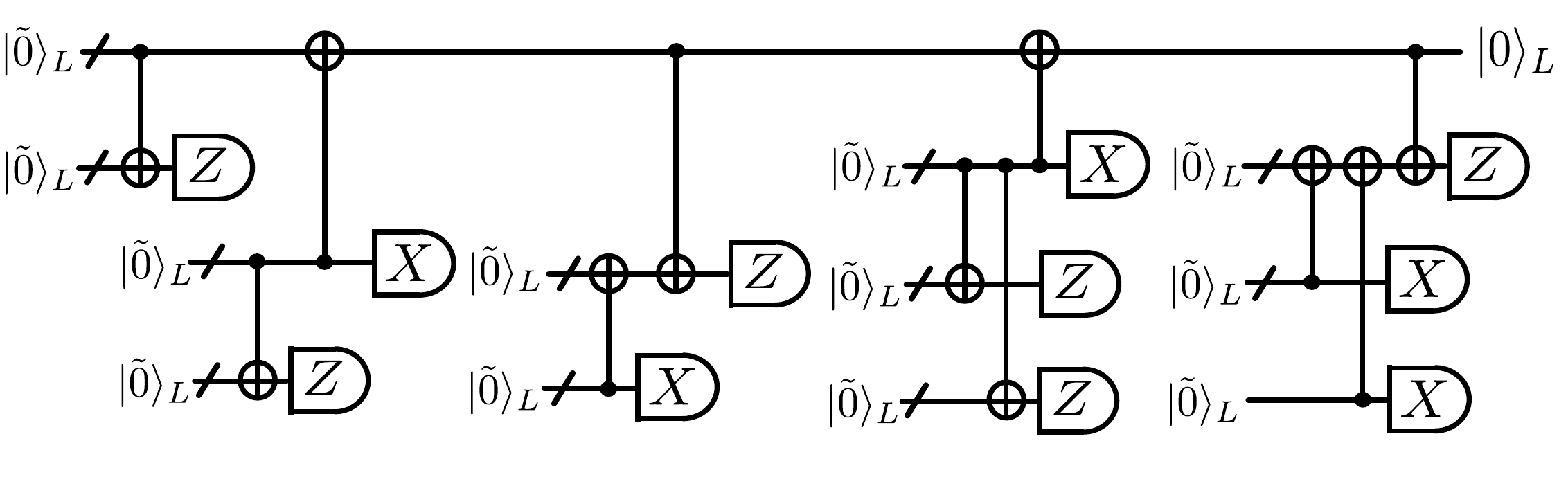}
\caption{\label{fig:naive} A naive circuit that produces a $|0\>_L$ state of the $[[23,1,7]]$ Golay code free of correlated errors. The CNOTs are transversal and measurements are bitwise. All $|\widetilde{0}\>_L$ are identically prepared. If any measurement suggests the existence of errors, all ancilla blocks involved are discarded and the procedure restarts.}
\end{figure}

In other contexts (e.g., entanglement purification~\cite{Bennett:1996:722,Bennett:1996:3824}) we may get around this formidable verification procedure by \emph{distillation}, which takes many of imperfect states and carries out a protocol to produce a smaller number of better states.
In Ref.~\cite{Ancilla_distillation_1}, we have shown that classical error-correcting codes can be applied to distill stabilizer states, such as $|0\>_L$ and $|+\>_L$, if the distillation circuit is perfect. When errors are present in blocks, we can correct them in most cases, rather than discarding the ancillas and starting over. Potentially, this scheme can greatly reduce the resource overhead.

In this section, we review distillation for CSS stabilizer states $|0\>_L^{\otimes k}$ and $|+\>_L^{\otimes k}$ using classical error-correcting codes with a \emph{noiseless} quantum circuit~\cite{Ancilla_distillation_1}.
We will formalize the protocol for general $[[n,k,d]]$ CSS codes. Generalization of the protocol to handle all necessary ancilla states for FTQC can be found in Appendix.~\ref{sec:logical ancilla}.
The problem of noisy distillation circuits will be discussed in detail in the next section.

We consider the ancilla states $|0\>_L^{\otimes k}$ and $|+\>_L^{\otimes k}$ of an $[[n,k,d]]$ CSS code $\mathcal{Q}$ constructed from an $[n,k_X,d_X]$ code $\cC_X$ and an $[n,k_Z,d_Z]$ code $\cC_Z$ such that $\cC_X^{\perp}\subset \cC_Z$.
The distillation protocol in~\cite{Ancilla_distillation_1} consists of two rounds of distillation, where two classical code $\cC_{c_1}$ and $\cC_{c_2}$ are used to correct $X$ and $Z$ errors, respectively.
Consider an $[n_{c_1},k_{c_1},d_{c_1}]$ classical code $\cC_{c_1}$ and an $[n_{c_2},k_{c_2},d_{c_2}]$ classical codes $\cC_{c_2}$ that can correct $t_{c_1}=\lfloor \frac{d_{c_1}-1}{2}\rfloor$ and $t_{c_2}=\lfloor \frac{d_{c_2}-1}{2}\rfloor$ errors, respectively.
Let $\textsf{H}_{c_1} = [\textsf{I}_{r_{c_1}} |  \ \textsf{A}_{c_1} ]$ be the parity-check matrix of $\cC_{c_1}$ in systematic form, where $r_{c_1} = n_{c_1}-k_{c_1}$, $\textsf{I}_{r_{c_1}}$ is the $r_{c_1}\times r_{c_1}$ identity, and $\textsf{A}_{c_1}$ is  an $r_{c_1}\times k_{c_1}$ binary matrix.

Suppose we have $n_{c_1}n_{c_2}$ independently prepared ancillas in the  $|0\>_L^{\otimes k}$ state (for $|+\>^{\otimes k}_L$ state, the situation is similar).
This can be done by an encoding circuit using CNOTs and physical state preparations in $|0\>$ and $|+\>$.
Note that these process are noisy, so after encoding, for each block, there exists a residual Pauli error $E$ of the form $X^eZ^f$. We then divide them into $n_{c_2}$ groups of $n_{c_1}$ blocks. For each group of $n_{c_1}$ noisy ancilla blocks,
let $g_{Z\,|\,i}^{(j)}$ be the eigenvalue of the stabilizer generators $G_{Z\,|\,i}$ of the $j$th block, and $\ell_{\bar{Z}_i}^{(j)}$ be the eigenvalue of $\bar{Z}_i$ for the $j$th block. Then the syndromes of stabilizers for the $n_{c_1}$ ancilla blocks of $\mathcal{Q}$ are:
\begin{equation}
\begin{array}{ccccc}
  g_{Z\,|\,1}^{(1)}       & g_{Z\,|\,2}^{(1)}      & \cdots  & g_{Z\,|\,{r_Z-1}}^{(1)}        &  g_{Z\,|\,{r_Z}}^{(1)}  \\ [3pt]
  g_{Z\,|\,1}^{(2)}       & g_{Z\,|\,2}^{(2)}    & \cdots  &   g_{Z\,|\,{r_Z-1}}^{(2)}      & g_{Z\,|\,{r_Z}}^{(2)} \\ [3pt]
    \vdots        & \vdots         & \ddots  &   \vdots               & \vdots \\ [3pt]
 g_{Z\,|\,1}^{(n_{c_1}-1)} & g_{Z\,|\,2}^{(n_{c_1}-1)}& \cdots  & g_{Z\,|\,{r_Z-1}}^{(n_{c_1}-1)}  & g_{Z\,|\,_{r_Z}}^{(n_{c_1}-1)} \\ [3pt]
  g_{Z\,|\,1}^{(n_{c_1})}   & g_{Z\,|\,2}^{(n_{c_1})}  & \cdots  & g_{Z\,|\,{r_Z-1}}^{(n_{c_1})}    & g_{Z\,|\,{r_Z}}^{(n_{c_1})},
\end{array}
\end{equation}
and
\begin{equation}
\begin{array}{ccccc}
  \ell_{\bar{Z}_1}^{(1)}       & \ell_{\bar{Z}_2}^{(1)}      & \cdots  & \ell_{\bar{Z}_{k-1}}^{(1)}        &  \ell_{\bar{Z}_{k}}^{(1)}  \\ [3pt]
  \ell_{\bar{Z}_1}^{(2)}       & \ell_{\bar{Z}_2}^{(2)}    & \cdots  &   \ell_{\bar{Z}_{k-1}}^{(2)}      & \ell_{\bar{Z}_{k}}^{(2)} \\ [3pt]
    \vdots        & \vdots         & \ddots  &   \vdots               & \vdots \\ [3pt]
 \ell_{\bar{Z}_1}^{(n_{c_1}-1)} & \ell_{\bar{Z}_2}^{(n_{c_1}-1)}& \cdots  & \ell_{\bar{Z}_{k-1}}^{(n_{c_1}-1)}  & \ell_{\bar{Z}_{k}}^{(n_{c_1}-1)} \\ [3pt]
  \ell_{\bar{Z}_1}^{(n_{c_1})}   & \ell_{\bar{Z}_2}^{(n_{c_1})}  & \cdots  & \ell_{\bar{Z}_{k-1}}^{(n_{c_1})}    & \ell_{\bar{Z}_{k}}^{(n_{c_1})}.
\end{array}
\end{equation}
Let $S_1$ be the set of $g_{Z\,|\,i}^{(j)}$ and $\ell_{Z_i}^{(j)}$, which form  a generalized syndrome array:
\beq
[\textsf{s}_1]_{j,i}=
\begin{cases}
&g^{(j)}_{Z\,|\,i}, \ \ \ \ \ \ \ \ \ \ \ \ \ \ \ \ \text{for} \ 1\leq i\leq r_Z; \\
&\ell^{(j)}_{\bar{Z}_{i-(n-n_Z)}}, \ \ \ \ \ \ \ \text{for} \ r_Z < i\leq r_Z+k.
\end{cases}
\eeq
{Here, the subscript 1 means that the eigenvalues of these stabilizers are processed at the first round of distillation.
Similarly, we define $S_2$ as the set of $X$ stabilizer generators  whose syndrome array $\textsf{s}_2$ would be estimated at the second round of distillation. Obviously, $S_1 \ \bigcup \ S_2$ is the complete set of stabilizer generators.}

If the preparation circuits are perfect, the entries of $\textsf{s}_1$ would all be 0s. However, if circuits are imperfect, some generalized syndromes bits will be flipped to 1s.
Let $[\textsf{s}_1]_{:,i}$ denote the $i$th column of $\textsf{s}_1$.
If one can correctly identify a row $j$ of $[\textsf{s}_1]$, $[\textsf{s}_1]_{j,:}$, i.e., the generalized syndromes for the $j$th block, then in principal quantum error correction can be carried out to correct the errors on the $j$th block. Observe that the entries of $[\textsf{s}_1]_{:,i}$ are independent binary random variables, since the ancilla blocks are prepared independently. In Ref.~\cite{Ancilla_distillation_1}, we introduce the idea of distillation by classical   codes, so that a subset of $[\textsf{s}_1]_{:,i}$ can be determined.

\begin{figure}[!ht]
\centering\includegraphics[width=40mm]{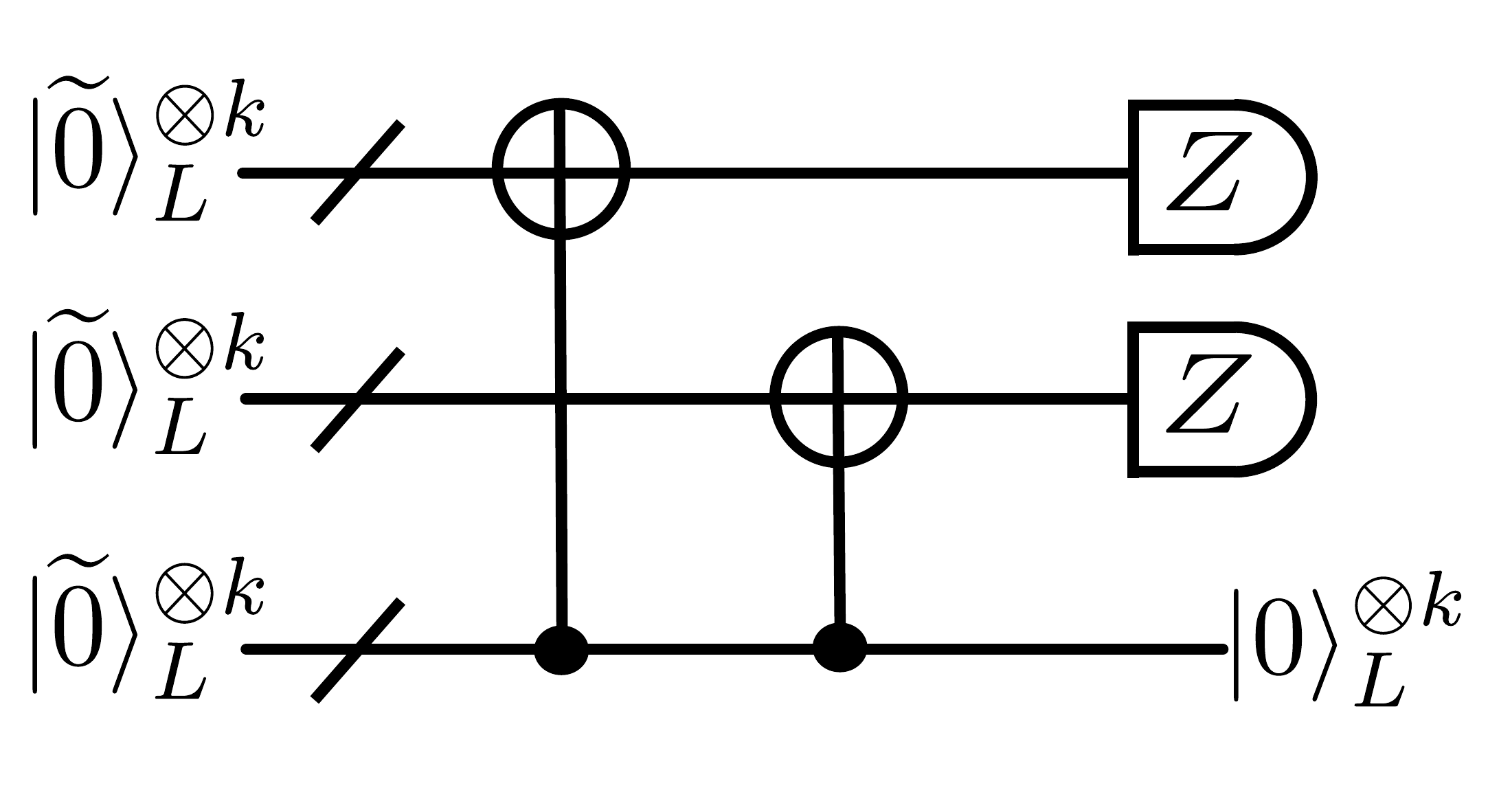}
\caption{\label{fig:0_x} The circuit for ancilla distillation by the classical $[3, 1, 3]$ repetition code to remove $X$ errors. The first two $|\widetilde{0}\>^{\otimes k}_L$ serve as parity-check ancillas.}
\end{figure}

Choose the first $r_{c_1}$ of the ancillas of each group to
hold the classical parity checks, and do transversal CNOTs from the remaining $k_{c_1}$ ancillas onto each of the parity-check ancillas according to the pattern of 1s in the rows of $\textsf{A}_{c_1}$. If $[\textsf{A}_{c_1}]_{i,j}=1$, we apply a transversal CNOT from the $(r_{c_1}+j)$th ancilla to the $i$th ancilla.
Note that $X$ errors will not propagate from the $i$th block to $(r_{c_1}+j)$th block through the CNOTs applied. Then measure all the qubits on each of the first $r_{c_1}$ ancilla blocks in the $Z$ basis, which destroys the states of those blocks. As an example,
Fig.~\ref{fig:0_x} demonstrates the distillation circuit
using the $[3,1,3]$ repetition code. From the measurement outcomes $\nu^{(1)}, \cdots, \nu^{(r_{c_1})}\in \mathbb{Z}_2^n$, we obtain $r_{c_1}$ strings $\sigma^{(i)}\in \mathbb{Z}_2^{r_Z+k},\ i=1,\dots, r_{c_1}$, which can be shown to be
\beq
\textsf{H}_{c_1}\textsf{s}_1=
\left[
  \begin{array}{c}
    \nu^{(1)}\left(\textsf{H}_{Z}'\right)^T \\
    \vdots \\
    \nu^{(r_{c_1})}\left(\textsf{H}_{Z}'\right)^T \\
  \end{array}
\right]\equiv
\left[
  \begin{array}{c}
    \sigma^{(1)} \\
    \vdots \\
    \sigma^{(r_{c_1})}\\
  \end{array}
\right]\equiv \sigma,
\eeq
where
\beq\label{eq:general_HZ}
\textsf{H}_Z'=\left[\begin{array}{c}
                             \textsf{H}_Z \\
                             \hline
                             \textsf{L}_Z\rule{0pt}{2.6ex}
                           \end{array}
\right]\equiv\left[\begin{array}{c}
                             \textsf{H}_Z\\
                             \hline
                             [\bar{z}_1^T\cdots
                             \bar{z}_k^T]^T \rule{0pt}{2.6ex}
                           \end{array}\right],
\eeq
and $\bar{z}_j$ is the binary representation of $\bar{Z}_j$.
In other words, the $i$th column of $\sigma$ matrix has the generalized syndrome bits of $[\textsf{s}_1]_{:,i}$.
Then we can use a decoder for $\cC_{c_1}$ to estimate the values of $[\textsf{s}_1]_{:,i}$ for all $i$ and thus obtain the estimated generalized syndrome array $[\widetilde{\textsf{s}}_1]$. After that, for code block $r_{c_1}+1 \leq j \leq n_{c_1}$, one can use  $\{\tilde{g}_{Z\,|\,i}^{(j)}\}$ and a classical decoder for $\cC_Z$ to estimate the $X$ errors. Suppose the decoder of $\mathcal{C}_Z$ always generates an estimate $\tilde{e}$ of an error vector $e$ such that $\textsf{H}_Z(e+\tilde{e})=0$, and assume the decoded $X$ error on block $j$ is $E_X^{(j)}$. Then for each logical operator $\bar{Z}_i^{(j)}$, $i=1,\dots,n$, $j=r_{c_1}+1,\dots, n_{c_1}$, follow the rules:
\begin{enumerate}
  \item If $[E^{(j)}_X, \bar{Z}_i^{(j)}] = 0$ and $\ell_{\bar{Z}_i}^{(j)} = 0$, then do nothing;
  \item If $\{E^{(j)}_X, \bar{Z}_i^{(j)}\} = 0$ and $\ell_{\bar{Z}_i}^{(j)} = 0$, then replace $E_X^{(j)}$ by $E_X^{(j)} \bar{X}_{i}^{(j)} $;
  \item If $[E^{(j)}_X, \bar{Z}_i^{(j)}] = 0$ and $\ell_{\bar{Z}_i}^{(j)} = 1$, then replace $E_X^{({j})}$ by $E_X^{(j)} \bar{X}_{i}^{(j)} $;
  \item If $\{E^{(j)}_X, \bar{Z}_i^{(j)}\} = 0$ and $\ell_{\bar{Z}_i}^{(j)} = 1$, then do nothing.
\end{enumerate}
The updated $E_X^{(j)}$ are the decoded $X$ errors, and one can correct them on the output block or just keep track of them. Alternatively, for a symmetric
CSS code, we may decode $\cC_Z^\perp$ directly from $[\widetilde{\textsf{s}}_1]_{j,:}$ to get
$E^{(j)}_X$ directly. However it may be very difficult to find an efficient decoder for $\cC_Z^\perp$ in practice.


Suppose the gate failure rate in the {encoding} circuit of $|0\>_L^{\otimes k}$ is $p$.
A gate failure may propagate through in the whole circuit so that, after encoding, the ancilla block contains some $X$ and $Z$ errors of weight up to the code length with probability $O(p)$. After the distillation process removes $X$ errors, the remaining $k_{c_1}$ ancillas from each group will have lower rates of $X$ errors than they started
with. The probability that $X$ errors of arbitrary weight remain on each block drops
from $O(p)$ to $O(cp^{t_{c_1}+1})$, where $c$ is a constant that depends
on the details of $\cC_{c_1}$ and its decoder. This is because the estimated generalized syndromes are wrong with probability $O(p^{t_c+1})$. But correlated $Z$ errors on the target qubits can propagate
via the CNOTs back onto the these $k_{c_1}$ blocks. The rate of $Z$ errors on the remaining $k_{c_1}$ ancillas will increase to  $O((M + 1)p)$, where $M$ is the maximum number of parity checks each control qubit is included in (or equivalently, the maximum number of 1s in each column of the matrix $\textsf{A}_{c_1}$). If $p$ is not too large, the rate of
$Z$ errors has grown roughly by a constant factor $M$, while the rate of $X$ errors has been substantially reduced. From the original large number of ancillas, we retain a fraction $k_{c_1}/n_{c_1}$, which have much lower $X$ error
rates and somewhat higher $Z$ error rates.

\begin{figure}[!ht]
\centering\includegraphics[width=60mm]{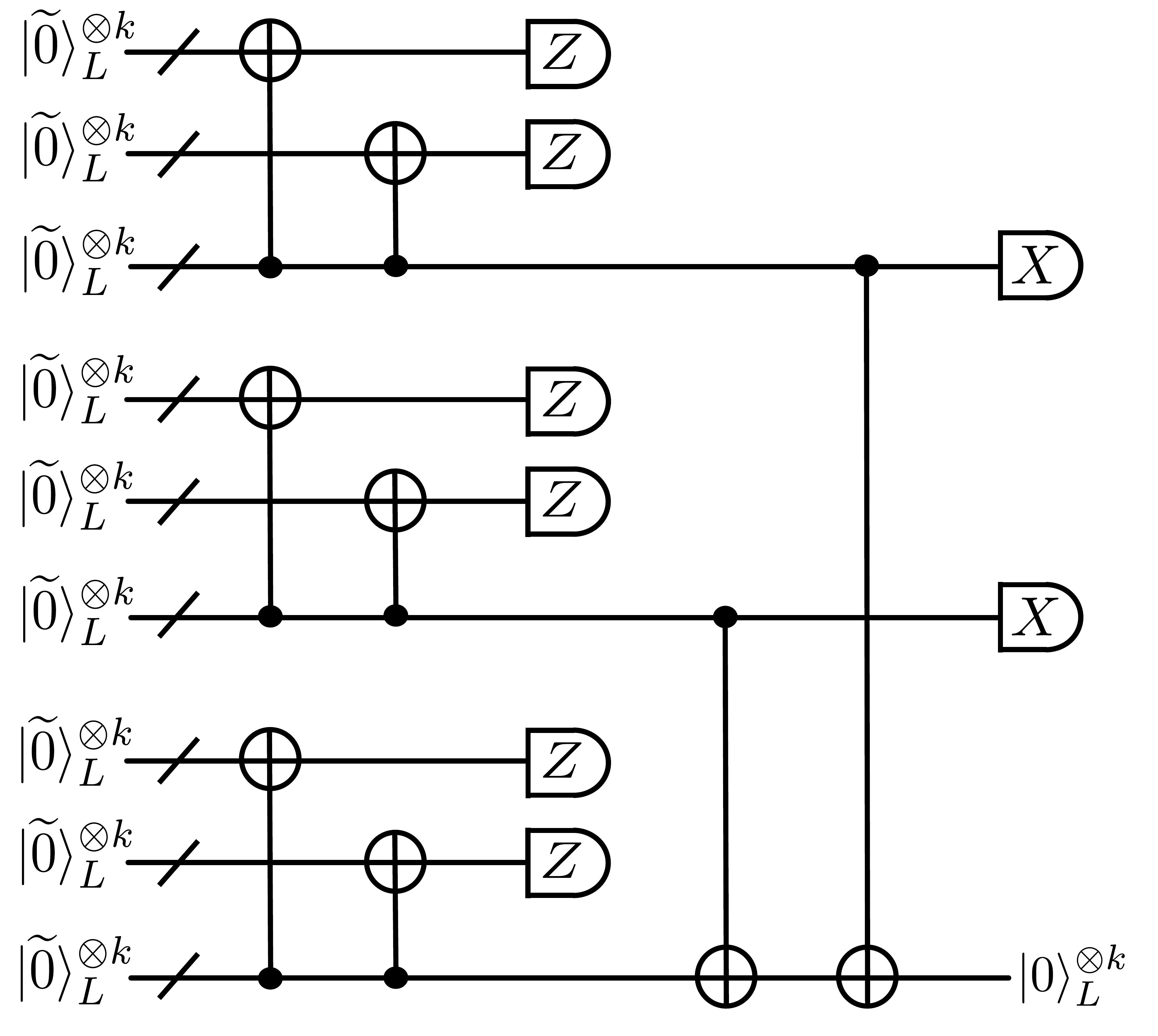}
\caption{\label{fig:0_x_z} The circuit for ancilla distillation by the classical $[3, 1, 3]$ repetition code to produce ancillas free of (low-weight) $X$ and $Z$ errors. }
\end{figure}
Each group will output $k_{c_1}$ blocks after the first round of distillation. To remove the remaining $Z$ errors, we divide the $k_{c_1}n_{c_2}$ output blocks into $k_{c_1}$ groups of $n_{c_2}$ blocks. Note that here it is important that
ancillas that were grouped together in the first round
\emph{not} be grouped together in the second round, because their $Z$ errors are now correlated between different blocks in the same group. This randomization process will ensure that the flipping of the generalized syndromes for different blocks in each group are independent so that classical error correction can work properly. Similarly, we can use the classical code $\cC_{c_2}$
to remove $Z$ errors. The process is similar to the first found, each qubit of the first $r_{c_2}=n_{c_2}-k_{c_2}$ ancillas is measured in the $X$ basis to estimate $\{\tilde{g}_{X\,|\, i}^{(j)}\}$ by classical decoding of $\cC_{c_2}$. (For $|0\>_L$, the eigenvalues of $\bar{X}_i$ do not need to be evaluated.) Since the error rate for $X$ is now higher, one may choose $\cC_{c_2}$ with a stronger error correction ability than $\cC_{c_1}$. An example of the whole distillation procedure using the $[3,1,3]$ repetition code is shown in Fig.~\ref{fig:0_x_z}. At the end, the overall yield rate is $\frac{k_{c_1}k_{c_2}}{n_{c_1}n_{c_2}}$.

\section{Fault-tolerant Ancilla State Preparation using classical codes}\label{sec:FTpreparation}
In the FTQC, CNOTs and measurements in a distillation circuit can be noisy, which could destroy the whole  protocol, producing unqualified ancillas as outputs. In this section, we will first study the negative effects of noisy circuits on the output ancillas, and show that the previous distillation protocol fails in this situation. We will then modify the distillation protocol to be a fully fault-tolerant, and show that the output ancillas are qualified.

In general, a fault-tolerant ancilla preparation procedure has two stages: ancilla preparation by noisy circuits and ancilla distillation by noisy circuits.
Before we go on, we introduce the notation that will be used later. Let $\mathscr{G}$ be the set of failures in the whole procedure and $|\mathscr{G}|$ represents the number of failures.
Note that we only consider failures that are Pauli operators.
Hence any single failure $\mathrm{g}\in \mathscr{G}$ can be decomposed as a product of an $X$ part and a $Z$ part, up to a phase: $\mathrm{g}=\mathrm{g}_X\mathrm{g}_Z$.  We define the collection of $\mathrm{g}_X$'s $(\mathrm{g}_Z\text{'s})$ from $\mathscr{G}$ as $\mathscr{G}_X$ $(\mathscr{G}_Z)$. Obviously, $|\mathscr{G}_{X }|, |\mathscr{G}_{Z}|\leq |\mathscr{G}|$.

We can define a subset $\mathscr{G}(\step)\subseteq \mathscr{G}$, which contains failures that occur at the $\step$th step of the circuit. Let $\text{Q}_{\mathscr{G}(\step)}$ be the \emph{support} of $\mathscr{G(\step)}$. For example, for a single CNOT gate at the $\step$th step, if $\mathrm{g}\in\mathscr{G}(\step)$ acts nontrivially on the control (target) qubit, then $\text{Q}_{\mathscr{G}(\step)}$ contains the control (target) qubit; if $\mathrm{g}\in\mathscr{G}(\step)$ affects both the control and target qubits, then $\text{Q}_{\mathscr{G}(\step)}$ contains both qubits. Let $\text{Q}_{\mathscr{G}(\step)}^{(j)}$ denote the support of failures on block~$j$ at the $\step$th step.
\begin{figure}[!htp]
\centering\includegraphics[width=45mm]{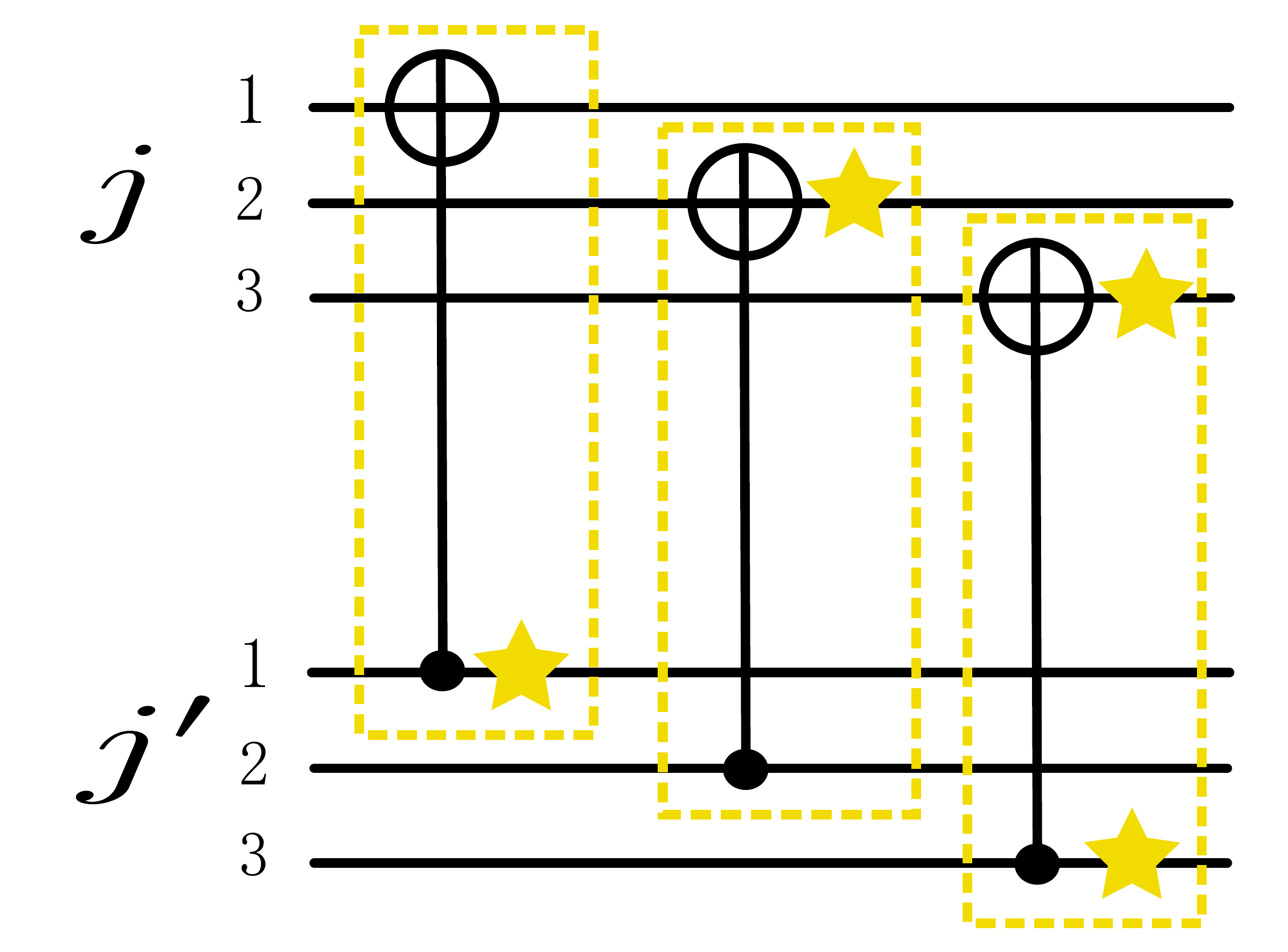}
\caption{\label{fig:supporting}(Color online) A single step in the preparation circuit involving two blocks $j$ and $j'$. The yellow stars represent the supporting qubits of failures that occur on the CNOT gates.
}
\end{figure}
Then $\text{Q}_{\mathscr{G}(\step)}=\bigcup_j \text{Q}_{\mathscr{G}(\step)}^{(j)}$.

We also define $\mathcal{I}\left(\text{Q}^{(j)}_{\mathscr{G}(\step)}\right)\subseteq\{1,\dots,n\}$, where $n$ is the code length, as the set of qubit indices of $\text{Q}^{(j)}_{\mathscr{G}(\step)}$ in block~$j$. For example, consider a single step in the preparation procedure as shown in Fig.~\ref{fig:supporting}. The corresponding $\text{Q}_{\mathscr{G}(\step)}^{(j)}$ and $\text{Q}_{\mathscr{G}(\step)}^{(j')}$ are qubits 2, 3 of block~$j$ and qubit 1, 3 of block $j'$.

\begin{figure}[!htp]
\centering\includegraphics[width=50mm]{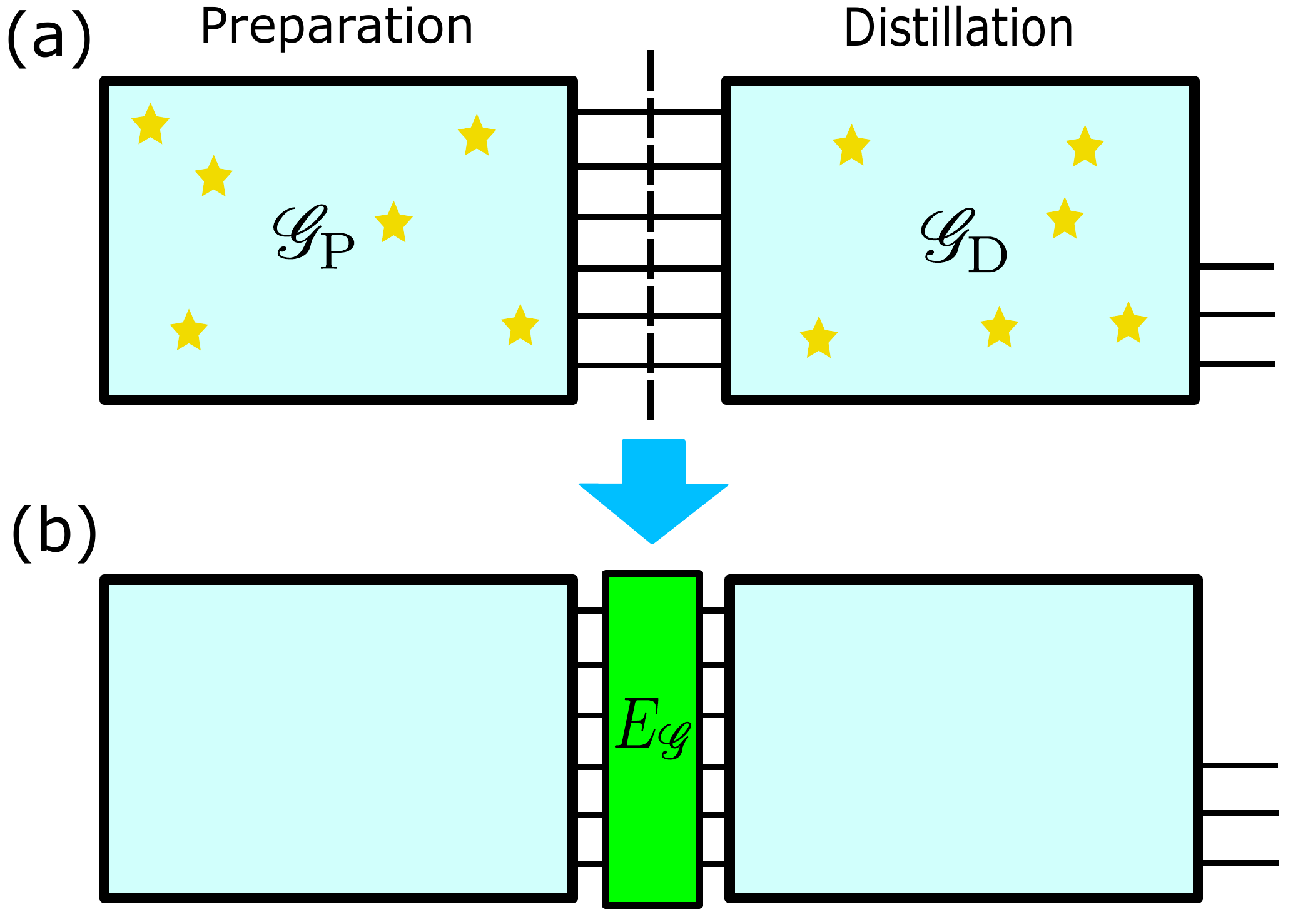}
\caption{\label{fig:equivalent}(Color online) (a) Failure set $\mathscr{G}$ (whose supports are represented by yellow stars) can be decomposed into failures occurring in the preparation stage $\mathscr{G}_{\rm P}$ and in the distillation stage $\mathscr{G}_{\rm D}$. (b) The effect of $\mathscr{G}$ is equivalent to perfect ancilla preparation and perfect distillation circuits with an effective failure process in between.
}
\end{figure}
Our strategy is to transform this fault-tolerant ancilla preparation procedure into one with three stages: 1) perfect ancilla preparation; 2) an effective failure process; and 3) perfect distillation circuits. Thus $\mathscr{G}$ can be decomposed as $\mathscr{G}=\mathscr{G}_{\rm P}\bigcup\mathscr{G}_{\rm D}$, where $\mathscr{G}_{\rm P}$ and $\mathscr{G}_{\rm D}$ are failure sets in the preparation and distillation stages, respectively. (See, for example, Fig.~\ref{fig:equivalent}.)
It is always possible for a noisy preparation circuit to be modeled as a perfect one followed by some failures,
and a distillation circuit can be modeled as a perfect one with preceding failures.
Let $$E_\mathscr{G}=\bigcup_j E_\mathscr{G}^{(j)}$$ be the  effective failures of $\mathscr{G}$ at the second stage,
where $E_\mathscr{G}^{(j)}$ is set of the effective failures  on block $j$.
One can also define $\text{Q}^{(j)}_{E_\mathscr{G}}$, $\mathcal{I}\left(\text{Q}^{(j)}_{E_\mathscr{G}}\right)$ in a similarly way. Define $S^{(j)}_{\mu}$ as the stabilizers of block whose eigenvalues need to be estimated in the $\mu$th round of distillation, for  $\mu=1,2$. Let $S^{(j)}_\mu|_{E_\mathscr{G}}\subseteq S^{(j)}_{\mu}$ be the subset of stabilizers  anticommuting with $E_\mathscr{G}$ on block $j$ (the set $E_\mathscr{G}$ is a Pauli operator)
and $\mathcal{I}\left(S^{(j)}_\mu|_{E_\mathscr{G}}\right)$ be the set of indices of stabilizers in $S^{(j)}_\mu|_{E_\mathscr{G}}$.

\subsection{A no-go theorem for the previous protocol with noisy circuits}\label{sec:nogo}
Observe that in a distillation circuit, the qubits that are measured in the $Z$ basis will never be the control qubits of CNOT gates.
If a noisy measurement is modeled as a perfect measurement with a preceding random Pauli $X$ failure, this $X$ failure will not propagate through the CNOTs.
Hence, this type of $X$ failure on the measurement can be directly absorbed into the effective failure preceding a perfect distillation.
Similarly for the $Z$ failures before the measurements in the $X$ basis.
Thus, measurement failures only have a very limited effect.
However, the effect of noisy CNOTs in the distillation circuit is  more complicated and will be discussed as follows.

\subsubsection{Perfect preparation circuit with noisy distillation circuits}
First, we consider the case that the input ancilla blocks are free of errors, while the distillation circuits are noisy. This seems absurd at a first glance, since then no distillation is needed. Nevertheless, it will give insight about the distribution of errors caused by the noisy CNOTs. For such a restricted set of failures $\mathscr{G}$,   the \emph{effective} support of $\mathscr{G}_{X}$ on block $j$ is defined as
\beq
\text{QE}_{\mathscr{G}_{X}}^{(j)}=\bigsqcup_{\step}
\text{Q}^{(j)}_{\mathscr{G}_{X}(\step)},
\eeq
where the operator $\bigsqcup$ takes the union of all $\text{Q}^{(j)}_{\mathscr{G}_{X}(\step)}$ except those qubits affected by an even number of different elements in $\mathscr{G}_{X}(\step)$.
For example, consider two CNOT gates with the same control qubit.
Suppose  that each of the two CNOTs suffers an $X$ failure on the control qubit, and $\mathscr{G}_X$ contains only these two failures.
Then $\text{QE}_{\mathscr{G}_{X}}=\emptyset$ since the two $X$ failures will cancel each other out.
$\text{QE}_{\mathscr{G}_{X}}^{(j)}$ is similarly defined.

Recall that the whole array of generalized syndromes before the $\mu$th round of distillation can be listed as follows:
\begin{equation}
{\small
\begin{array}{ccccc}
  [\textsf{s}_\mu]_{1,1}       &[\textsf{s}_\mu]_{1,2}     & \cdots  & [\textsf{s}_\mu]_{1,|S_\mu|-1}        &  [\textsf{s}_\mu]_{1,|S_\mu|}  \\ [3pt]
  [\textsf{s}_\mu]_{2,1}       & [\textsf{s}_\mu]_{2,2}    & \cdots  &   [\textsf{s}_\mu]_{2,|S_\mu|-1}      & [\textsf{s}_\mu]_{2,|S_\mu|} \\ [3pt]
    \vdots        & \vdots         & \ddots  &   \vdots               & \vdots \\ [3pt]
  [\textsf{s}_\mu]_{n_{c_\mu}-1,1}       & [\textsf{s}_\mu]_{n_{c_\mu}-1,2}    & \cdots  &   [\textsf{s}_\mu]_{n_{c_\mu}-1,|S_\mu|-1}      & [\textsf{s}_\mu]_{n_{c_\mu}-1,|S_\mu|} \\ [3pt]
   [\textsf{s}_\mu]_{n_{c_\mu},1}       & [\textsf{s}_\mu]_{n_{c_\mu},2}    & \cdots  &   [\textsf{s}_\mu]_{n_{c_\mu},|S_\mu|-1}      & [\textsf{s}_\mu]_{n_{c_\mu},|S_\mu|}
\end{array}
}
\end{equation}
where $\mu=1,2$. With the existence of CNOT failures in the distillation circuits, the generalized syndromes in this array can be correlated between different rows. As an illustration, we study the following example first.

\begin{figure}[!htp]
\centering\includegraphics[width=85mm]{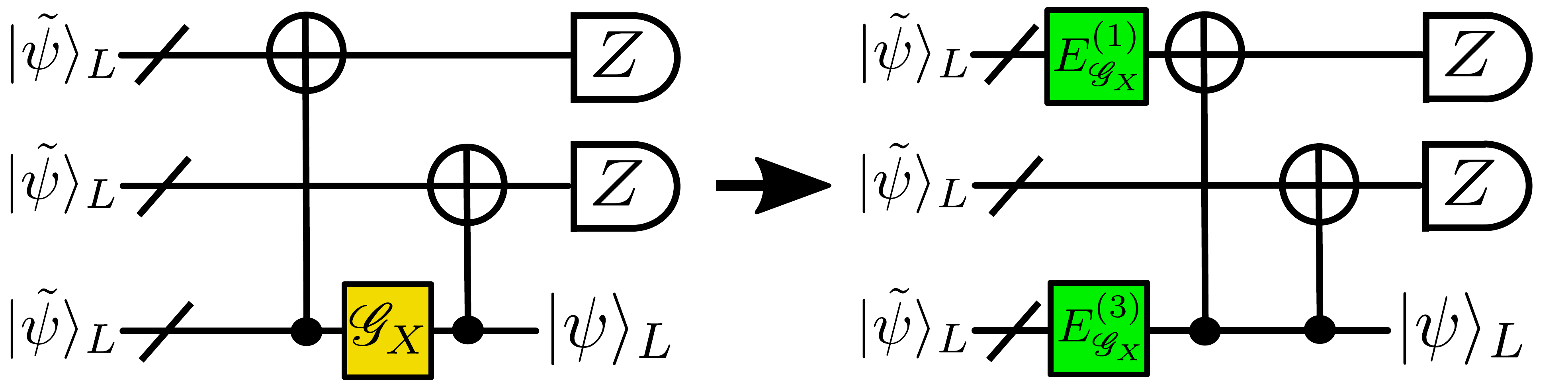}
\caption{\label{fig:distillation_error}An example of failures $\mathscr{G}_X$  associated with the CNOT gates on the third block.
It is equivalent to have effective failures on the first and third blocks before a perfect distillation circuit.
}
\end{figure}
\begin{example}\label{ex:benign}
Consider the distillation circuit that uses the $[3,1,3]$ code to remove $X$ failures. Let  $\mathscr{G}$ be a failure set on the third block  due to the transversal CNOTs from block~3 to block~1 as shown in the left circuit of Fig.~\ref{fig:distillation_error}.
Since only $X$ failures are supposed to be removed at this round, we only need to consider $\mathscr{G}_X$ and its support $\text{QE}_{\mathscr{G}_X}^{(3)}$. Apparently, $\left|\text{QE}_{\mathscr{G}_X}^{(3)}\right|=\left|\mathscr{G}_X\right|$, since $\mathscr{G}_X$ only has failures in one time step.
The circuit is equivalent to a perfect distillation circuit with preceding failures  $ E_{\mathscr{G}_X}^{(1)}\bigcup E_{\mathscr{G}_X}^{(3)}$
on the first and third blocks. Since the circuit consists only of transversal gates, we have
\beq
\mathcal{I}\left({\text{Q}_{E_{\mathscr{G}_X}}^{(1)}}\right)=\mathcal{I}
\left(\text{Q}_{E_{\mathscr{G}_X}}^{(3)}\right)
=\mathcal{I}\left(\text{QE}_{\mathscr{G}_X}^{(3)}\right).
\eeq
Thus, $\mathcal{I}\left(S^{(3)}_\mu|_{E_\mathscr{G}}\right) =\mathcal{I}\left(S^{(1)}_\mu|_{E_\mathscr{G}}\right)
\equiv \mathcal{I}_{S_\mu\,|\,{E_{\mathscr{G}_X}}}$. For those syndrome columns with index $i\in\mathcal{I}_{S_\mu\,|\,{E_{\mathscr{G}_X}}}$, $[\textsf{s}_\mu]_{1,i}=[\textsf{s}_\mu]_{3,i}=1$. Since the   $[3,1,3]$ code can correct only one error, it will mistakenly decode syndrome column $i\in\mathcal{I}_{S_\mu\,|\,{E_{\mathscr{G}_X}}}$ and diagnose a flip in $[\textsf{s}_{\mu}]_{2,i}$, leaving $E_{\mathscr{G}_X}^{(3)}$ on block 3 uncorrected after distillation.

Since the failures occur independently on the qubits, we are left with an $X$ error of weight $w=\left|\text{QE}_{\mathscr{G}_X}^{(3)}\right|$  with probability $p^{|\mathscr{G}|}$.
By Def.~\ref{def:uncorrelatation}, this is an  uncorrelated error since $|\mathscr{G}|\geq
|\mathscr{G}_X|=\left|\text{QE}_{\mathscr{G}_X}^{(3)}\right|$, so this is fine. This fact suggests that certain types of CNOT failures should be harmless, and only those output qubits in the  support of the faulty CNOTs would potentially be affected.
\end{example}

According to this observation, we can show that there are a large number of gate failures that are harmless as in Example~\ref{ex:benign}. We have following result:
\begin{mylemma}\label{lemma:distillation_error}
Consider a single round of the distillation protocol and, without loss of generality, suppose $X$ errors are to be removed.
Assume all the input ancilla states are error free. Suppose a gate failure set $\mathscr{G}$ occurs during distillation.  If for any $j,j'$  such that $\text{QE}^{(j)}_{\mathscr{G}_X},\text{QE}^{(j')}_{\mathscr{G}_X}\neq \emptyset$ and
\beq\label{eq:benign_failure}
\mathcal{I}\left(\text{QE}_
{\mathscr{G}_X}^{(j)}\right)= \mathcal{I}\left(\text{QE}_
{\mathscr{G}_X}^{(j')}\right),
\eeq
then for each output block, it is either free of $X$ errors or contains an $X$ error of weight $w= \left|\mathcal{I}_{\text{Q}\,|\,\mathscr{G}_X}\right|$,
where $\mathcal{I}_{\text{Q}\,|\,\mathscr{G}_X}$ is the set $\mathcal{I}\left(\text{QE}_{\mathscr{G}_X}^{(j)}\right)$ for some $j$ such that $\text{QE}^{(j')}_{\mathscr{G}_X}\neq \emptyset$.
\end{mylemma}
\begin{proof}
Since nonempty $\mathcal{I}\left(\text{QE}^{(j)}_{\mathscr{G}_X}\right)
=\mathcal{I}\left(\text{Q}^{(j)}_{E_{\mathscr{G}_X}}\right)$ are all identical, according to the transversal property of the distillation circuit, when propagating the error back before the distillation circuit, the $\mathcal{I}\left({S^{(j)}_\mu\,|\,_{E_{\mathscr{G}_X}}}\right)$ are identical for such a $S^{(j)}_\mu\,|\,_{E_{\mathscr{G}_X}}\neq \emptyset$.
Let $\mathcal{I}_{S_\mu\,|\,{E_{\mathscr{G}_X}}}
=\mathcal{I}\left({S^{(j)}_\mu\,|\,_{E_{\mathscr{G}_X}}}\right)$ for such $j$. For any output block $r_{c_\mu}+1\leq j\leq n_{c_\mu}$,
the estimated syndrome is that either
$[\widetilde{\textsf{s}}_\mu]_{j,i}=[\textsf{s}_\mu]_{j,i},\forall i\in \mathcal{I}_{S_\mu\,|\,{E_{\mathscr{G}_X}}}$, or $[\widetilde{\textsf{s}}_\mu]_{j,i}\neq[\textsf{s}_\mu]_{j,i}$, $\forall i\in \mathcal{I}_{S_\mu\,|\,{E_{\mathscr{G}_X}}}$. That is, either all the syndromes related to $\text{Q}^{(j)}_{E_{\mathscr{G}_X}}$ are correctly estimated, or they all are correctly estimated,
Consequently, after quantum error correction, the qubits with index in $\mathcal{I}_{\text{Q}\,|\,\mathscr{G}_X}$ for each output block are either all free of error, or all left with $X$ errors of weight $\left|\mathcal{I}_{\text{Q}\,|\,\mathscr{G}_X}\right|$.
\end{proof}
Lemma~\ref{lemma:distillation_error} implies that if gate failures satisfying Eq.~(\ref{eq:benign_failure}) are the only possible type of failures, then the output blocks are qualified,
since the probability of weight $w\leq |\mathscr{G}|$ errors remaining in the output blocks will be $O(p^{|\mathscr{G}|})$. Such a set of gate failures is said to be \emph{benign}. (Note that Example~\ref{ex:benign} is a special case, in that only the third block satisfies $\text{QE}^{(3)}_{\mathscr{G}_X}\neq \emptyset$.)

\subsubsection{Noisy preparation and distillation circuits}
Now we  consider the general case where the preparation circuit is also noisy.
After encoding, errors of arbitrary weight can exist with probability $O(p)$, and the syndrome flips are independent among different blocks.
However, there are additional correlated syndrome flips across different blocks due to the faulty CNOTs in the distillation circuit, which may compromise the validity of syndrome estimation.
This is summarized in the following no-go theorem, which says that the previous distillation protocol cannot handle the situation where the distillation circuit is noisy.
\begin{thm}[No go]\label{thm:nogo}
Consider a single round of the distillation protocol by  an $[n_{c_\mu}, k_{c_\mu}, d_{c_\mu}]$ classical code $\cC_{c_\mu}$ with parity check matrix $\left[\sfI_{n_{c_\mu}}|\sfA_{c_{\mu}}\right]$. Denote the number of 1s in the $j$th column of $\textsf{A}_{c_\mu}$  by $m_j$ and let $M=\max\{m_j\}$. If the distillation circuit is noisy with failure rate $p$, then the probability of an error of arbitrary weight $w$ remaining on the output blocks will be $O(p^2)$ if $t_{c_\mu}\leq M$, and $O(p^{t_{c_\mu}-M+2})$ if $t_{c_\mu} > M$,
where $t_{c_\mu}=\lfloor\frac{d_{c_\mu}-1}{2}\rfloor$.
\end{thm}
\begin{proof}
Without loss of generality, we consider the round of distillation that is supposed to remove the $X$ errors.  Suppose we have a failure set $\mathscr{G}$, which contains a single CNOT failure in the distillation circuit involving block $j_1$. The support of $\mathscr{G}_X$ is $\text{QE}^{(j_1)}_{\mathscr{G}_X}$. After  the failure propagates to the beginning of the distillation circuit, it will introduce failures in a certain number of blocks. In the worst case, $M$ blocks will be affected. Denoted this set of $M$ blocks by $\mathcal{B}_{\mathscr{G}_X}$.
Since the distillation circuit gates are all transversal, $\mathscr{G}_X$ is benign and $\mathcal{I}\left({S^{(j)}_\mu|_{E_{\mathscr{G}_X}}}\right)$ are all identical to $\mathcal{I}\left({S^{(j_1)}_\mu|_{E_{\mathscr{G}_X}}}\right)$  for $j\in \mathcal{B}_{\mathscr{G}_X}$  (denote it by $\mathcal{I}_{S_\mu\,|\,{E_{\mathscr{G}_X}}}$). Overall, $[\textsf{s}_\mu]_{j,i}$ will be flipped to 1 for $j\in \mathcal{B}_{\mathscr{G}_X}$ and $i \in \mathcal{I}_{S_\mu\,|\,{E_{\mathscr{G}_X}}}$.

If $M\geq t_{c_\mu}$, then consider an additional single failure $\mathscr{R}$ that occurs in the previous stage on a block $j_2\notin\mathcal{B}_{\mathscr{G}_X}$ such that for some $i_1 \in \mathcal{I}_{S_\mu\,|\,{E_{\mathscr{G}_X}}}$, $[\textsf{s}_\mu]_{j_2,i_1}=1$,
and for all other $i\in \mathcal{I}_{S_\mu\,|\,{E_{\mathscr{G}_X}}}$, $[\textsf{s}_\mu]_{j_2,i_1}=0$. Thus, $[\widetilde{\textsf{s}}_\mu]_{,:i_1}\neq [\textsf{s}_\mu]_{,:i_1}$ because $[\textsf{s}_\mu]_{,:i_1}$ has more than $t_{c_\mu}$ 1s. But all the other columns of syndromes are still correctly estimated.
In the worst case, for a specific output block $r_{c_\mu}+1\leq  j_3\leq n_{c_\mu}$, syndromes for all stabilizers are correctly estimated except the $i_1$th one:
In general, the decoded error $\widetilde{E}_X^{(j_3)}$ based on $[\widetilde{\textsf{s}}_{\mu}]_{j_3,:}$ can be quite different from the actual $X$ error on the $j$th block, ${E}_X^{(j_3)}={E}_{\mathscr{G}_X}^{(j_3)}{E}_{\mathscr{R}_X}^{(j_3)}$. As a result, the residual error $\widetilde{E}_X^{(j_3)}E_X^{(j_3)}$ can have any weight.
Thus, there is a probability  $O(p^2)$ that an error of any weight $w$  remains in the output blocks, since $\mathscr{G}$ and $\mathscr{R}$ are independent single failures.

For $M < t_{c_\mu}$, if $t_{c_\mu}-M+1$ failures occur  in ${M-t_{c_\mu}+1}$ separated blocks outside $\mathcal{B}_{\mathscr{G}_X}$ in the previous stages, the situation is similar to the case where $M\geq t_{c_\mu}$, and this  occurs with probability $O(p^{t_{c_\mu}-M+2})$.
\end{proof}

\subsection{Fault-tolerant ancilla state distillation}\label{sec:ftdistillation}
In this subsection we show how calculating additional syndrome bits \cite{ashikhmin_chingyi2014robust,fujiwara2014ability,ashikhmin2016correction} can lead to the postselection of compatible error syndromes in the distillation protocol. As we will see, one can get round Theorem~\ref{thm:nogo} by such a postselection, so that the resulting  distillation protocol is fully fault-tolerant.

Consider the two groups generated by the stabilizer generators $S_1, S_2$ of the first and second rounds of distillation:
\beq
\mathcal{SC}_\mu=\<S_\mu\>, \quad \quad \mu=1,2.
\eeq
\begin{mydef}
For $\mu=1,2$,
the eigenvalues associated with \emph{all} the elements of  $\mathcal{S}_\mu$  is called the \emph{complete syndrome}  for group $\mathcal{SC}_\mu$,  and is denoted by $[\textsf{sc}_{\mu}]_{j,i}$  for the $i$th syndrome bit in the $j$th block.
\end{mydef}
Note that   $\textsf{sc}_\mu$ can be estimated from the bitwise measurements of the check blocks. Denote the estimated complete syndrome array by $\widetilde{\textsf{sc}}_\mu$.
As before, let $\mathcal{SC}_\mu^{(j)}|_{E_{\mathscr{G}}}$  be the subset of stabilizers in $\mathcal{SC}_\mu^{(j)}$ that anticommute  with $E_\mathscr{G}$.
Using the distillation protocol proposed in the previous section, one can estimate the eigenvalue of any element in $\mathcal{SC}_\mu$ (syndromes) with only additional classical computation and no additional physical operations.

Now consider three elements $\text{SC}_{\mu,1}\in S_\mu$, $\text{SC}_{\mu,2}\in S_\mu$ and $\text{SC}_{\mu,3}=\text{SC}_{\mu,1}\text{SC}_{\mu,2} \in \mathcal{SC}_\mu \backslash S_\mu$. Denote the corresponding estimated generalized syndrome bits on a specific output block $j$ by $[\widetilde{\textsf{sc}}]_{j,1}$, $[\widetilde{\textsf{sc}}]_{j,2}$ and $[\widetilde{\textsf{sc}}]_{j,3}$, which are supposed to satisfy $[\widetilde{\textsf{sc}}]_{j,3}=[\widetilde{\textsf{sc}}]_{j,1}
+[\widetilde{\textsf{sc}}]_{j,2}$.
In such case,  the estimated syndrome bits are said to be \emph{compatible}. In general, although the values of $[\widetilde{\textsf{sc}}]_{j,3}$ and $[\widetilde{\textsf{sc}}]_{j,1}
+[\widetilde{\textsf{sc}}]_{j,2}$ are correlated, they are not always the same when failures occur. If $[\widetilde{\textsf{sc}}]_{j,3}\neq[\widetilde{\textsf{sc}}]_{j,1}
+[\widetilde{\textsf{sc}}]_{j,2}$, it strongly implies that the estimated syndrome set is problematic and the corresponding ancilla block cannot be used. This method can be extended to check any number of products of elements in $S_{\mu}$.

We therefore propose a modified distillation protocol by estimating additional syndrome bits and checking  whether the estimated syndrome bits are compatible with each other. If they are not, we discard the distilled ancillas and try again.

In this paper, we use  a classical error-detecting code to define a suitable set of additional stabilizers in $\mathcal{SC}_\mu \backslash S_\mu$.

Consider an $[n_{d_\mu}, k_{d_\mu}=|S_\mu|]$ classical code $\cC_{d_\mu}$ with parity-check matrix $\textsf{H}_{d_\mu}=[\textsf{I}_{n_{d_\mu}-|S_\mu|} \ |\  \textsf{A}_{d_\mu}]$, where
$\textsf{A}_{d_\mu}$ is an $(n_{d_\mu}-|S_\mu|)\times |S_\mu|$ binary matrix. We define a set of additional stabilizers $S'_\mu$ where
\beq
S'_{\mu_j}=\prod_{i=1}^{|S_\mu|} S_{\mu_i}^{[\textsf{A}_{d_\mu}]_{j,i}}, \quad \quad j=1,\dots, n_{d_\mu}-|S_\mu|,
\eeq
where $S'_{\mu_j}$ is the $j$th element of $S'_{\mu}$ and $S_{\mu_i}$  us the $i$th element of $S_\mu$.
$S'_\mu$ and $S_\mu$ can be combined to form an extended set of stabilizer
$SE_{\mu}=S'_\mu \bigcup S_\mu$, whose eigenvalues need to be estimated. The corresponding extended syndrome array of $SE_\mu$ for all blocks can be defined as
\beq
\textsf{se}_\mu=[\textsf{s}'_\mu|\ \textsf{s}_\mu].
\eeq
After classical decoding, each output block $j$ has an estimated syndrome vector $[\widetilde{\textsf{se}}_\mu]_{j,:}$. By the definition of $S'_\mu$, $[\widetilde{\textsf{se}}_\mu]_{j,:}$ is compatible only if
\beq\label{eq:postselection}
\textsf{H}_{d_\mu}[\widetilde{\textsf{se}}_\mu]_{j,:}^T=0.
\eeq
If the block does not pass this test, its estimated syndrome bits are likely to be chaotic. Since for each $j$, the values of $[\widetilde{\textsf{se}}_\mu]_{j,i}$ are strongly correlated between different $i$, we can not do error correction using $\cC_{d_\mu}$ but instead simply discard the outputs. Such process is called \emph{postselection}:
we accept the output ancillas if Eq.~(\ref{eq:postselection}) holds, and  we do an error correction procedure based on $[\widetilde{\textsf{s}}_{\mu}]$ for the remaining blocks as before. This scheme   is fault-tolerant only if the high-weight errors remaining on the output blocks are greatly suppressed by postselection, and the ancillas that pass the postselection are qualified. In addition, we want that the rejection rate not to be too high. In the next subsection, we will address these issues in detail.

\subsection{Fault-tolerance of the protocol}
In this section, we study the conditions such that the proposed protocol in the previous subsection is fault-tolerant.
A glossary is provided in Appendix.~\ref{appendix:table} (Table~\ref{tab:notation}) to help the readers with the various symbols used in the text.

In the previous subsection, we use a classical error-detecting code to define additional syndrome bits and check whether they are compatible.
In the ideal case, if the complete syndrome set is estimated and checked for compatibility, the postselection will work better than the nonideal case, since it has the most error information.

\begin{mydef}
Suppose the complete syndrome set of block $j$ for $\mathcal{SC}_\mu^{(j)}$ is estimated ($\mu=1,2$).
For any $\text{SC}_{\mu, i_l}\in \<S_\mu\>$ with a decomposition  $\text{SC}_{\mu,i_1}\times \cdots \times  \text{SC}_{\mu,i_k}$, where $\text{SC}_{\mu,i_j}\in S_{\mu}$, if the estimated syndrome bits satisfy
\beq\label{eq:valid_syndrome_cond}
[\widetilde{\textsf{sc}}_\mu]_{j,i_l}= [\widetilde{\textsf{sc}}_\mu]_{j,i_1}+\cdots \ + [\widetilde{\textsf{sc}}_\mu]_{j,i_k},
\eeq
then the estimated complete syndrome set is said to be \emph{valid} for block $j$.
\end{mydef}

\begin{remark}
Even if the estimated complete syndrome set is valid, it is not necessarily correct; instead, it simply means that the estimated syndromes are compatible with each other and thus more trustworthy.
\end{remark}

\begin{mydef}
A postselection is said to be \emph{ideal} if those output blocks with valid estimated complete syndrome sets are kept and the rest are discarded in a round of distillation.
\end{mydef}

As we have seen in the proof of Lemma~\ref{lemma:distillation_error}, certain syndrome patterns caused by a benign failure set are harmless. We can define such syndrome patterns as follows.
\begin{mydef}
Consider the complete syndrome set for all the ancilla blocks. If the values of the syndrome bits for each block are identical, then this complete syndrome array is said to be \emph{good}. The subset of blocks containing these non-zero syndrome bits are denoted as $\mathcal{B}_{\text{f}}$.
\end{mydef}
An example of a good syndrome pattern is shown in Fig.~\ref{fig:good_pattern} for five blocks.
\begin{figure}[!htp]
\centering\includegraphics[width=50mm]{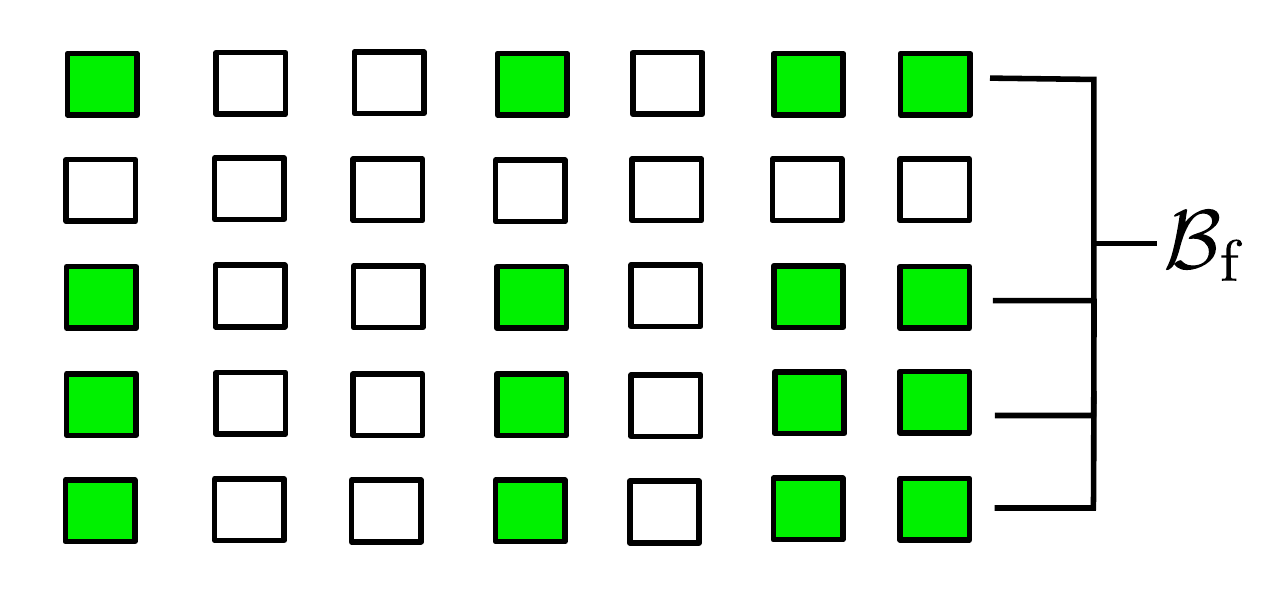}
\caption{\label{fig:good_pattern}(Color online) An illustration of a good syndrome pattern for syndrome array $\textsf{sc}_\mu$, with $|S_\mu|=3$.  The green boxes represent 1-value syndromes and white ones represent 0-value syndromes. }
\end{figure}
One can directly have the following lemma from the observation
 in Sec.~\ref{sec:nogo}:
\begin{mylemma}
A benign failure set causes a good syndrome pattern before distillation.
\end{mylemma}

Another  useful property of a good syndrome pattern is as follows.
\begin{mylemma}\label{lemma:good_pattern}
All the nonzero columns of the syndrome array of a good syndrome pattern are identical.
\end{mylemma}

Since a good syndrome pattern is harmless, we may expect that it will pass ideal postselection. Indeed, we have the following lemma:
\begin{mylemma}\label{lemma:goodpass}
If the complete syndrome pattern is good, then the estimated complete syndrome set is valid.
\end{mylemma}
\begin{proof}
If the number of blocks with syndrome bits flipped to 1 is no more than $t_c$, then the estimated complete syndrome sets are all correct, that is, $[\widetilde{\textsf{sc}}_\mu]_{j,i}=[\textsf{sc}_\mu]_{j,i}$ for all $j$ and $i$, and the estimated complete syndrome sets are automatically valid.
Otherwise, the estimated syndromes can either be wrong or correct, depending on the classical decoding procedure.

Consider two columns $i_1$ and $i_2$ and block $j\in\mathcal{B}_{\text{f}}$. 
Suppose $\text{SC}_{\mu,i_3}=\text{SC}_{\mu,i_1}\text{SC}_{\mu,i_2}\in \<S_\mu\>$.
If $[\textsf{sc}_\mu]_{j,i_1}=1$ and $[\textsf{sc}_\mu]_{j,i_2}=0$, then $[\textsf{sc}_\mu]_{j,i_3}=1$.  Then $[\widetilde{\textsf{sc}}_\mu]_{j,i_1}=[\widetilde{\textsf{sc}}_\mu]_{j,i_3}$ because $[\textsf{sc}]_{:,i_1}=[\textsf{sc}]_{:,i_3}$ when the syndrome pattern is good. Meanwhile, $[\textsf{sc}_\mu]_{:,i_2}$ must be all 0 according to Lemma~\ref{lemma:good_pattern}, and thus  $[\widetilde{\textsf{sc}}_\mu]_{j,i_2}=0$. Then one obtains
\begin{equation*}
[\widetilde{\textsf{sc}}_\mu]_{j,i_1}+[\widetilde{\textsf{sc}}_\mu]_{j,i_2}=[\widetilde{\textsf{sc}}_\mu]_{j,i_3}.
\end{equation*}
The above equation holds for arbitrary values of $[\textsf{sc}_\mu]_{j,i_1}$ and $[\textsf{sc}_\mu]_{j,i_2}$, . Similarly, we can see that
\begin{equation*}
[\widetilde{\textsf{sc}}_\mu]_{j,i_l}= [\widetilde{\textsf{sc}}_\mu]_{j,i_1}+\cdots \ + [\widetilde{\textsf{sc}}_\mu]_{j,i_k}
\end{equation*}
holds in general for $\text{SC}_{\mu,i_l}=\text{SC}_{\mu,i_1}\cdots \text{SC}_{\mu,i_k}\in \<S_\mu\>$.
\end{proof}

A natural question arises: can we mainly rely on such a postselection to reduce high weight errors? In other words, can we just use some classical codes of very small distance for $\cC_{c_1}$ and $\cC_{c_2}$, say distance three,  to roughly estimate the general syndrome bits, and then take error detection codes $\cC_{d_1}$ and $\cC_{d_2}$ of large distance to remove the blocks with incompatible syndrome bits? Then after quantum error correction, would the output ancillas be qualified?

To answer these questions, first consider a perfect distillation circuit  and check how well the postselection mechanism can perform. Can we suppress the error rate to order higher than $t_c+1$? The answer is negative because of the correlation among the estimated general syndrome bits after classical decoding. Here, we show the limitation of the postselection mechanism alone:
\begin{thm}\label{lemma:postselection_nogo}
For a single round of the distillation process, if the classical error-correcting code $\cC_{c_\mu}$ can correct $t_{c_\mu} < t-1$ errors, where the underlying quantum code can correct $t$ errors, then the probability of any error of weight $w > t$ on the output is at most $O(p^{t_{c_\mu}+1})$ even for ideal postselection.
\end{thm}
\begin{proof}
Without loss of generality, let us assume this round of distillation is to remove $X$ errors. Suppose the failure set $\mathscr{G}$ occurs in the preparation stage and the distillation circuit is perfect. Assume the support of $\mathscr{G}_X$ is over $t_{c_\mu}+1$ blocks. Denote this set of $t_{c_\mu}+1$ blocks by $\mathcal{B}_{\mathscr{G}_X}$.
Suppose that $\mathcal{I}\left(\text{Q}^{(j)}_{E_{\mathscr{G}_X}}\right)$ is the same for $j\in \mathcal{B}_{{\mathscr{G}_X}}$:
\beq
\mathcal{I}\left(\text{Q}^{(j)}_{E_{\mathscr{G}_X}}\right) \triangleq \mathcal{I}_{\text{Q}\,|\,E_{\mathscr{G}_X}},\quad\quad \quad j\in \mathcal{B}_{{\mathscr{G}_X}}.
\eeq
This means that each block in $\mathcal{B}_{\mathscr{G}_X}$ has exactly the same effective $X$ error $E^{(j)}_{\mathscr{G}_X}$ before distillation. Obviously, $\text{Q}^{(j)}_{E_{\mathscr{G}_X}}$ for $j\in \mathcal{B}_{\mathscr{G}_X}$ shares the same index set, denoted by $\mathcal{I}_{\text{Q}\,|\,E_{\mathscr{G}_X}}$. Thus, $\mathcal{I}\left(\mathcal{SC}_\mu^{(j)}|_{E_{\mathscr{G}_X}}\right)$ is also the same for all $j\in \mathcal{B}_{\mathscr{G}_X}$:
\beq
\mathcal{I}\left(\mathcal{SC}_\mu^{(j)}|_{E_{\mathscr{G}_X}}\right)\triangleq \mathcal{I}_{\mathcal{SC}_\mu\,|\,E_{\mathscr{G}_X}}, \quad\quad \quad j\in \mathcal{B}_{\mathscr{G}_X}.
\eeq
Obviously, $[\textsf{sc}_\mu]_{j,i}=1$ for $\forall i\in \mathcal{I}_{\mathcal{SC}_\mu\,|\, E_{\mathscr{G}_X}}$, $\forall j\in \mathcal{B}_{\mathscr{G}_X}$, otherwise $[\textsf{sc}_\mu]_{j,i}=0$. This syndrome pattern is good, and according to Lemma~\ref{lemma:goodpass} the estimated complete syndrome set will be valid for all blocks. Thus, any output block will pass ideal postselection. On the other hand, since $|\mathcal{B}_{\mathscr{G}_X}|=t_{c_\mu}+1>t_{c_\mu}$, $[\widetilde{\textsf{sc}}_\mu]_{:,i}\neq [\textsf{sc}_\mu]_{:,i}$ for $i\in\mathcal{I}_{\mathcal{SC}_\mu\,|\,E_{\mathscr{G}_X}}$. But one still has $[\widetilde{\textsf{sc}}_\mu]_{:,i_1}=[\widetilde{\textsf{sc}}_\mu]_{:,i_2}$ for any $i_1,i_2\in\mathcal{I}_{\mathcal{SC}_\mu\,|\,E_{\mathscr{G}_X}}$. Therefore, for any specific output block $r_{c_\mu}+1\leq j\leq n_{c_\mu}$ either  $[\widetilde{\textsf{sc}}_\mu]_{j,i}=[\textsf{sc}_\mu]_{j,i}, \forall i\in\mathcal{I}_{\mathcal{SC}_\mu|E_{\mathscr{G}_X}}$ or $[\widetilde{\textsf{sc}}_\mu]_{j,i}\neq[\textsf{sc}_\mu]_{j,i}, \forall i\in\mathcal{I}_{\mathcal{SC}_\mu|E_{\mathscr{G}_X}}$. In the worst case, $[\widetilde{\textsf{sc}}_\mu]_{j,i}\neq[\textsf{sc}_\mu]_{j,i}, \forall i\in\mathcal{I}_{\mathcal{SC}_\mu|E_{\mathscr{G}_X}}$, and   $X$ errors are left on qubits in $\mathcal{I}_{\text{Q}\,|\,E_{\mathscr{G}_X}}$ after quantum error correction. Note that the failure set $\mathscr{G}$ occurs with probability $O(p^{|\mathscr{G}|})\leq O(p^{t_{c_\mu}+1})$ and $\left|\mathcal{I}_{\text{Q}\,|\,E_{\mathscr{G}_X}}\right|$ can be arbitrary. So there always exists a $\mathscr{G}$ that leaves an $X$ error of weight $w=\left|\mathcal{I}_{\text{Q}\,|\,E_{\mathscr{G}_X}}\right|>t$  on output block $j$ with probability $O(p^{t_{c_\mu}+1})$, after ideal postselection and quantum error correction.
\end{proof}

Complementary to Lemma~\ref{lemma:goodpass}, one may infer that the pattern of a complete syndrome array is close to a good one if any of the output blocks passes ideal postselection. For simplicity, we define the concept of $\epsilon$-sparse for a quantum code $\mathcal{Q}$:
\begin{mydef} \label{def:epsilon-sparse}
Consider a  quantum code $\mathcal{Q}$.
Let $\mathbb{S}_\cap$ ($\mathbb{S}_{\slashed{\cap}}$) be the set of  complete syndrome pairs $(\textsf{sc1}_\mu, \textsf{sc2}_\mu)$, with $\textsf{sc1}_\mu\neq \textsf{sc2}_\mu$, such that they correspond to two errors $E_{1\mu},E_{2\mu}$ whose supports have an nonempty (empty) intersection. If $|\mathbb{S}_\cap|/|\mathbb{S}_{\slashed{\cap}}|<\epsilon$, $\cQ$ is called an \emph{$\epsilon$-sparse code}.
\end{mydef}
Quantum codes such as surface codes and quantum LDPC codes have low $\epsilon$ when the code length is large.
\begin{mylemma}\label{lemma:valid_to_good}
Suppose $\mathcal{Q}$ is $\epsilon$-sparse. If the estimated complete syndrome set is valid for a specific output block, then the corresponding syndrome array for all input blocks is, with probability {at least $1-O(\epsilon)$},  either a good syndrome pattern or close to a good syndrome pattern.
\end{mylemma}
\begin{proof}
Without loss of generality, assume output block $j$ has  a valid complete syndrome set and its $X$ errors are removed. The first observation is that if $[\widetilde{\textsf{sc}}_\mu]_{j,i}=[\textsf{sc}_\mu]_{j,i}$, then it is likely the whole column of syndrome bits  $[\widetilde{\textsf{sc}}_\mu]_{:,i}$  was correctly decoded by $\cC_{c_\mu}$ $([\widetilde{\textsf{sc}}_\mu]_{:,i}=[\textsf{sc}_\mu]_{:,i})$, and thus the number of 1s in $[\textsf{sc}_\mu]_{j,i}$ is fewer than $t_{c_\mu}$.
Similarly, if $[\widetilde{\textsf{sc}}_\mu]_{j,i}\neq[\textsf{sc}_\mu]_{j,i}$, then the column of syndrome bits $[\widetilde{\textsf{sc}}_\mu]_{:,i}$ is incorrectly estimated.
(Note that it is possible that $[\textsf{sc}_\mu]_{:,i}\neq [\widetilde{\textsf{sc}}_\mu]_{:,i}$, but the estimate on the $j$th block is still correct: $[\widetilde{\textsf{sc}}_\mu]_{j,i}=[\textsf{sc}_\mu]_{j,i}$. See Fig.~\ref{fig:pattern_decompose} for an example. We will address this case later.)

Since $[\widetilde{\textsf{sc}}_\mu]_{j,:}$ is valid, Eq.~(\ref{eq:valid_syndrome_cond}) must be satisfied, that is, for $\text{SC}_{\mu,i_3}=\text{SC}_{\mu,i_1}\text{SC}_{\mu,i_2}\in \<S_\mu\>$,
$[\widetilde{\textsf{sc}}_\mu]_{j,i_1}
+[\widetilde{\textsf{sc}}_\mu]_{j,i_2}
=[\widetilde{\textsf{sc}}_\mu]_{j,i_3}$ must hold.
If $[\widetilde{\textsf{sc}}_\mu]_{:,i_1}\neq [\textsf{sc}_\mu]_{:,i_1}$ and $[\widetilde{\textsf{sc}}_\mu]_{:,i_2}\neq [\textsf{sc}_\mu]_{:,i_2}$, then $[\widetilde{\textsf{sc}}_\mu]_{j,i_1}$
and $[\widetilde{\textsf{sc}}_\mu]_{j,i_2}$ are likely to be incorrectly estimated. It then is necessary that $[\widetilde{\textsf{sc}}_\mu]_{j,i_3}=[\textsf{sc}_\mu]_{j,i_3}$ . Consequently, it is likely that $[\widetilde{\textsf{sc}}_\mu]_{:,i_3} = [\textsf{sc}_\mu]_{:,i_3}$. Similarly,
if $[\widetilde{\textsf{sc}}_\mu]_{:,i_1}\neq [\textsf{sc}_\mu]_{:,i_1}$ but $[\widetilde{\textsf{sc}}_\mu]_{:,i_2}= [\textsf{sc}_\mu]_{:,i_2}$, then it is likely $[\widetilde{\textsf{sc}}_\mu]_{:,i_3} \neq [\textsf{sc}_\mu]_{:,i_3}$.

\begin{figure}[!htp]
\centering\includegraphics[width=50mm]{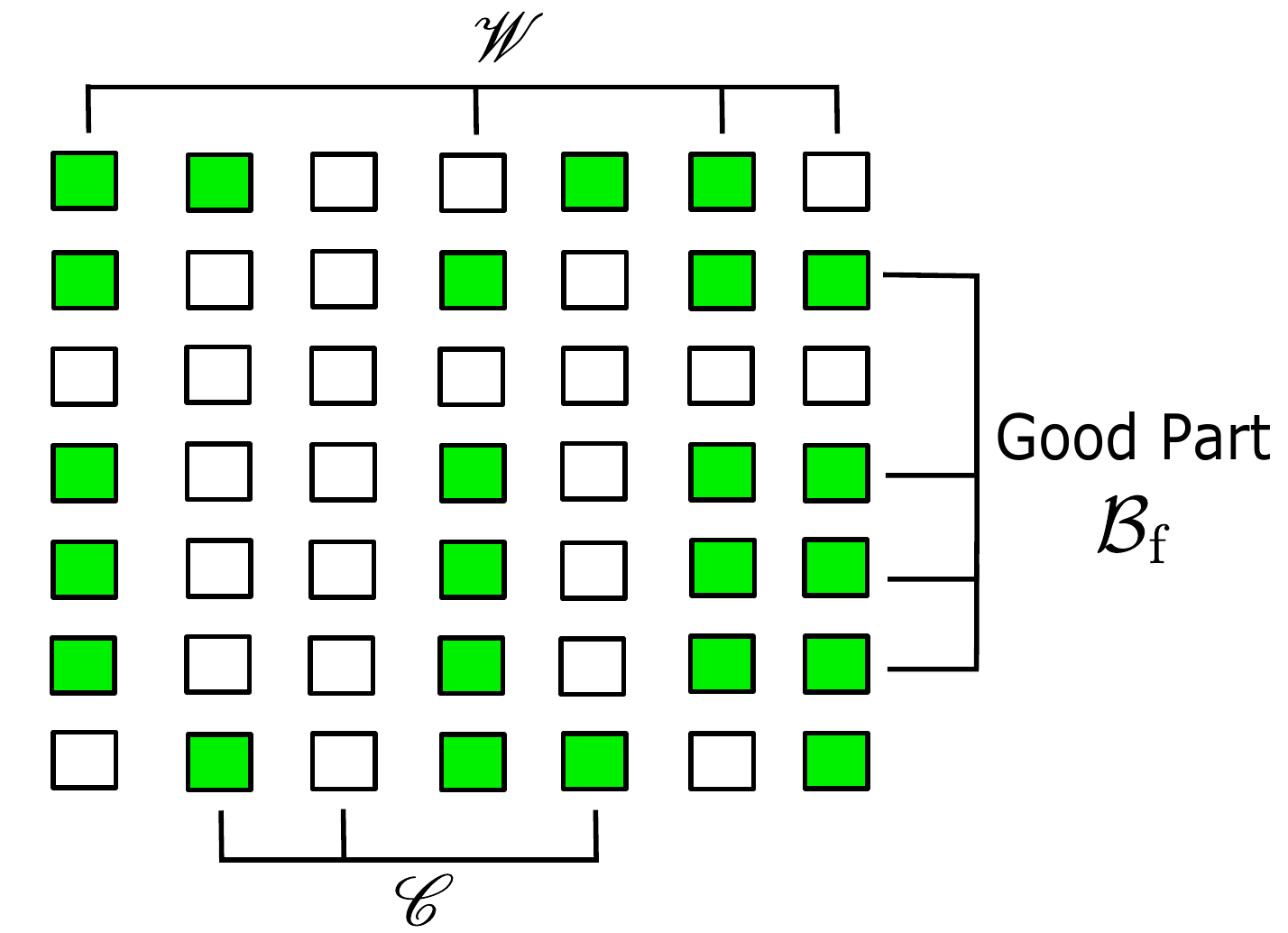}
\caption{\label{fig:good_pattern_approx}(Color online) An illustration of an approximately good syndrome pattern for syndrome array $\textsf{sc}_\mu$. Here, we assume that $\cC_{c_\mu}$ can correct two errors and $|S_\mu|=3$. The green boxes represent syndrome bit 1s  and the white ones represent syndrome bit 0s. The good part of this pattern, $\mathcal{B}_{\text{f}}$, represents the maximum subset of rows with identical syndromes.}
\end{figure}

\begin{figure*}[!htp]
\centering\includegraphics[width=130mm]{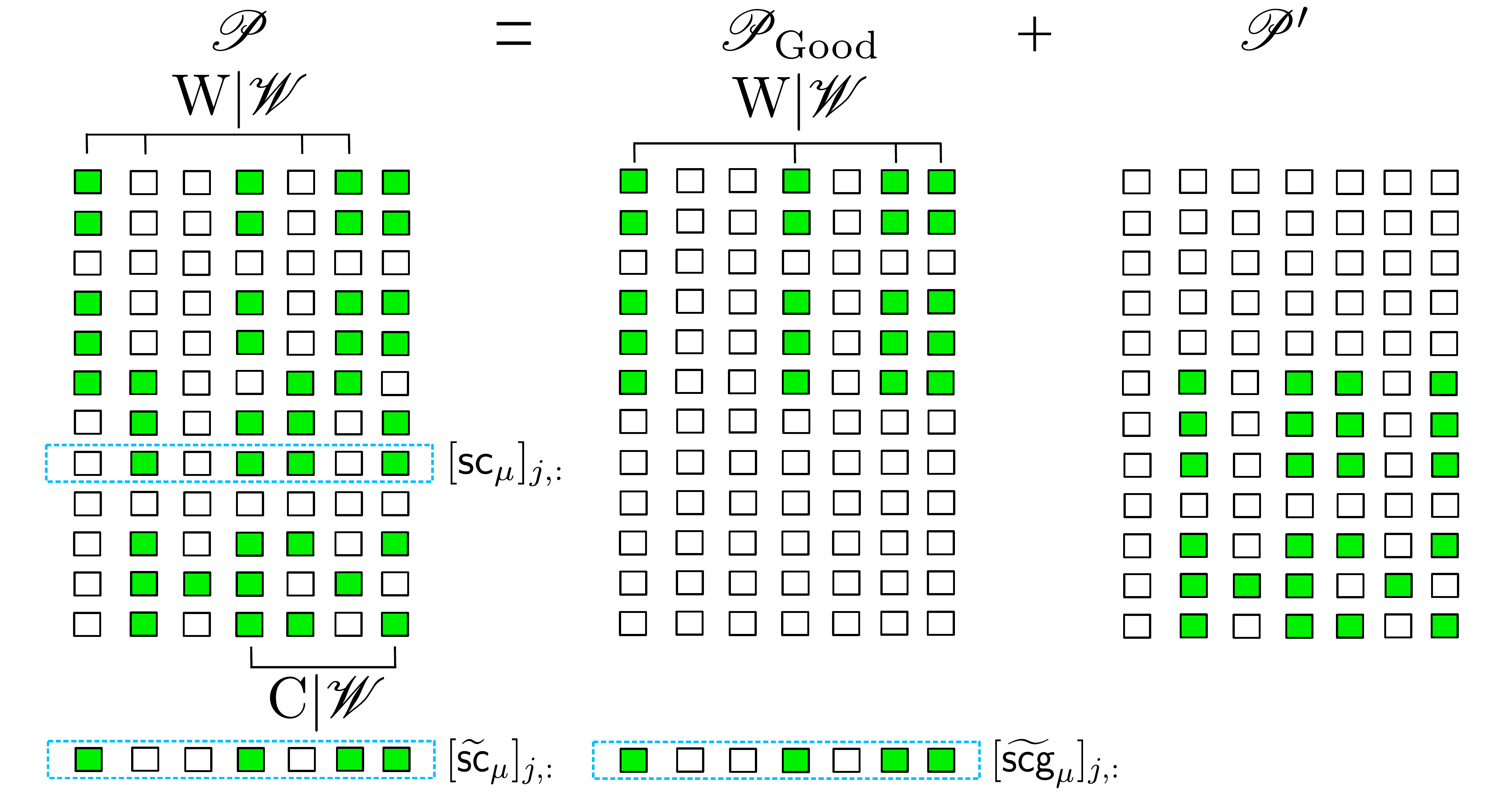}
\caption{\label{fig:pattern_decompose}(Color online) An illustration of a pattern $\mathscr{P}$ that is not close to good pattern, but has valid estimated complete syndrome set on output block $j$ after classical decoding.
Here, we assume that $\cC_{c_\mu}$ can correct two errors, and $|S_\mu|=3$. The green boxes represent 1-value syndromes and white ones represent 0-value syndromes.
The complete syndromes of output block $j$ are circled by a dashed line and its estimated syndrome $[\widetilde{\textsf{sc}}_\mu]_{j,:}$ is shown at the bottom of the left figure. $\mathscr{P}$ can always be decomposed into a good syndrome pattern $\mathscr{P}_{\text{Good}}$ and another pattern $\mathscr{P}'$, such that $\mathscr{P}_{\text{Good}}$ generates exactly the same $[\widetilde{\textsf{sc}}_\mu]_{j,:}$ after classical decoding. In this example, $\mathscr{P}_{\text{Good}}$ is constructed as follows: 1. each column $[\textsf{scg}]_{:,i}$ such that $[\widetilde{\textsf{sc}}]_{j,i}=0$ is set to be an all-zero syndrome column; 2. each column $[\textsf{scg}]_{:,i}$ such that $[\widetilde{\textsf{sc}}]_{j,i}=1$ is set to be the first column of $\mathscr{P}$ (which gives $[\widetilde{\textsf{sc}}]_{j,i}=1$ after classical decoding.).}.
\end{figure*}

We can generalize this observation to more than three columns of syndromes. First, denote the set of   columns such that $[\widetilde{\textsf{sc}}_\mu]_{:,i} \neq [\textsf{sc}_\mu]_{:,i}$ as $\mathscr{W}$. Take any $2x$ columns $[\textsf{sc}_\mu]_{:.i_1},\dots,[\textsf{sc}_\mu]_{:.i_{2x}}\in \mathscr{W}$, and  suppose $\text{SC}_{\mu,i_l}=\text{SC}_{\mu,i_1}\times \cdots \times \text{SC}_{\mu,i_{2x}}\in \<S_\mu\>$.
Since
\begin{equation*}
[\widetilde{\textsf{sc}}_\mu]_{j,i_1}
+,\dots,+[\widetilde{\textsf{sc}}_\mu]_{j,i_{2x}}
=[\widetilde{\textsf{sc}}_\mu]_{j,i_l},
\end{equation*}
$[\widetilde{\textsf{sc}}_\mu]_{j,i_l}=[\textsf{sc}_\mu]_{j,i_l}$ from the validity of output block $j$. Thus, the column of syndrome $[\textsf{sc}_\mu]_{:,i_l}$ is likely equal to
$[\widetilde{\textsf{sc}}_\mu]_{:,i_l}$ and contains less than $t_c$ $1$s. Consequently, the columns in $\mathscr{W}$ have a large amount of overlap in syndrome bits that are 1s.
Denote the set of columns such that $[\widetilde{\textsf{sc}}_\mu]_{:,i} = [\textsf{sc}_\mu]_{:,i}$ as $\mathscr{C}$.
Consider  any $[\textsf{sc}_\mu]_{:,i}\in \mathscr{W}$ and  $[\textsf{sc}_\mu]_{:,i^\prime}\in \mathscr{C}$.
Suppose $\text{SC}_i \text{SC}_{i'} = \text{SC}_{i''} \in \<S_\mu\>$. Since
$[\widetilde{\textsf{sc}}_\mu]_{j,i}
+[\widetilde{\textsf{sc}}_\mu]_{j,i^\prime}
=[\widetilde{\textsf{sc}}_\mu]_{j,i^{\prime\prime}}$,
it is likely that $[\textsf{sc}_\mu]_{:,i^{\prime\prime}}\in \mathscr{W}$. Since we know that $[\textsf{sc}_\mu]_{:,i}$ and $[\textsf{sc}_\mu]_{:,i^{\prime\prime}}$ have large overlaps, it also implies that for any $[\textsf{sc}_\mu]_{:,i^\prime}\in \mathscr{C}$, the number of 1s it contains is likely to be small, which is consistent with our previous observation.
Such syndrome array parttern is called an \emph{approximately good} syndrome pattern. Figure~\ref{fig:good_pattern_approx}  illustrates such a pattern. The overlaps  of different columns in $\mathscr{W}$ are called the good part of the pattern. So the syndrome pattern is likely to be close to a good pattern when any specific output block has a valid estimated complete syndrome set.

Now we consider those syndrome patterns whose columns in $\mathscr{W}$ do not   have a large amount of overlap, but still generate a valid estimated complete syndrome set for output block $j$. In other words, the complete syndrome pattern is not good or approximately good. An example of such a pattern is shown  in Fig.~\ref{fig:pattern_decompose}. One can further separate the column set of $\mathscr{W}$ into two types: $\textrm{W}|\mathscr{W}$ and $\textrm{C}|\mathscr{W}$. For $[\textsf{sc}_\mu]_{:,i}\in\textrm{W}|\mathscr{W}$, one has $[\widetilde{\textsf{sc}}_\mu]_{:,i}\neq[\textsf{sc}_\mu]_{:,i}$ and  $[\widetilde{\textsf{sc}}_\mu]_{j,i}\neq[\textsf{sc}_\mu]_{j,i}$; while for $[\textsf{sc}_\mu]_{:,i}\in \textrm{C}|\mathscr{W}$, one has $[\widetilde{\textsf{sc}}_\mu]_{:,i}\neq[\textsf{sc}_\mu]_{:,i}$ but $[\widetilde{\textsf{sc}}_\mu]_{j,i}=[\textsf{sc}_\mu]_{j,i}$.

It is always possible to find a good pattern $\mathscr{P}_{\text{Good}}$ whose complete syndrome array $\textsf{scg}_\mu$ generates exactly the same estimated complete syndromes for block $j$ as $\textsf{sc}_\mu$ $([\widetilde{\textsf{scg}}_\mu]_{j,:}=[\widetilde{\textsf{sc}}_\mu]_{j,:})$ after classical decoding. This $\mathscr{P}_{\text{Good}}$ can be constructed as follows:
If $[\widetilde{\textsf{sc}}_\mu]_{j,i}=0$, then  $[\textsf{scg}_\mu]_{:,i}$ is set to be an  all-zero column; if $[\widetilde{\textsf{sc}}_\mu]_{j,i}=1$, then
choose an arbitrary column $[\textsf{sc}_\mu]_{:,i'}$ in $\mathscr{P}$ such that $[\widetilde{\textsf{sc}}_\mu]_{j,i'}=1$, and then
set  $[\textsf{scg}_\mu]_{:,i}=[\textsf{sc}_\mu]_{:,i'}$.
We can then define another pattern $\mathscr{P}'$ as the difference between $\mathscr{P}$ and $\mathscr{P}_{\text{Good}}$, whose complete syndrome set array is denoted by $\textsf{sc}_\mu'$. From this construction, since $[\widetilde{\textsf{sc}}_\mu]_{j,:}$ is valid, $\textsf{scg}_\mu$ (and also $\textsf{sc}_\mu'$) is legal, in the sense that its syndromes are compatible in each block.
An example of this construction is shown in Fig.~\ref{fig:pattern_decompose}. Note that both $\mathscr{P}$ and $\mathscr{P}_{\text{Good}}$ give the same $[\widetilde{\textsf{sc}}_\mu]_{j,:}$ by the same classical decoding procedure.

Now, we show that the probability of pattern $\mathscr{P}$ is much smaller than that of $\mathscr{P}_{\text{Good}}$.
Suppose the failure sets $\mathscr{G}$, $\mathscr{G}_{\text{Good}}$ and $\mathscr{G}'$ induce  $\mathscr{P}$, $\mathscr{P}_{\text{Good}}$ and $\mathscr{P}'$, respectively. For each $j'$ such that $[\textsf{sc}_\mu]_{j',:}\neq 0$, we have $[\textsf{scg}_\mu]_{j',:}\neq [\textsf{sc}'_\mu]_{j',:}$, if $[\textsf{sc}_\mu]_{j,:}$ is not all zero (see Fig.~\ref{fig:pattern_decompose}, for example).
Then, for each $j'$, $\Pr\left\{\text{Q}^{(j')}_{E_{\mathscr{G}_{X \text{Good }}}}\bigcap
\text{Q}^{(j')}_{E_{\mathscr{G}'_X}} \neq \emptyset\right\}< \epsilon$ by assumption. (This probability should be much smaller, since the nontrivial contribution of $|\mathbb{S}_\cap|/|\mathbb{S}_{\slashed{\cap}}|$ comes  generally from high weight errors in Def.~\ref{def:epsilon-sparse}, which are less likely to happen.)
This suggests that for every time step $\step$ of the procedure, $\Pr\left\{\text{Q}^{(j)}_{\mathscr{GM}_{X\text{Good}}(\step)}\bigcap \text{Q}^{(j)}_{\mathscr{GM}_{X}'(\step)}\neq \emptyset \right\}<  \epsilon$.
Here, $\mathscr{GM}_{X \text{Good }}$ and $\mathscr{GM}_{X}'$ are the failure sets which have the highest probability among those failure sets $\mathscr{G}_{X \text{Good }}$ and $\mathscr{G}_{X}'$ causing $\mathscr{P}_{\text{Good}}$  and $\mathscr{P}'$. As a result, $\mathscr{GM}_{X\text{Good}}$ and $\mathscr{GM}'_X$ are independent with probability at least $1-O(\epsilon)$. We then have
\beq
\max_{\mathscr{G}}\Pr(\mathscr{P})\leq \max_{\mathscr{G}_{\text{Good}},\mathscr{G}'}\Pr(\mathscr{P}_{\text{Good}})\Pr(\mathscr{P}'),
\eeq
which implies that there exists at least one decomposition of $\mathscr{P}$ such that $\Pr(\mathscr{P})\ll\Pr(\mathscr{P}_{\text{Good}})$ for all $\mathscr{G}$ causing $\mathscr{P}$.
Thus, one can always find another failure set that generates a good syndrome pattern which causes the same valid complete estimated syndrome set $[\widetilde{\textsf{sc}}_\mu]_{j,:}$ with a much higher probability.

Overall, the corresponding syndrome array for all input blocks is most likely to be either good syndrome pattern or close to good pattern if $[\widetilde{\textsf{sc}}_\mu]_{j,:}$ is valid.

\end{proof}

Theorem~\ref{lemma:postselection_nogo} states the limitation on suppression of the high weight errors by a postselection protocol alone, and suggests that $t_{c_\mu}$ must be no less than $t-1$ to get qualified output ancillas. Moreover, we can prove that when this is true, an output block that passes ideal postselection is qualified, for sufficiently large block codes:
\begin{thm} \label{thm:3}
Suppose $\mathcal{Q}$ is an $\epsilon$-sparse quantum code. Then if the estimated complete syndrome set for one output block is valid, and $t_{c_\mu}\ge t-1$, this output block is qualified with probability at least $1-O(\epsilon)$.
\end{thm}
\begin{proof}
Consider a round of distillation to remove $X$ errors with output block $j$. Since $[\widetilde{\textsf{sc}}_\mu]_{j,:}$ can pass ideal postselection, it is valid. From Lemma~\ref{lemma:valid_to_good}, the syndrome pattern $\mathscr{P}$ is either good or approximately good with probability at least $1-O(\epsilon)$.
\begin{figure}[!htp]
\centering\includegraphics[width=70mm]{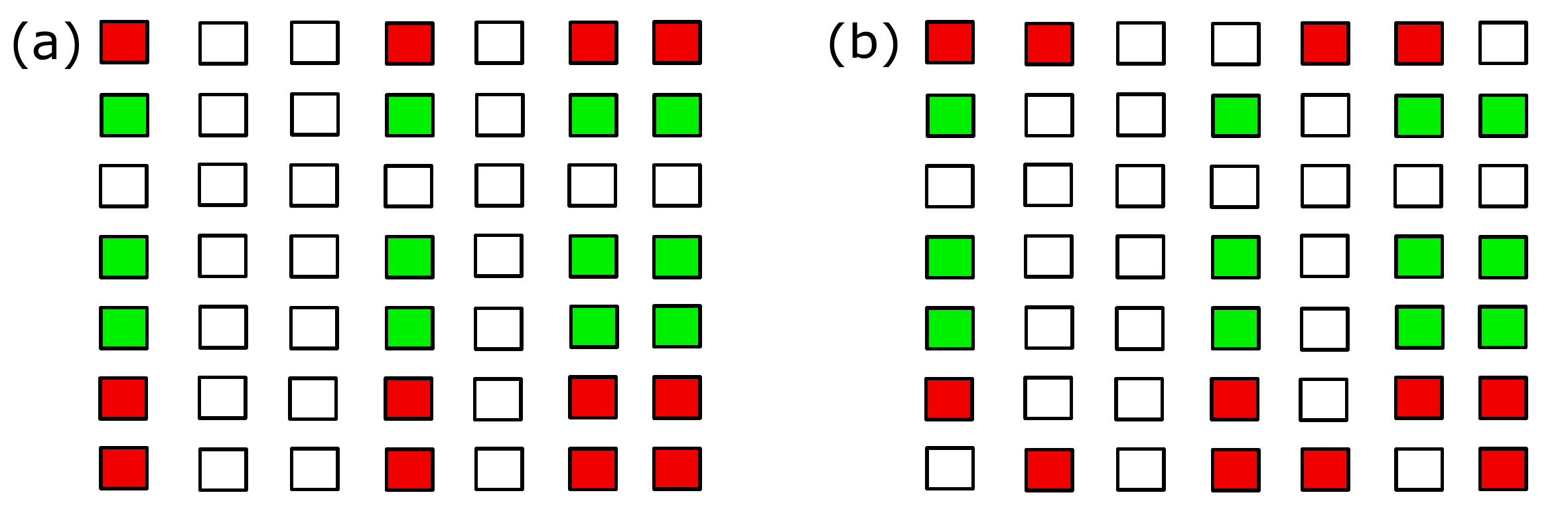}
\caption{\label{fig:approximate_good_decompose}(Color online) An illustration of an approximately good syndrome pattern for syndrome array $\textsf{sc}_\mu$ caused by $\mathscr{G}_{X}$ (green boxes) and $\mathscr{R}_X$ (red boxes) for good pattern (a) and approximately good pattern (b).}
\end{figure}

Consider first the case when $\mathscr{P}$ is good (see Fig.~\ref{fig:approximate_good_decompose}~(a) for an illustration). If the number of nonzero rows of $\textsf{sc}$, say $|\mathcal{B}_{\text{f}}|$, is less or equal to $t_c$, then  $[\widetilde{\textsf{sc}}_\mu]_{j,:}=[\textsf{sc}_\mu]_{j,:}$ after classical decoding and no $X$ error is left after quantum error correction.
If $|\mathcal{B}_{\text{f}}|>t_{c_\mu}$,  we assume $\mathscr{P}$ is caused by a combination of two failure sets:
1) a benign CNOT failure set $\mathscr{G}_X$ in the distillation circuit that flips certain rows of the syndrome array  to 1 (which are shown as the green squares in Fig.~\ref{fig:approximate_good_decompose}~(a));
and 2) the remaining CNOT failures in the distillation circuit and the  failures from the preparation stage which together form a failure set $\mathscr{R}_X$ that flips the rest of the nonzero rows of the syndrome array (which are shown as the red squares in Fig.~\ref{fig:approximate_good_decompose}~(a)).
Denote by $\mathcal{B}_{\mathscr{G}_X}$ and $\mathcal{B}_{\mathscr{R}_X}$ the index sets of the rows of the syndrome array affected by $\mathscr{G}_X$ and $\mathscr{R}_X$, respectively.  Consider the following two cases:
\begin{enumerate}
  \item  $\mathscr{G}_X\neq \emptyset$. For all $j'\in \mathcal{B}_{\mathscr{G}_X}$, one has $\text{QE}^{(j')}_{\mathscr{G}_X}
      =\text{Q}^{(j')}_{E_{\mathscr{G}_X}}\neq \emptyset$, since the $\mathscr{G}_X$ is benign. Since the distillation circuit is transversal, for all $j'\in \mathcal{B}_{\mathscr{G}_X}$, one has
\beq
\mathcal{I}\left(\mathcal{SC}_\mu^{(j')}|_{E_{\mathscr{G}_X}}\right)
=\mathcal{I}_{\mathcal{SC}_\mu\,|\,E_{\mathscr{G}_X}}=\mathcal{I}_{\mathcal{SC}_\mu\,|\,E_{\mathscr{R}_X}},
\eeq
where the second equality is because the syndrome pattern is good. Then, for output block $j$,  either $[\widetilde{\textsf{sc}}_\mu]_{j,i}=[\textsf{sc}_\mu]_{j,i},\forall i\in\mathcal{I}_{\mathcal{SC}_\mu\,|\,E_{\mathscr{G}_X}}$,
or $[\widetilde{\textsf{sc}}_\mu]_{j,i}\neq [\widetilde{\textsf{sc}}_\mu]_{j,i},\forall i\in\mathcal{I}_{\mathcal{SC}_\mu\,|\,E_{\mathscr{G}_X}}$   after classical decoding.
In the worst case, $[\widetilde{\textsf{sc}}_\mu]_{j,i}\neq [\textsf{sc}_\mu]_{j,i}, \forall i\in\mathcal{I}_{\mathcal{SC}_\mu\,|\,E_{\mathscr{G}_X}}$, which will leave an $X$ error on output block $j$. The support is the same as $\mathcal{I}\left(\text{QE}^{(j')}_{\mathscr{G}_X}\right)$. This case occurs with probability $O(p^{|\mathscr{G}_X|+|\mathscr{R}_X|})$. Since $w=\left|\mathcal{I}\left(\text{QE}^{(j')}_{\mathscr{G}_X}\right)\right|\leq |\mathscr{G}_X|$, block $j$ is qualified.

\item  $\mathscr{G}_X=\emptyset$. Now  all the colored squares in Fig.~\ref{fig:approximate_good_decompose}~(a) will be red, and this occurs with probability $O(p^{|\mathscr{R}_X|})$. In this case, an $X$ error of arbitrary weight (depending on $\mathcal{I}_{\mathcal{SC}_\mu\,|\,E_{\mathscr{G}_X}}$) may remain in output block $j$. Since now $\left|\mathscr{R}_X\right|\geq\left|\mathcal{B}_{\mathscr{R}_X}\right|>t_{c_\mu}\geq t-1$, it is uncorrelated, and block $j$ is still qualified.
\end{enumerate}

Now we consider the case of an approximately good syndrome pattern $\mathscr{P}_A$ (see Fig.~\ref{fig:approximate_good_decompose}~(b) for an illustration).
The argument is similar here.
If the maximum number of nonzero rows in $[\textsf{sc}_{\mu}]_{:,i}$ is less then $t_{c_\mu}$, $[\widetilde{\textsf{sc}}_\mu]_{j,:}=[\textsf{sc}_\mu]_{j,:}$ after classical decoding, and it leaves no error in the output block after quantum error correction. If this is not the case, we can also decompose the failures that lead to $\mathscr{P}_A$ into two sets:
a benign CNOT failure set $\mathscr{G}_X$ in the distillation circuit that flips certain identical rows of syndrome bits (the green squares in Fig.~\ref{fig:approximate_good_decompose}~(b)), and the remaining CNOT failures in the distillation circuit combined with the failures $\mathscr{R}_X$ in the preparation stage, which flip the syndrome bits of the rest of the nonzero rows (the red squares in Fig.~\ref{fig:approximate_good_decompose}~(b)).
$\mathcal{B}_{\mathscr{G}_X}$ and $\mathcal{B}_{\mathscr{R}_X}$ can be defined in the same way as for good pattern, and $\mathcal{B}_{\mathscr{G}_X}$ lies in the good part $\mathcal{B}_{\text{f}}$ of $\mathscr{P}_A$.  If $\mathscr{G}_X=\emptyset$, the situation is similar to the case of a good syndrome pattern. An $X$ error of arbitrary weight
may remain in output block $j$ with probability $O(p^{|\mathscr{R}_X|})$, where $\left|\mathscr{R}_X\right|\geq\left|\mathcal{B}_{\mathscr{R}_X}\right|>t_{c_\mu}\geq t-1$. When $\mathscr{G}_X\neq \emptyset$, for all $j'\in \mathcal{B}_{\mathscr{G}_X}$ one has $\text{QE}^{(j')}_{\mathscr{G}_X}=\text{Q}^{(j')}_{E_{\mathscr{G}_X}}\neq \emptyset$.
Since the distillation circuit is transversal, for all $j'\in \mathcal{B}_{\mathscr{G}_X}\subset \mathcal{B}_{\text{f}}$, one has
\beq
\mathcal{I}\left(\mathcal{SC}_\mu^{(j')}|_{E_{\mathscr{G}_X}}\right)
=\mathcal{I}_{\mathcal{SC}_\mu\,|\,E_{\mathscr{G}_X}}.
\eeq
Since $\mathscr{P}_A$ is an approximately good pattern, for $i\notin \mathcal{I}_{\mathcal{SC}_\mu\,|\,E_{\mathscr{G}_X}}$, $[\textsf{sc}_\mu]_{:,i}\in \mathscr{C}$ and $[\widetilde{\textsf{sc}}_\mu]_{j,i}=[\textsf{sc}_\mu]_{j,i}$ for output block $j$ (see Fig.~\ref{fig:good_pattern_approx}).
For $i_1, i_2\in\mathcal{I}_{\mathcal{SC}_\mu\,|\,E_{\mathscr{G}_X}} $ and $\text{SC}_{\mu,i_3}=\text{SC}_{\mu,i_1}\text{SC}_{\mu,i_2}\in \<S_\mu\>$, $[\textsf{sc}_\mu]_{:,i_3}=[\textsf{sc}_\mu]_{:,i_1}+ [\textsf{sc}_\mu]_{:,i_2}\in \mathscr{C}$, so either
\beq
[\widetilde{\textsf{sc}}_{\mu}]_{j,i_1}=[\textsf{sc}_{\mu}]_{j,i_1}, \ \ \ [\widetilde{\textsf{sc}}_{\mu}]_{j,i_2}=[\textsf{sc}_{\mu}]_{j,i_2},
\eeq
or
\beq
[\widetilde{\textsf{sc}}_{\mu}]_{j,i_1}\neq[\textsf{sc}_{\mu}]_{j,i_1}, \ \ \ [\widetilde{\textsf{sc}}_{\mu}]_{j,i_2}\neq [\textsf{sc}_{\mu}]_{j,i_2},
\eeq
because of the validity of $[\widetilde{\textsf{sc}}_\mu]_{j,:}$. As a result,  either $[\widetilde{\textsf{sc}}_\mu]_{j,i}=[\textsf{sc}_\mu]_{j,i}$ $\forall i\in \mathcal{I}_{\mathcal{SC}_\mu\,|\,E_{\mathscr{G}_X}}$,
or $[\widetilde{\textsf{sc}}_\mu]_{j,i}\neq [\widetilde{\textsf{sc}}_\mu]_{j,i}$ $\forall i\in \mathcal{I}_{\mathcal{SC}_\mu\,|\,E_{\mathscr{G}_X}}$  after classical decoding. In the worst case, $[\widetilde{\textsf{sc}}_\mu]_{j,i}\neq [\textsf{sc}_\mu]_{j,i}, \forall i\in\mathcal{I}_{\mathcal{SC}_\mu\,|\,E_{\mathscr{G}_X}}$, which will leave $X$ errors on the qubits corresponding to the index set $\mathcal{I}\left(\text{QE}^{(j')}_{\mathscr{G}_X}\right)$ on block $j$ with probability $O(p^{|\mathscr{G}_X|+|\mathscr{R}_X|})$. Since $w=\left|\mathcal{I}\left(\text{QE}^{(j')}_{\mathscr{G}_X}\right)\right|\leq |\mathscr{G}_X|$, block $j$ is qualified.

In summary, output block $j$ is qualified with probability $1-O(\epsilon)$ if it can pass ideal postselection.
\end{proof}
Combining the above two theorems, we have the following main theorem of this paper:
\begin{thm}\label{thm:main}
Suppose  $\mathcal{Q}$ is an $\epsilon$-sparse quantum code. For a single round of the distillation process, if and only if the following conditions are satisfied:
\begin{enumerate}
  \item The  failure rate $p$ of any noisy quantum gates  is sufficiently small;
  \item The classical error-correcting code $\cC_{c_\mu}$ can correct $t_{c_\mu}\ge t-1$ errors;
\end{enumerate}
then the output blocks will be qualified after ideal postselection with probability at least $1-O(\epsilon)$.
\end{thm}

\begin{remark}
Here, we emphasize that the sparsity of a quantum code is technically assumed for the convenience of proof, but it may not be necessary in practice. This is supported numerically by the state distillation for the quantum Golay code in the next section. In fact, for $\cC_{c_\mu}$ with sufficient sparse parity check matrix, (e.g., $M\leq t_{c_\mu}/2$), one can see that the probability to generate a syndrome pattern which is not good or approximately good but can pass the post-selection is less than $O(p^{t_{c_\mu}})$. In such scenario, the output ancillas are qualified. This is true for arbitrary CSS code.
\end{remark}
Although ideal postselection is not practical, one can use a more efficient error-detecting code, as long as the code length of $\cC_{d_1}$ $(\cC_{d_2})$ is less than $|\mathcal{SC}_1|$ $(|\mathcal{SC}_2|)$, to approximate it. Since one just needs to check parities  of the estimated syndrome bits and no actual decoding is required, a broad range of classical codes with good error-detecting ability can be chosen.

In general,  if we have a family of $\epsilon$-sparse codes, $\epsilon$ will approach zero when the code length is sufficiently large.
We also need to emphasize that Theorem~\ref{thm:main} just guarantees the output states are qualified  when $p$ is small and the length of an $\epsilon$-sparse code $\mathcal{Q}$ is sufficiently large. However, the constant behind big-O notation in Def.~\ref{def:uncorrelatation} could be large if $M$ for the $\textsf{A}$ matrix is not small enough. This fact suggests that one need to consider classical codes with sparse $\textsf{A}$ matrices for FTQC purpose.

Asymptotically, the rejection rate for each round of distillation is $O(p^2)$, because at least two failures are needed to cause a rejection of the output blocks. Since the rejection rate is $O(p)$ for naive verification and other ancilla preparation schemes~\cite{steane2002fast,paetznick_ben2011fault}, our distillation protocol is potentially more efficient in the small $p$ regime. Suppose the rejection rates are, respectively, $R_1(p)$ and $R_2(p)$ for the two rounds of distillation.
The overall yield rate for the fault-tolerant distillation protocol would be
\beq
\text{Yield}_{\text{FT}}(p)=
\frac{k_{c_1}k_{c_2}(1-R_{\text{FT}}(p))}{n_{c_1}n_{c_2}},
\eeq
where $R_{\text{FT}}(p)=1-(1-R_1(p))(1-R_2(p))$ is the overall block rejection rate of the fault-tolerant distillation protocol. It is possible that $R_1(p)$ and $R_2(p)$ are both negligible in the small $p$ regime, so that one may obtain qualified ancilla states without sacrificing the overall yield rate. Compared with the yield rate of a naive verification process (Eq.~(\ref{eq:naive_yield})), our distillation protocol can achieve a yield rate \emph{independent} of the distance of $\mathcal{Q}$, by adjusting the error-detection codes used for postselection, thus boost from $O(t^{-2})$ to $O(1)$ in practice for an $[[n,k,d=2t+1]]$ CSS code if appropriate family of classical codes are chosen.

\section{Numerical Estimation: Quantum Golay code}\label{sec:numerical}
In this section, we numerically study the performance of the distillation protocol proposed in Sec.~\ref{sec:FTpreparation}. We will focus on the preparation of $|0\>^{\otimes k}_L$, since the preparation procedure for other ancilla states is essentially the same.

The first purpose of numerical simulation is to check whether the asymptotic scaling of the error weight distribution behaves correctly. That is, the probabilities of residual errors in the output blocks with weight $w$, $P_X(w)$ and $P_Z(w)$, should scale like $O(p^w)$ when $p$ is sufficiently small. This will ensure the quality of the ancilla states.
We also want to show that the results in the previous section hold for quantum codes of finite length.

The second purpose is to study the error performance of the distillation procedure and the overall yield rate. We may ask  the effective error rate when the prepared ancillas contain only independent errors. A proper way to study  this effective error rate is to approximate the error performance by some binomial distribution, where the probability parameter of that binomial distribution can be regarded as the effective error rate. Since  correlated errors of higher weight are what concern us, we can quantify the distribution of the highest weight errors by a binomial distribution.
Since we are preparing $|0\>_L^{\otimes k}$, we need to consider $X$ errors of weight higher than $t$ and $Z$ errors of weight $t$  (a $Z$ error of weight $w > t$ is equivalent to one with $w'\le t$ for $|0\>_L^{\otimes k}$). That is, for the $X$ and $Z$ errors, the effective error rate $p_{\text{eff}\,|\,X}$ and $p_{\text{eff}\,|\,Z}$ can be defined by:
\beq
P_X(w>t)=\sum_{w>t} \binom{n}{w} \ p_{\text{eff}\,|\,X}^w (1-p_{\text{eff}\,|\,X})^{n-w},
\eeq
and
\beq
P_{Z}(w=t)=\binom{n}{t} \ p_{\text{eff}\,|\,Z}^t (1-p_{\text{eff}\,|\,Z})^{n-t}.
\eeq
Clearly, $p_{\text{eff}}$ depends on the chosen quantum and classical codes, which makes it difficult to estimate analytically.
Similarly, the yield rate is difficult to estimate, since the overall rejection rate $R_{\text{FT}}$ depends on the gate failure rate, the structure of  $\textsf{A}_{c_1}$, $\textsf{A}_{c_2}$ in the parity-check matrices, and the specific error-detecting codes. However, $P_X(w>t)$ and $P_Z(w=t)$ can be directly estimated from numerical simulations.

In the following we choose our quantum code $\mathcal{Q}$ to be the $[[23,1,7]]$ quantum Golay code, which is constructed from the classical  $[23,12,7]$ Golay code.
Let both $\cC_X$ and $\cC_Z$  be  the $[23,12,7]$ Golay code. Note that $C_{X}^\perp$ is a $[23,11,8]$ cyclic code, which is contained in $\cC_X$ and $\cC_Z$. Efficient ancilla preparation for the quantum Golay code  has been extensively studied in Ref.~\cite{paetznick_ben2011fault}
by using its permutation symmetry. But we ignore this symmetry in this paper, and focus on the performance of our protocol. The classical $[23,12,7]$ Golay code is the smallest \emph{perfect} code with $t$ larger than one, and this simplifies the study of correlated noise after distillation. This is because its decoder will always take the state back to the code space.  The parity-check matrices $\textsf{H}_X$ and $\textsf{H}_Z$ can be represented as follows:
\beq
{\small
\begin{bmatrix}
1 0 0 0 0 0 0 0 0 0 0 1 1 1 1 1 0 0 1 0 0 1 0\\
0 1 0 0 0 0 0 0 0 0 0 0 1 1 1 1 1 0 0 1 0 0 1\\
0 0 1 0 0 0 0 0 0 0 0 1 1 0 0 0 1 1 1 0 1 1 0\\
0 0 0 1 0 0 0 0 0 0 0 0 1 1 0 0 0 1 1 1 0 1 1\\
0 0 0 0 1 0 0 0 0 0 0 1 1 0 0 1 0 0 0 1 1 1 1\\
0 0 0 0 0 1 0 0 0 0 0 1 0 0 1 1 1 0 1 0 1 0 1\\
0 0 0 0 0 0 1 0 0 0 0 1 0 1 1 0 1 1 1 1 0 0 0\\
0 0 0 0 0 0 0 1 0 0 0 0 1 0 1 1 0 1 1 1 1 0 0\\
0 0 0 0 0 0 0 0 1 0 0 0 0 1 0 1 1 0 1 1 1 1 0\\
0 0 0 0 0 0 0 0 0 1 0 0 0 0 1 0 1 1 0 1 1 1 1\\
0 0 0 0 0 0 0 0 0 0 1 1 1 1 1 0 0 1 0 0 1 0 1 \\
\end{bmatrix}.
}
\eeq
Similarly, both logical $X$ and $Z$ operators of $\cQ$ share the binary representation:
\[l=
\begin{bmatrix}
0 0 0 0 0 0 0 0 0 0 0 1 0 1 0 1 1 1 0 0 0 1 1\\
\end{bmatrix}.
\]
and thus $\bar{X}=X^l$ and $\bar{Z}=Z^l$.

We need to choose classical codes $\cC_{c_1}$ and $\cC_{c_2}$ that can correct at least two errors to produce qualified ancillas, according to Theorem~\ref{thm:main}.
Meanwhile, we hope that the number of 1s in each column of $\textsf{A}_{c_1}$ and $\textsf{A}_{c_2}$ is small enough so that the effect of error propagation during state distillation will not be too bad.

Here we choose two  classical codes of distance five: the $[15,7,5]$ BCH code, and the $[5,1,5]$ repetition code. The parity-check matrix for the $[15,7,5]$ code is
\beq\label{eq:parity_check_15}
{\small
\begin{bmatrix}
100000001101000\\
010000000110100\\
001000000011010\\
000100000001101\\
000010001101110\\
000001000110111\\
000000101110011\\
000000011010001\\
\end{bmatrix}.
}
\eeq
For comparison, we also study the case when $\cC_{c_1}$ and $\cC_{c_2}$ are codes of distance three. Two typical examples are the $[3,1,3]$ repetition code and the $[7,4,3]$ Hamming code. Codes with distance 5 will be used in Sec.~\ref{sec:d5} and codes with distance 3 will be discussed in Sec.~\ref{sec:d3}.

We also need two more classical codes $\cC_{d_1}$ and $\cC_{d_2}$ to check the compatibility of  the estimated generalized syndrome bits. $\cC_{d_1}$ and $\cC_{d_2}$ should encode $k_{d_1}=r_Z+k=12$ and $k_{d_2}=r_X=11$ bits, respectively.  Two very natural candidates for our purpose are the $[23,12,7]$ Golay code and its dual, which is a $[23,11,8]$ cyclic code.

Numerical simulation of a full cycle of the state distillation procedure is done as follows:
\begin{enumerate}
  \item We first prepare $n_{c_1}(n_{c_2}+n^\prime)$ different blocks of qubits, and prepare each block in the $|0\>_L$ state of the quantum Golay code by a noisy encoding circuit. An optimized encoding circuit can be found in ~\cite{paetznick_ben2011fault}. Every $n_{c_1}$ blocks are put together to form a group. The extra $n^\prime$ groups of blocks are also prepared so that when some blocks are discarded in the first round, we can still keep $k_{c_1}n_{c_2}$ blocks for each group by appending extra blocks, to apply the second round distillation. The value of $n^\prime$ is determined by the rejection rate of the first round, and it can be optimized during simulation.
  \item For $(n_{c_2}+n^\prime)$ groups, the distillation circuit of the first round is applied to each group, with noisy CNOT gates followed by noisy bitwise $Z$ basis measurements on the first $r_{c_1}$ blocks to get parity matrix $\sigma_1$ of the generalized syndromes of $Z$ stabilizer generators and $\bar{Z}$.
  \item Estimate the generalized syndromes $[\textsf{s}_1]_{j,:}$ for each remaining block by classical decoding of $\cC_{c_1}$. For each block, use the parity check matrix of the $[23,8,12]$ code in systematic form, $\textsf{H}_{d_1}$, to check the compatiblity of these estimated generalized syndromes. For the $j$th block, if $\textsf{H}_{d_1}[\widetilde{\textsf{s}}_1]^T_{j,:}=0$, that block is accepted; otherwise, the block will be discarded.
  \item For each group of blocks, if certain blocks have been discarded, append accepted output blocks from the extra $n'$ groups so that each group retains $k_{c_1}$ blocks.
  \item Do quantum error correction on the remaining blocks, following the recipe in Sec.~\ref{sec:noiseless_prep} to remove $X$ errors.
  \item For each index $j$ of the output blocks, assemble the $j$th block from each group to form a new group. Each new group contains $n_{c_2}$ ancilla blocks, and there are $k_{c_1}$ such groups in total. This step randomizes the correlated $Z$ errors among blocks previously in the same group.
  \item The second round distillation circuit is applied to each new group, with noisy CNOTs followed by noisy bitwise $X$ basis qubit measurements on the first $r_{c_2}$ blocks to get $\sigma_2$, the parity of generalized syndromes of the $X$ stabilizer generators.
  \item Estimate the binary vector $[\textsf{s}_2]_{j,:}$ for each output block by classical decoding. For each output block, check if $\textsf{H}_{d_2}[\widetilde{\textsf{s}}_2]^T_{j,:}=0$ to decide whether the corresponding block is accepted or discarded.
  \item Do quantum error correction on the remaining blocks to remove $Z$ errors.
  \item Check the weight of the remaining error and calcuate the weight distribution for both $X$ and $Z$ errors. Note that, since we are preparing $|0\>_L$ state with code distance 7, $Z$ errors with weight larger than 3 are equivalent to $Z$ errors with weight less than or equal to 3. So we just need to calculate the probability distribution of weight 1, 2, 3, and larger than 3 errors for $X$ errors, and weight 1, 2, 3 for $Z$ errors.
\end{enumerate}

Note that for the last step above, one can not directly count the weights of the remaining error, because of the error degeneracy of quantum codes: two errors are equivalent if their product is a generalized stabilizer element. So certain high weight errors remaining after distillation in the simulation may actually be equivalent to low weight error. To cope with this problem, we check the generalized syndromes of the remaining error $E_{\text{rem}}$ for each block and compare to a list of syndromes of all operators that are tensor products of $X$ and $Z$ with weight less than or equal to 4. If a certain operator has the same syndrome as $E_{\text{rem}}$, it is an equivalent error of $E_{\text{rem}}$ with the lowest weight, and its weight is regarded as the weight of $E_{\text{rem}}$. If $E_{\text{rem}}$ does not have an equivalent error with weight less or equal to 4, its weight is higher than 4. For larger quantum codes with higher distances, the enumerating all low weight errors would be too tedious. For special states like $|0\>_L$, we can decode $\cC_X^\perp$ after state preparation to find the minimum weight $X$ error that shares the same syndromes with $E_{\text{rem}}$. The difficulty comes from the lack of an efficient decoder for $\cC_X^\perp$ for a typical code $\cC_X$. However, a decoder with worse performance may be used to approximately find the minimum weight error for low error weight. (This is just a technical difficulty in numerical benchmarking, and cause no trouble in practise, since such decoding is not necessary when real state distillation is carried out.)

\subsection{Classical codes of distance five }\label{sec:d5}

\subsubsection{Combination A: [15,7,5]+[15,7,5]}
Figure~\ref{fig:1515_x} shows the probability of remaining $X$ errors after distillation with weight 1, 2, 3, and larger than 3 ($P_X(w=1)$, $P_X(w=2)$, $P_X(w=3)$, and $P_X(w>3)$, respectively) using the $[15,7,5]$ classical BCH code for both rounds of distillation, with postselection. The three dashed reference lines represent the probability distributions of $O(p)$ (green), $O(p^2)$ (brown), and $O(p^3)$ (purple) are also given for comparison. It is clear that the probability of weight $w$ $X$ errors scales as $O(p^w)$ asymptotically for $w\leq3$ when $p$ is sufficiently small; when $w>3$, since $[15,7, 5]$ can only correct 2 errors, the  distribution can only scale as $O(p^3)$.

\begin{figure}[!htp]
\centering\includegraphics[width=65mm]{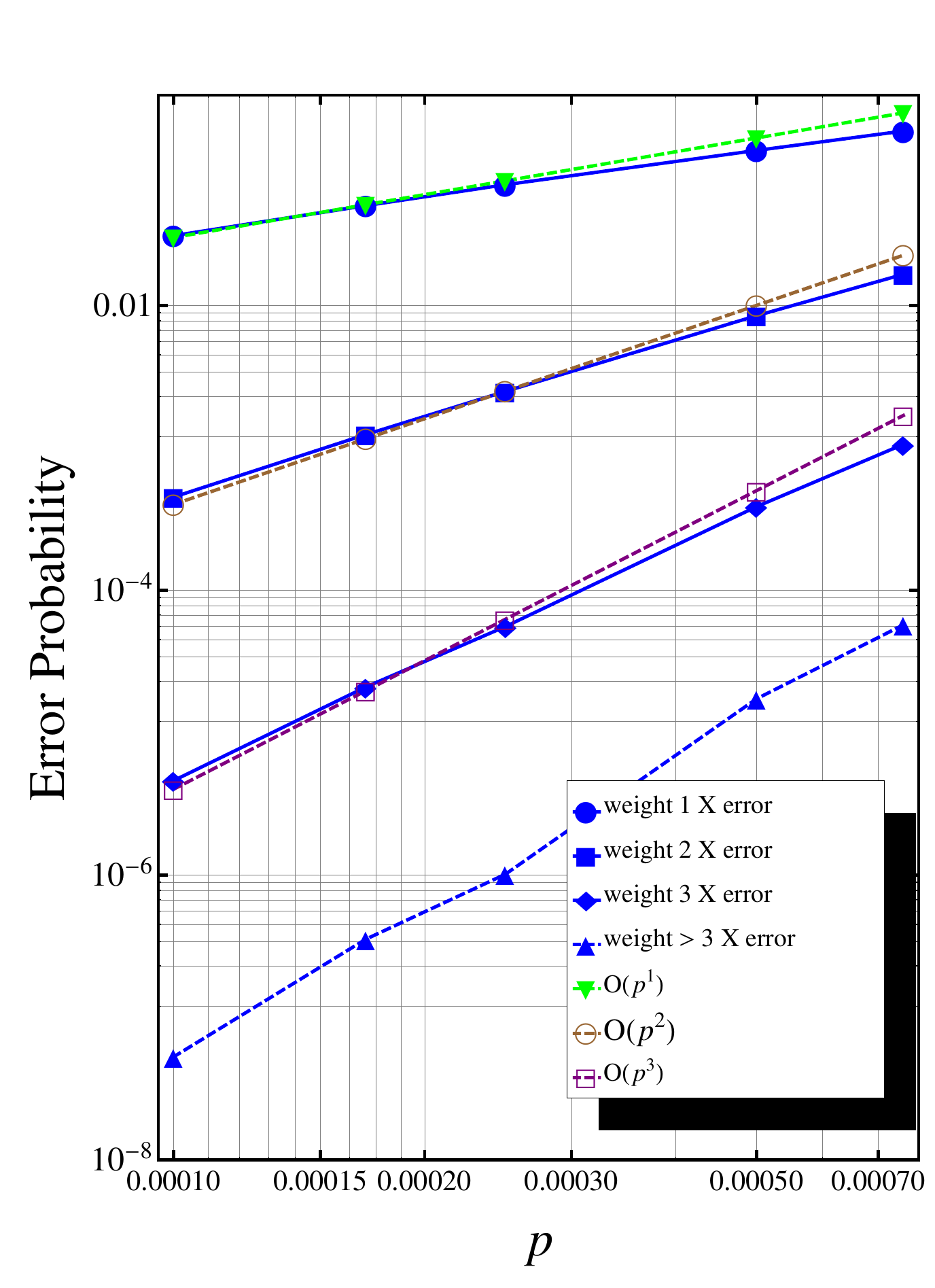}
\caption{\label{fig:1515_x}(Color online) Probability of remaining $X$ errors with weight 1, 2, 3  and larger than 3 on a single output block vs gate failure rate $p$, after distillation using the classical $[15, 7, 5]$ BCH code for both rounds, with postselection. The dashed green, brown and purple lines correspond to weight distributions $O(p)$, $O(p^2)$ and $O(p^3)$, for reference.}
\end{figure}
\begin{figure}[!htp]
\centering\includegraphics[width=65mm]{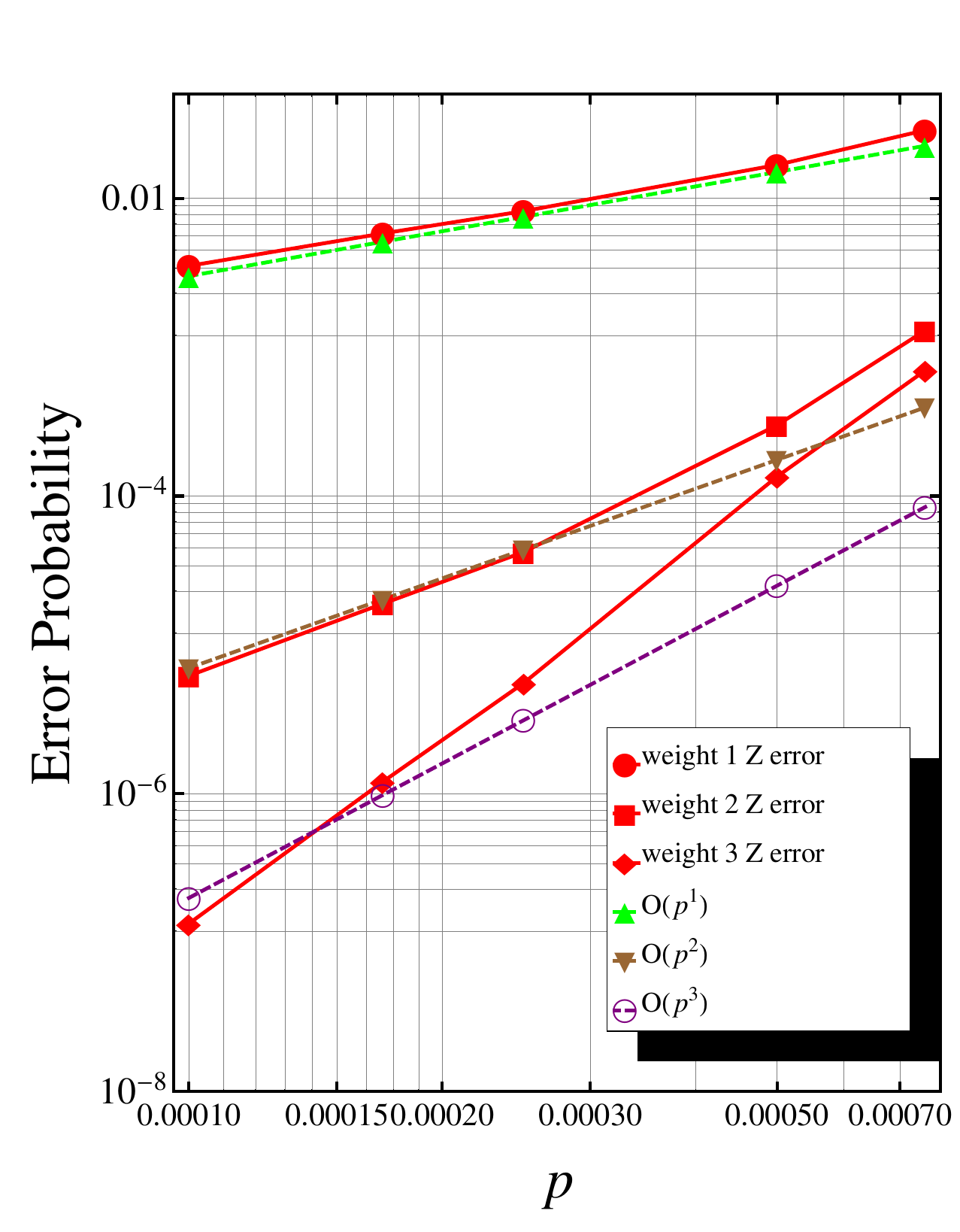}
\caption{\label{fig:1515_z}(Color online) Probability of remaining of $Z$ errors with weight 1, 2, 3 on a single output block vs gate failure rate $p$, after distillation using the classical $[15, 7, 5]$ BCH code for both rounds.}
\end{figure}

Similarly, Fig.~\ref{fig:1515_z} shows the error weight distribution for $Z$ errors. Note that for $Z$ errors, we only care about those of weight less than or equal to 3. Compared with $X$ errors, $Z$ errors scale correctly only when $p < 0.0003$ for $w=2,3$. This is because after the first round of distillation, the error rate of the remaining $Z$ errors on the output blocks is amplified, and will be beyond the error correction ability of the $[15,7,5]$ code for large values of $p$. Also note that the error rates for $Z$ errors of all weights are much smaller than those of $X$ errors of the same weight, because the remaining $X$ errors from the first round propagate in the second round to the output blocks, and these $X$ errors, though mostly uncorrelated, cannot be detected and corrected. Although the probability scaling up for $X$ errors is correct, their absolute values are much higher than for $Z$ errors. Overall, Fig.~\ref{fig:1515_x} and Fig.~\ref{fig:1515_z} qualitatively verify the validity of Theorem~\ref{thm:main} for using the $[15,7,5]$ code for distillation, and also suggest that the $[23,8,11]$ and $[23,7,12]$ codes are sufficiently good as error detection codes $\cC_{d_1}$ and $\cC_{d_2}$ for the purpose of postselection.

\begin{figure}[!htp]
\centering\includegraphics[width=85mm]{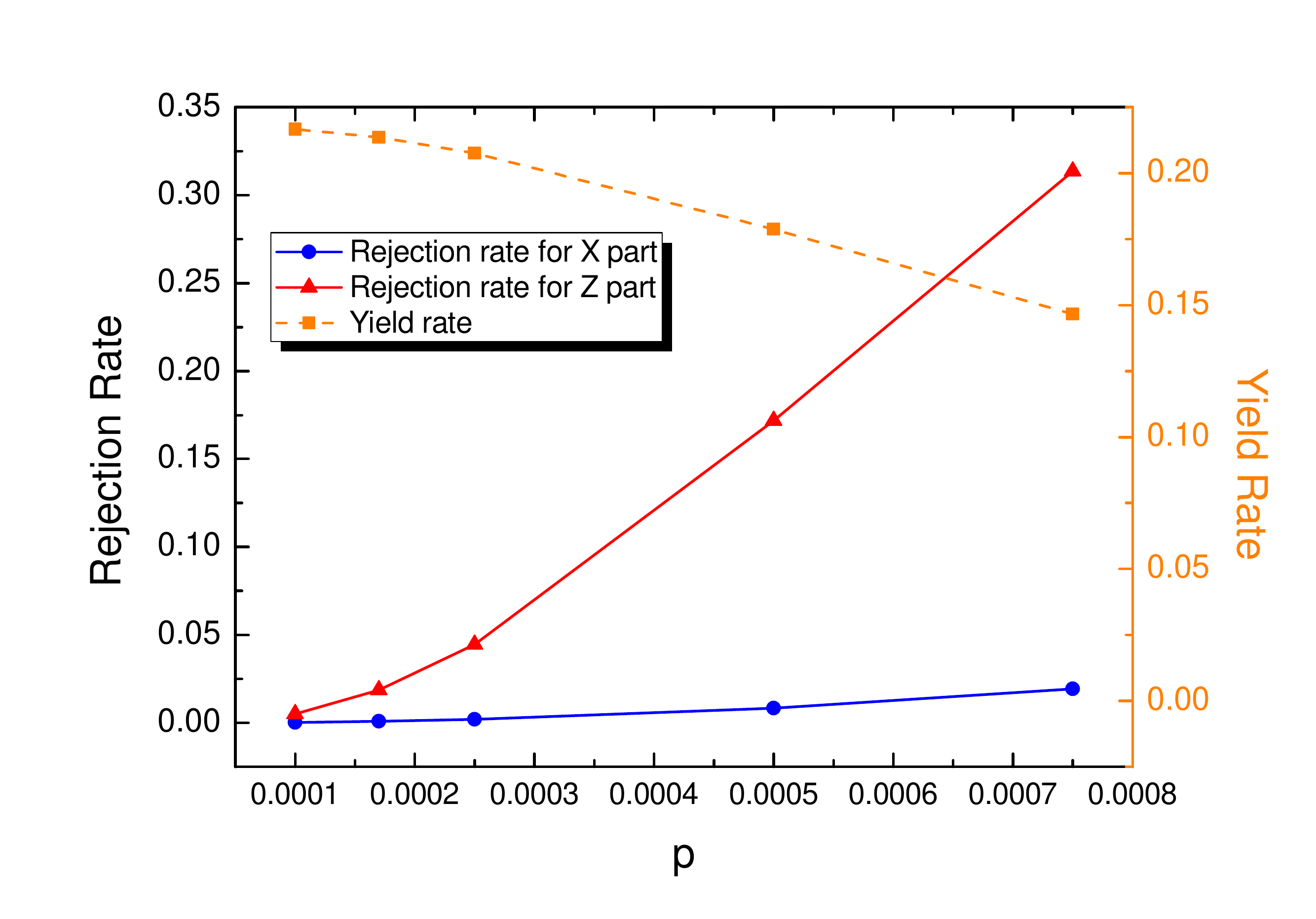}
\caption{\label{fig:1515_rej}(Color online) The rejection rates for both rounds of distillation, and overall yield rate, using the $[15,7,5]$ code in both rounds.}
\end{figure}

The rejection rates in both rounds of distillation and overall yield rate are shown in Fig.~\ref{fig:1515_rej}. For all values of $p$ considered here, the rejection rate for the $X$ part of the distillation can be neglected, and for $p < 0.00025$, the rejection rate for the $Z$ part is less than $5\%$. The overall yield rate will be around $20\%$ in that gate failure regime, which is much more efficient than naive way  verification (see Fig.~\ref{fig:naive}). The rejection rates for both the $X$ and $Z$ parts scale as $O(p^2)$, as predicted in Theorem~\ref{thm:nogo}. Note that for larger distance codes like the $[[255, 143, 15]]$ quantum BCH code, which is of practical interest for the large block code FTQC scheme, it would be much more efficient to use the distillation protocol than a verification circuit in the low physical failure rate region if proper $\cC_{c_1}$ and $\cC_{c_2}$ are chosen. This will be studied in more detail in future work.

\begin{figure}[!ht]
\centering\includegraphics[width=70mm]{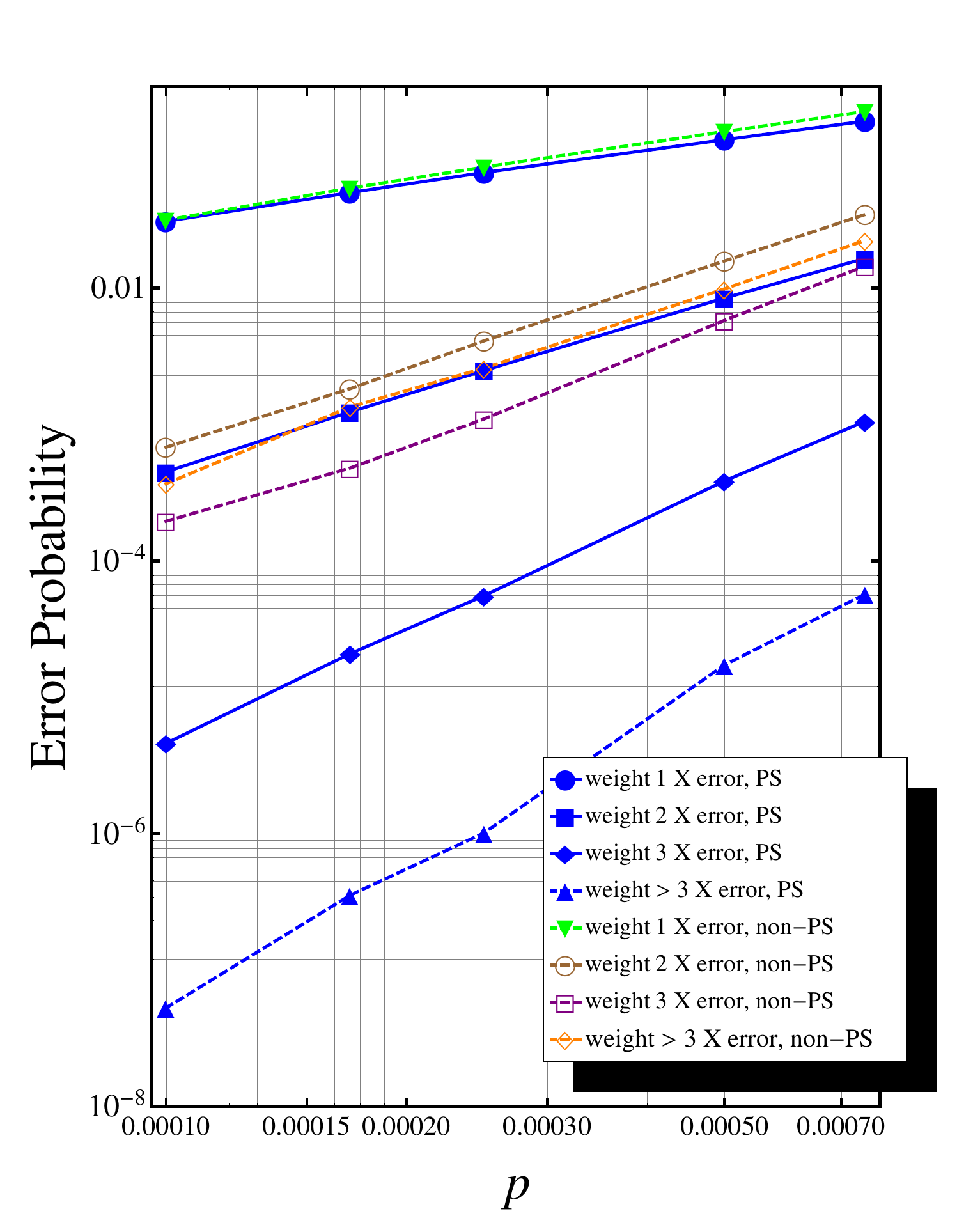}
\caption{\label{fig:X_ps_non_ps} (Color online) Error weight distribution after distillation for $X$ errors, using the classical $[15, 7, 5]$ BCH code for both rounds, with and without postselection. The dashed green, brown, orange and purple lines represents the probability of remaining $X$ errors with weights 1,2,3, and larger than 3, without postselection.}
\end{figure}

In Figs.~\ref{fig:X_ps_non_ps} and~\ref{fig:Z_ps_non_ps} the error weight distributions with and without postselection are compared. For both parts, as stated in Theorem~\ref{thm:nogo},
the probability of errors of weight larger than 2  scale as $O(p^2)$ in the low failure rate region. Moreover, the probabilities of errors of weight 2 and larger than 2 remaining are close to each other for both $X$ and $Z$ errors, which means that the distillation protocol almost cannot  reduce high weight errors at all if no postselection procedure is applied.
\begin{figure}[!ht]
\centering\includegraphics[width=65mm]{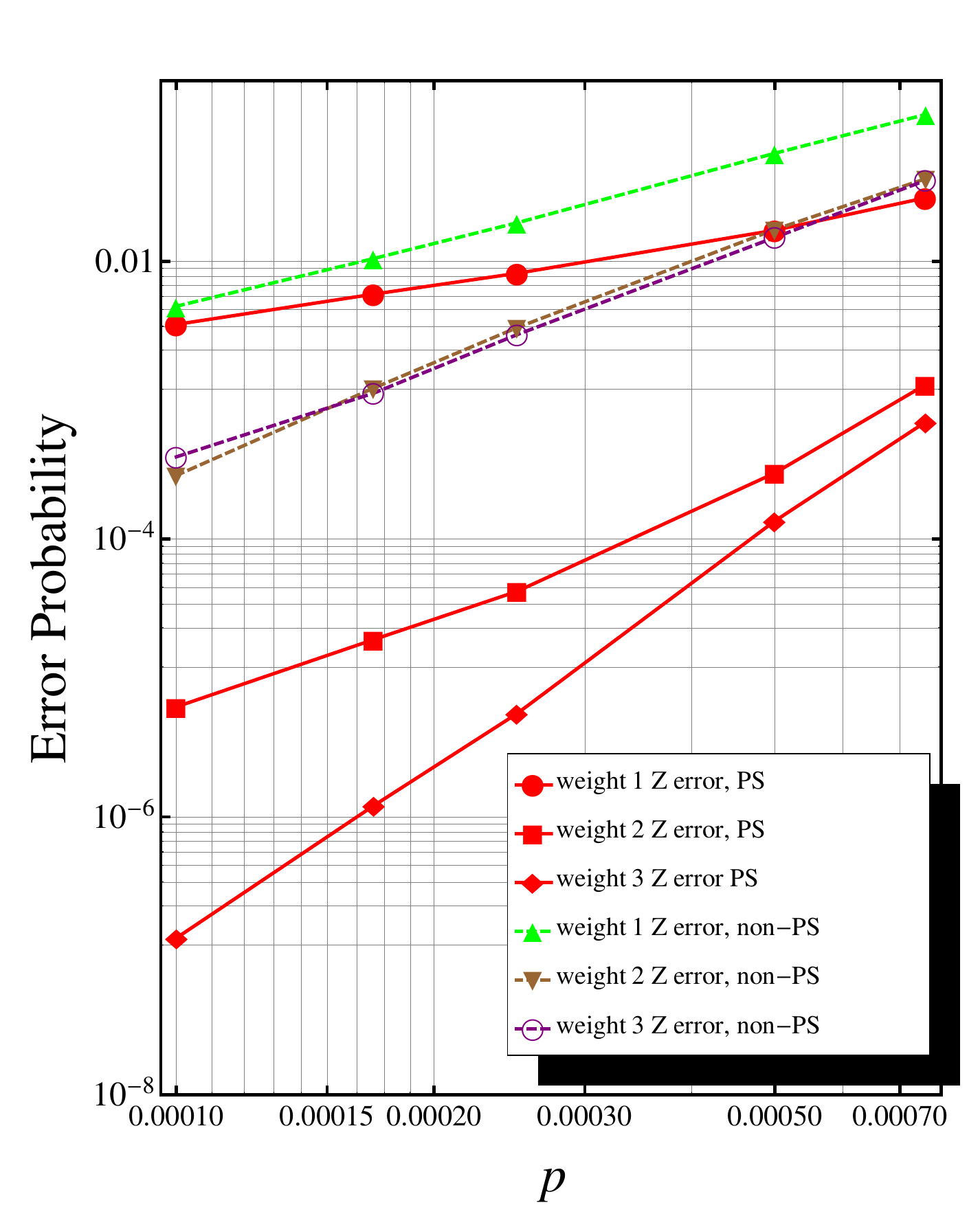}
\caption{\label{fig:Z_ps_non_ps}(Color online) Error weight distribution after distillation for $Z$ errors using the classical $[15, 7, 5]$ BCH code for both rounds, with and without postselection (PS). The dashed lines show the probabilities of the remaining errors without postselection. }
\end{figure}

For $p=0.0001$, if postselection is applied, the probability of $X$ errors with weight 3 is two orders of magnitude smaller; while for $X$ errors with weight larger than 3, the probability is reduced by four orders of magnitude. Similarly, the probability of remaining $Z$ errors of weight 2 is reduced by two orders of magnitude; while for weight-3 errors, the probability is reduced by three orders of magnitude. Since the rejection rate is negligibly small in this regime, postselection seems to discard almost  exactly all bad ancilla blocks with only slight overkill! Thus, it can greatly improve the quality of the output ancillas   with very little sacrifice of the overall yield rate.
\begin{figure}[!ht]
\centering\includegraphics[width=65mm]{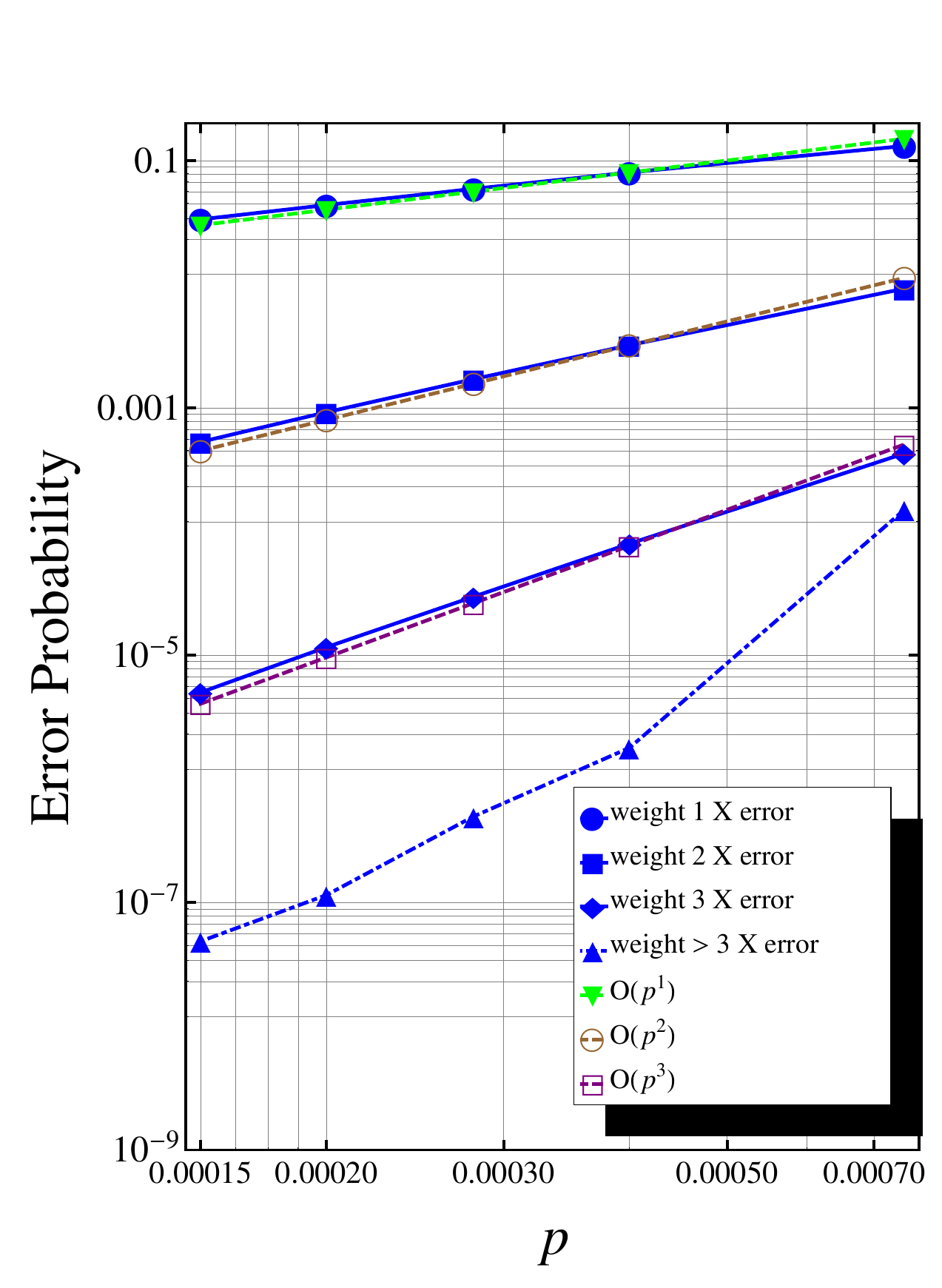}
\caption{\label{fig:155X}(Color online) Weight distribution for $X$ errors remaining on output blocks using the $[15,7,5]$ BCH code and the $[5,1,5]$ repetition code with postselection.}
\end{figure}

We are also interested in the strength of the remaining errors in the low gate failure rate region, where schemes using large block codes are good enough to complete a nontrivial quantum algorithm. The effective rates for $X$ and $Z$ errors are $p_{\text{eff}|X} = 1.67\times 10^{-3}$ and  $p_{\text{eff}|Z} = 3.83\times 10^{-4}$, barely good enough for the scheme in Ref.~\cite{brun2015teleportation}, where the quantum Golay code is concatenated with another code to form a large data code block. One can see that there is an asymmetry for the two types of errors: for $Z$ errors, the effective error rate is about the same as the physical failure rate, while for $X$, the error rate is amplified by an order of magnitude. This seems always be the case for a two-stage distillation protocol.

\begin{figure}[!ht]
\centering\includegraphics[width=65mm]{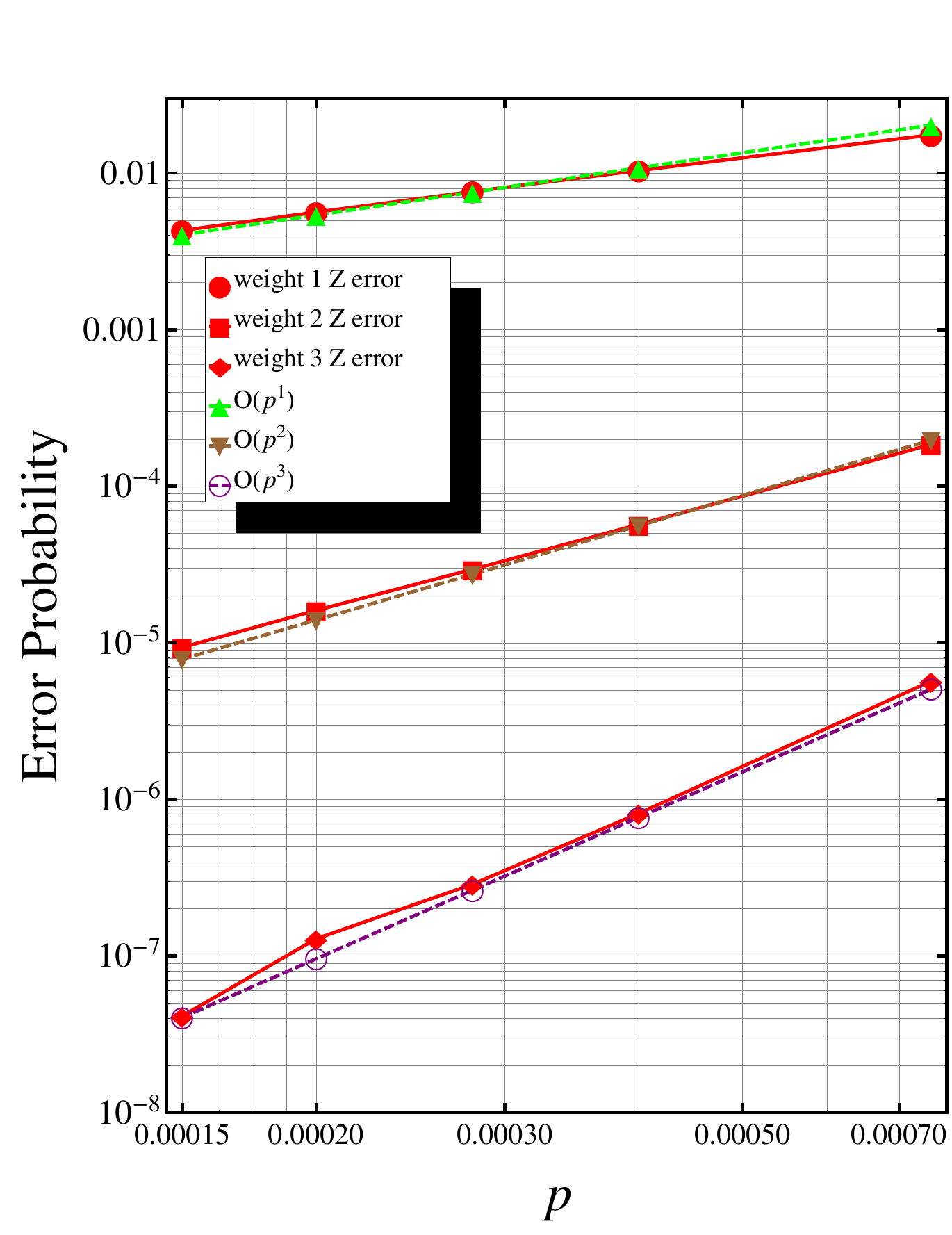}
\caption{\label{fig:155Z}(Color Online) Weight distribution of $Z$ errors remaining on output blocks using the $[15,7,5]$ BCH code and the $[5,1,5]$ repetition code, with postselection.}
\end{figure}

\subsubsection{Combination B:[15,7,5]+[5,1,5]}
Alternatively, we can use different distance-5 classical codes $\cC_{c_1}$ and $\cC_{c_2}$ to see how different codes affect the quality of the output ancillas. We keep the $[15,7,5]$ BCH code for $\cC_{c_1}$, and change $\cC_{c_2}$ to be the $[5,1,5]$ repetition code, which has relatively larger error correction ability than the $[15,7,5]$ code. One might expect that the ability to suppress $Z$ errors is even stronger in this case, which indeed turns out to be true.
\begin{figure}[!ht]
\centering\includegraphics[width=85mm]{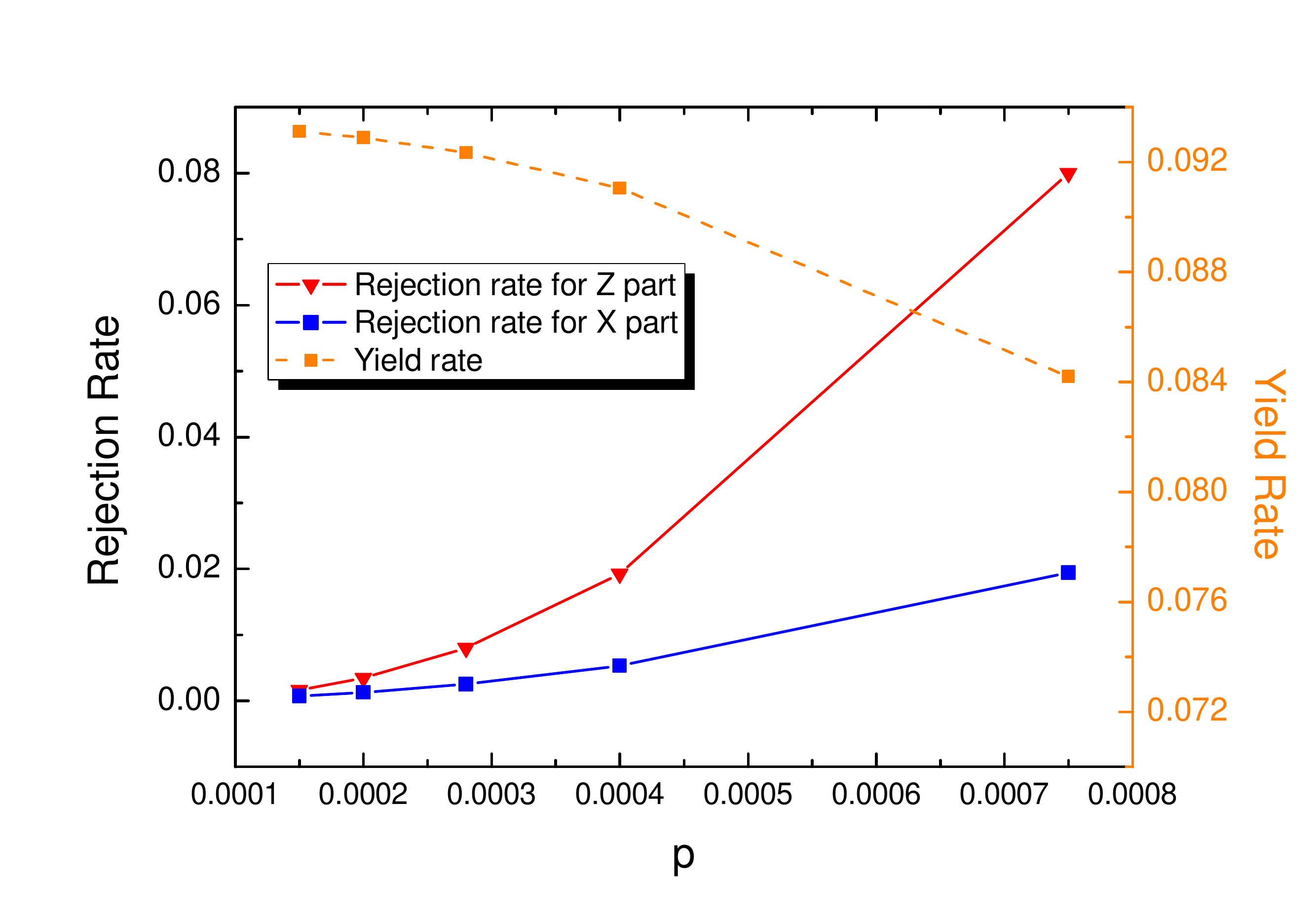}
\caption{\label{fig:155_rej}(Color online) The rejection rate for both rounds of distillation and overall yield using the $[15,7,5]$ and $[5,1,5]$ codes.}
\end{figure}

Figures~\ref{fig:155X} and~\ref{fig:155Z} show the weight distributions for $X$ and $Z$ errors with this combination of codes and Fig.~\ref{fig:155_rej} shows the rejection rate and overall yield rate. Compared to Combination A, the weight distribution scales correctly for all physical failure rates that we considered for both $X$ and $Z$ errors.
The probability for a weight-2 $Z$ error remains the same, while for weight-3 errors it can be reduced by an order of magnitude for $p<0.00025$. However, the probabilities for residual $X$ errors are almost the same (or slightly less) for all weight as in the previous case. This is because for the $[5,1,5]$ code, the corresponding $\textsf{A}_{c_2}$ matrix has four 1s in a single column, which is the same for five out of the seven columns of $\textsf{A}_{c_2}$ for the $[15,7,5]$ code (see Eq.~(\ref{eq:parity_check_15}): two columns contains five 1s). The rejection rate for the $Z$ part is also reduced for all $p$, because only a small set of CNOT failure combinations will cause a rejection in the $Z$ part. The overall yield rate is less than half that of Combination A, because of the low encoding rate of the $[5,1,5]$ code. Overall, this combination reduces the error rates for distilled blocks, but with much lower yield.

From these two examples, it seems that there is a tradeoff between the yield and the rate of effective errors remaining on the distilled blocks for small classical codes.
Asymptotically, this tradeoff may disappear for large quantum codes   and large classical codes with sparse parity-check matrices.

\subsection{Classical codes of distance three }\label{sec:d3}
In this section, we study the case when distance-3 classical codes are used in the distillation protocol, to further verify the necessity of the conditions in Theorem~\ref{thm:main}, and complete the study of distilling $|0\>_L$ for the quantum Golay code.
\begin{figure}[!ht]
\centering\includegraphics[width=60mm]{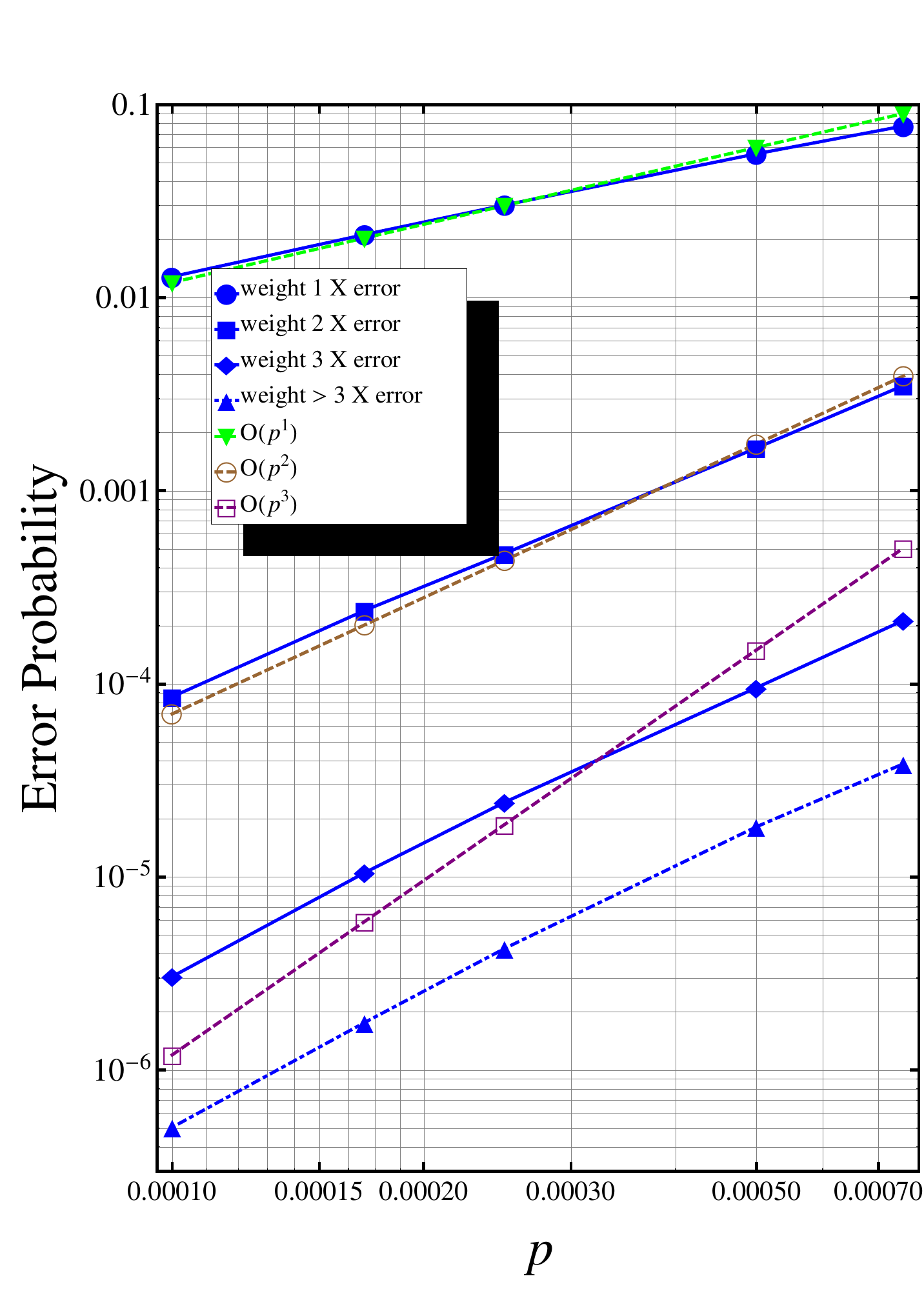}
\caption{\label{fig:Hamming_X}(Color online) Weight distribution of $X$ errors remaining on output blocks using the Hamming code for both rounds with postselection. }
\end{figure}

\subsubsection{Combination C: Hamming code+Hamming code}

Figures~\ref{fig:Hamming_X} and~\ref{fig:Hamming_Z} show the error weight distributions for remaining $X$ and $Z$ errors on a single output block using the Hamming code for both rounds of distillation, with reference dashed lines representing $O(p)$, $O(p^2)$ and $O(p^3)$.

\begin{figure}[!ht]
\centering\includegraphics[width=60mm]{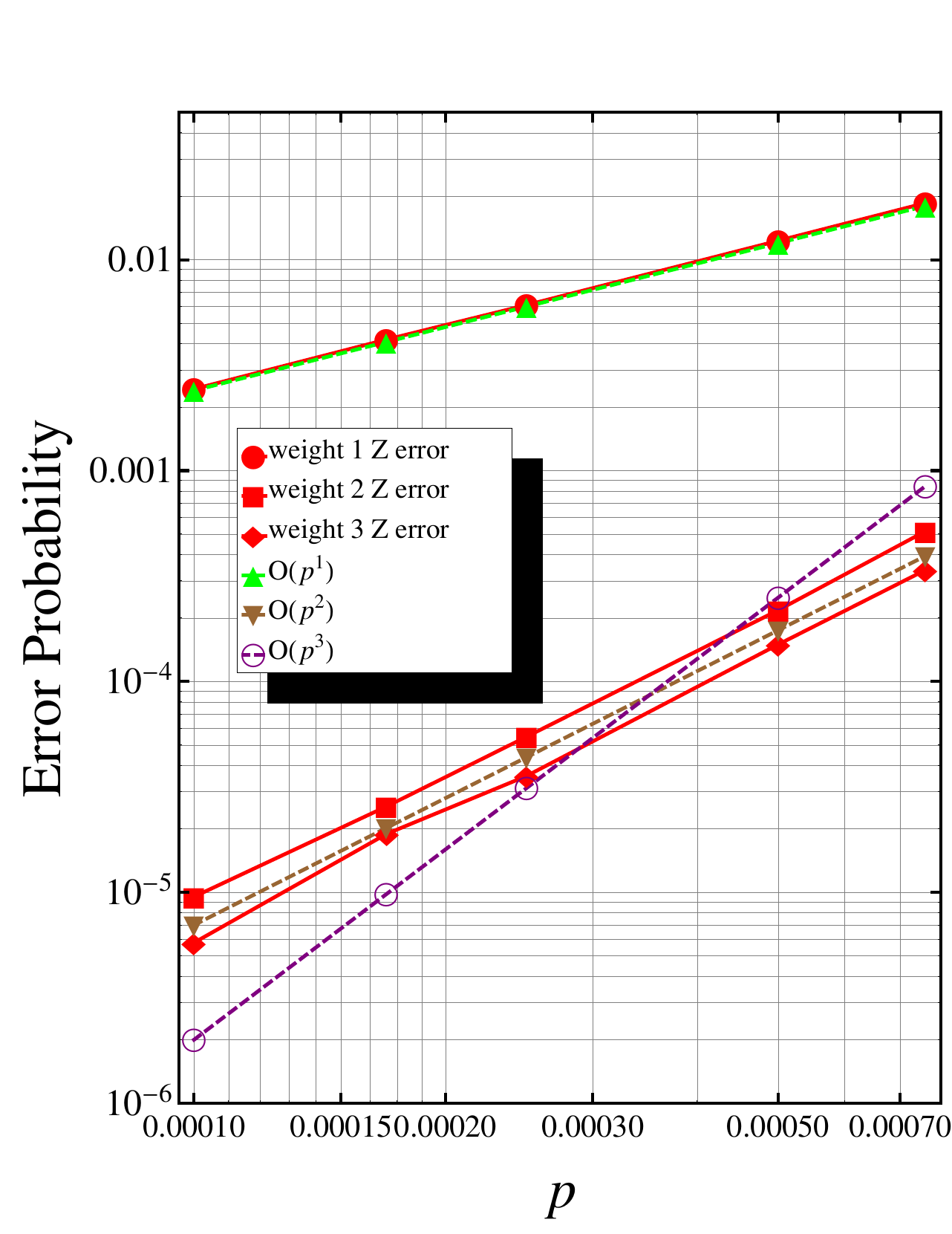}
\caption{\label{fig:Hamming_Z}(Color online) Weight distribution of $Z$ errors remaining on output blocks using the Hamming code for both rounds with postselection.  }
\end{figure}

Probabilities of both types of errors with weight larger than 2 scale as $O(p^2)$ even with postselection, as indicated by Theorem~\ref{thm:main}. The Hamming code is not able to reduce the probability of wrong syndrome estimation to $O(p^3)$. Although the postselection removes as many incompatible blocks as possible, output blocks with compatible syndromes can contain weight-3 errors with probability $O(p^2)$.

In addition, we see that the probabilities for weight 2 and 3 are close to each other for $Z$ errors, while the probability for weight-3 $X$ errors is much smaller than weight-2 $X$ errors. This is because $Z$ errors are propagated before the second stage, and the rates of syndrome flips caused by weight-2 and weight-3 $Z$ errors are rather close.
\begin{figure}[!ht]
\centering\includegraphics[width=80mm]{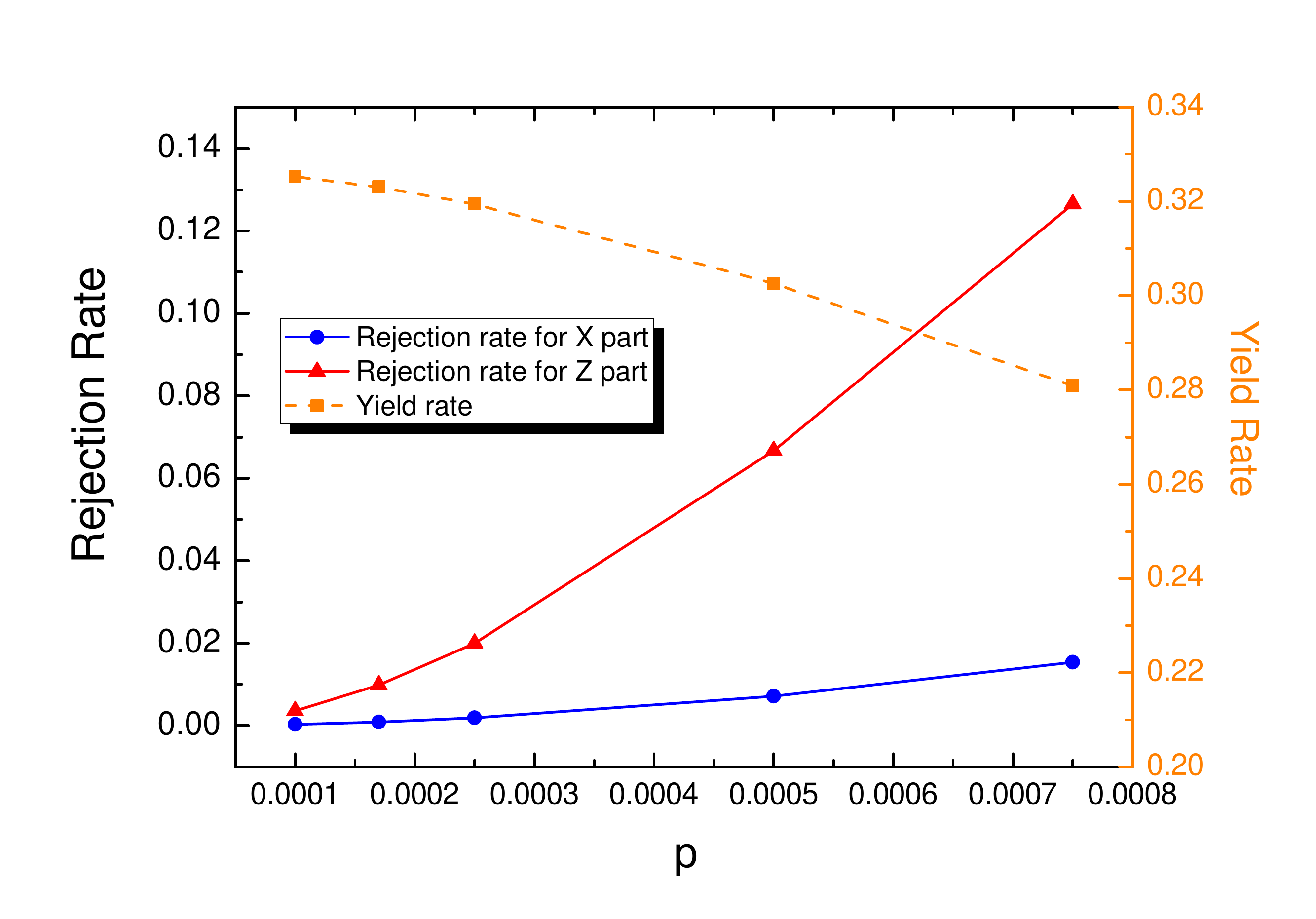}
\caption{\label{fig:rej_HH}(Color online) The rejection rate for both rounds of distillation and overall yield using the Hamming codes. }
\end{figure}

The low distance of the Hamming code also increases the effective error rate for both types of errors. For example, for $p=0.0001$,  $p_{\text{eff}|Z}=1.5\times 10^{-3}$, about four times higher than for Combination A ($[15,7,5]+[15,7,5]$), while $p_{\text{eff}|X}=2.7\times10^{-3}$, compared with $1.67\times 10^{-3}$ for Combination A.

The rejection rate and overall yield for the Hamming code combination are shown in Fig.~\ref{fig:rej_HH}, which are comparable to those for the distance-5 code combinations. This also suggests that the $[23,8,12]$ and $[23,7,11]$ codes have removed almost all blocks with incompatible general syndromes, and postselection can make no further improvement with the limited error correction ability of classical codes.
Overall, the yield is quite high, because of the high rate of the Hamming code, but the distilled states are not qualified, and the rate of the remaining effective errors is higher (especially for $Z$ errors) than  for the distance-5 code combinations studied in last subsection.

\subsubsection{Combination D: [3,1,3]+[3,1,3]}
The final classical code combination uses the $[3,1,3]$ repetition code for both rounds of distillation. It is much simpler than the Hamming code but with a lower rate. If the distillation circuit is perfect, the distilled state will have much higher fidelity compared with states distilled by Hamming codes~\cite{Ancilla_distillation_1}.

\begin{figure}[!ht]
\centering\includegraphics[width=60mm]{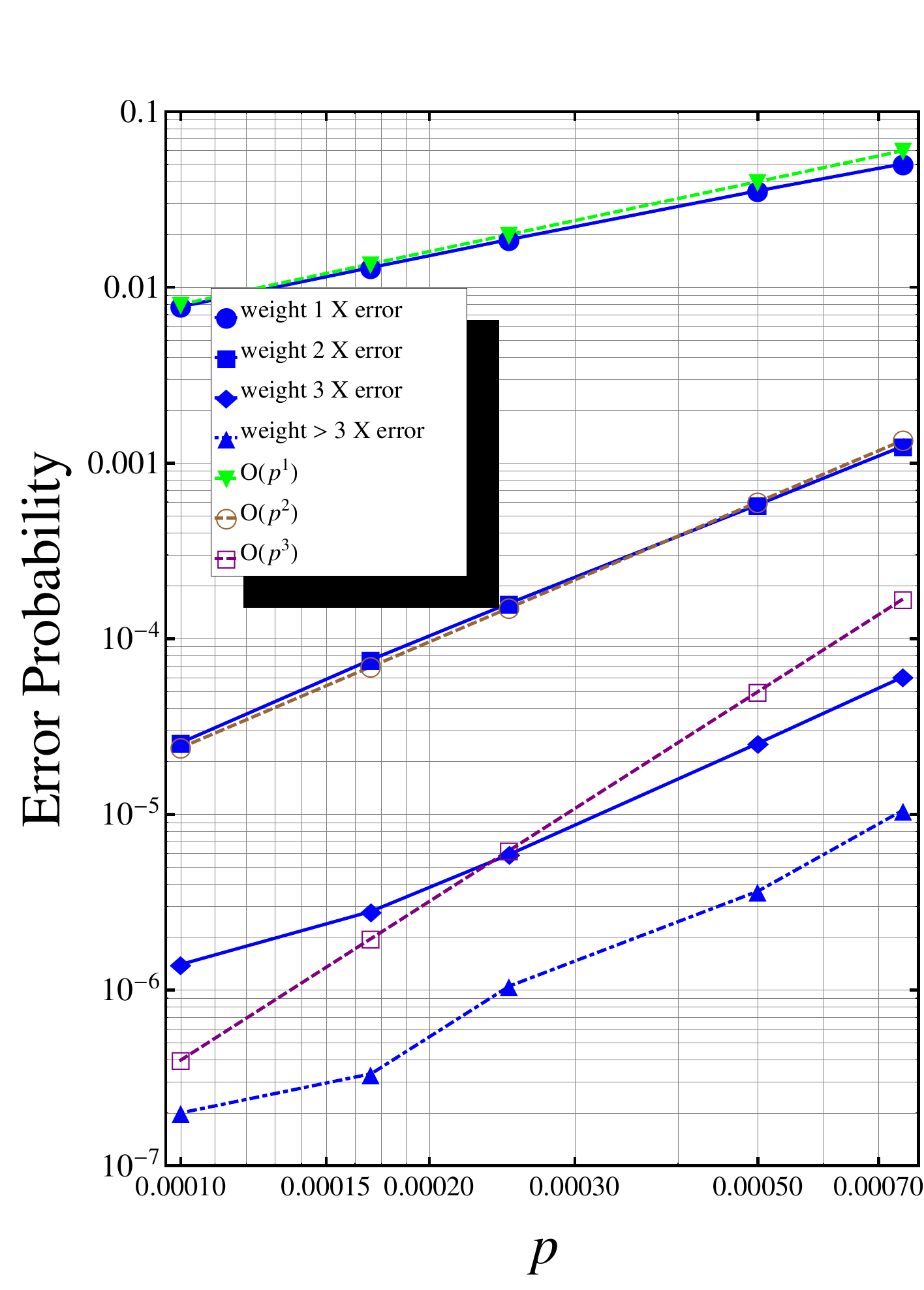}
\caption{\label{fig:33x}(Color online) Weight distribution of $X$ errors remaining on output blocks using the $[3,1,3]$ code for both rounds with postselection.  }
\end{figure}

\begin{figure}[!ht]
\centering\includegraphics[width=60mm]{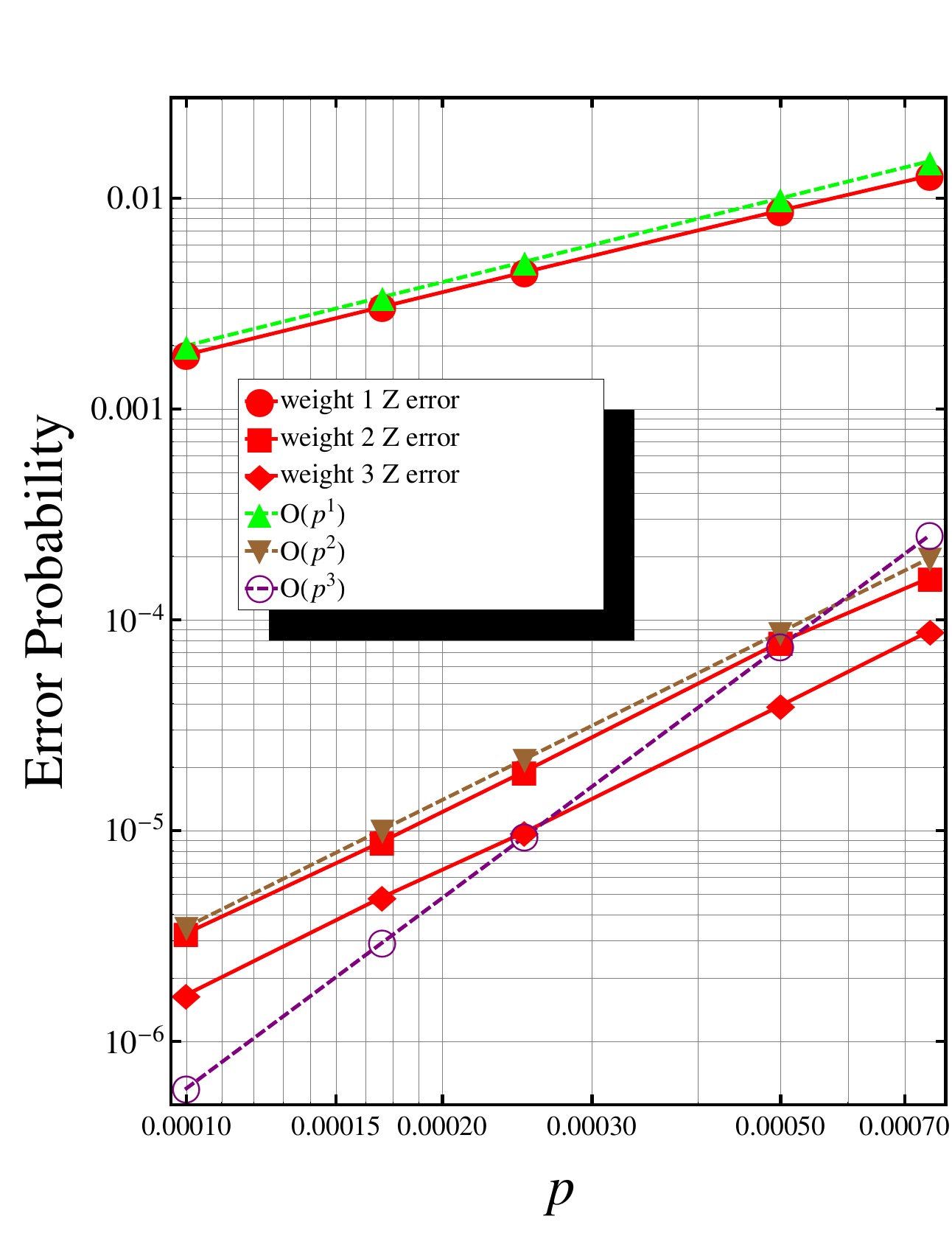}
\caption{\label{fig:33z}(Color online) Weight distribution of $Z$ errors remaining on output blocks using the $[3,1,3]$ code for both rounds with postselection. }
\end{figure}

Error weight distributions for both types of errors are shown in Fig.~\ref{fig:33x} and Fig.~\ref{fig:33z}. As for the Hamming code, the asymptotic behavior of probability of errors with weight larger than 2 is stuck at $O(p^2)$, but the error rate is reduced because of stronger error correction ability and a simplified distillation circuit. One may also note that the difference between weight-2 and 3 $Z$ errors is also  slightly increased.

\begin{figure}[!ht]
\centering\includegraphics[width=80mm]{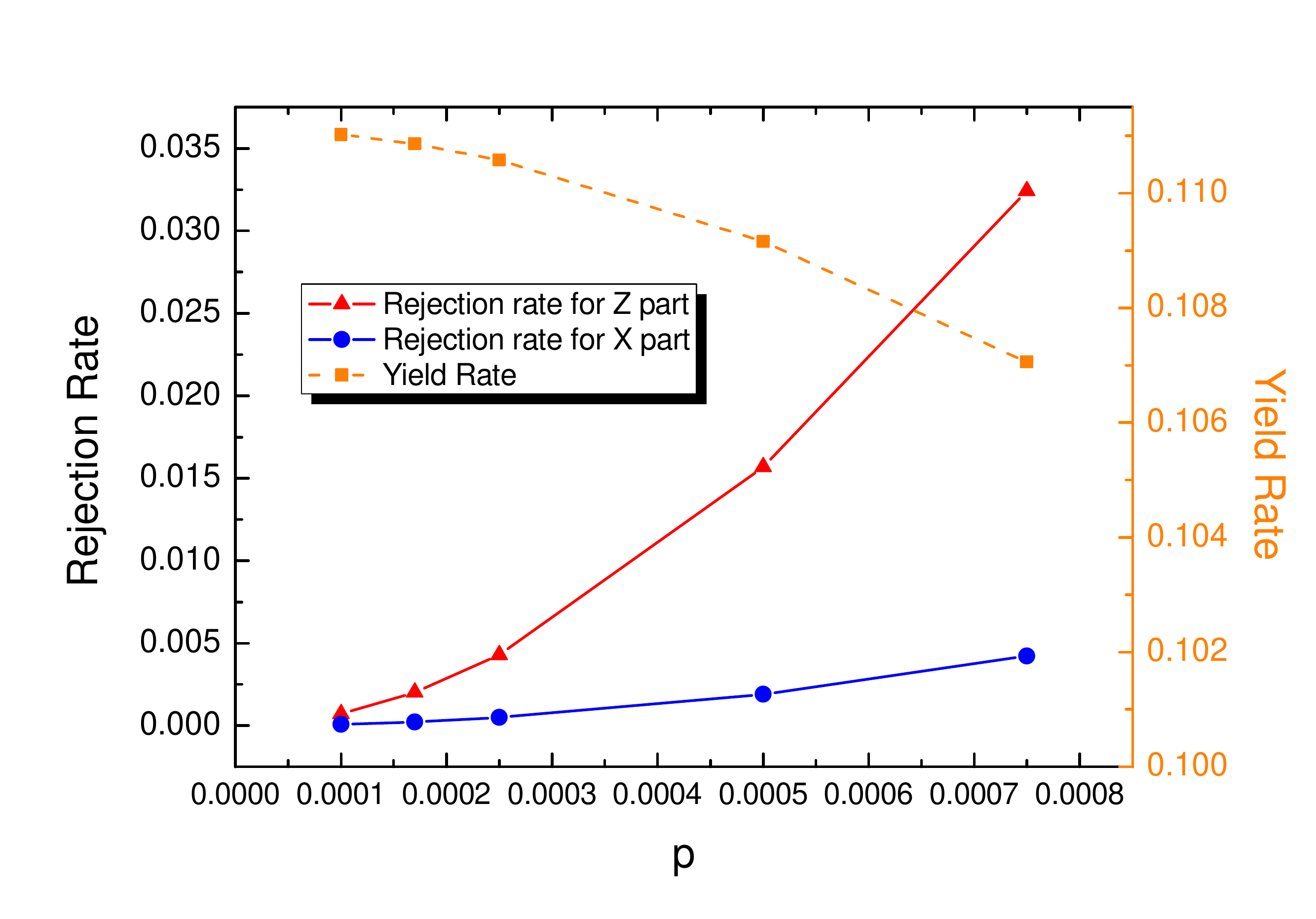}
\caption{\label{fig:33rej}(Color online) The rejection rate for both rounds of distillation and overall yield using the $[3,1,3]$ codes. }
\end{figure}

The effective error rates for $X$ and $Z$ errors remaining are $p_{\text{eff}\ |X}=2.18 \times 10^{-3}$ and $p_{\text{eff}\ |Z}=1\times 10^{-3}$ for $p=0.0001$, respectively, all slightly less than for the Hamming code combination. Although error propagation is not serious in this combination, the limited error correction ability of the $[3,1,3]$ code mans that $p_{\text{eff}\ |X}$ and $p_{\text{eff}\ |X}$ are still higher than for Combinations A and B. Figure~\ref{fig:33rej} shows the overall yield and rejection rates for $[3,1,3]+[3,1,3]$. The rejection rates are the lowest compared with the other classical code combinations partly due to the simpler distillation circuit introducing low rates of incompatible generalized syndromes.  The overall yield is also low, due to the low encoding rate.

\subsection{Summary}
We have studied four combinations of classical codes to distill $|0\>_L$ for the quantum Golay code. Figure~\ref{fig:yvsp} gives a summary, which shows the effective error rate and overall yield for these combinations for $p = 0.0001$. Two combinations of distance-5 codes (blue) sit on the left side, while two combinations of distance-3 codes (red) sit on the  right side. By contrast, the distance-3 codes sit on top of the distance-5 codes. And the effective error rates for $X$ errors are several times larger than for $Z$ errors, because of the asymmetry of the two rounds of distillation.
\begin{figure}[!ht]
\centering\includegraphics[width=85mm]{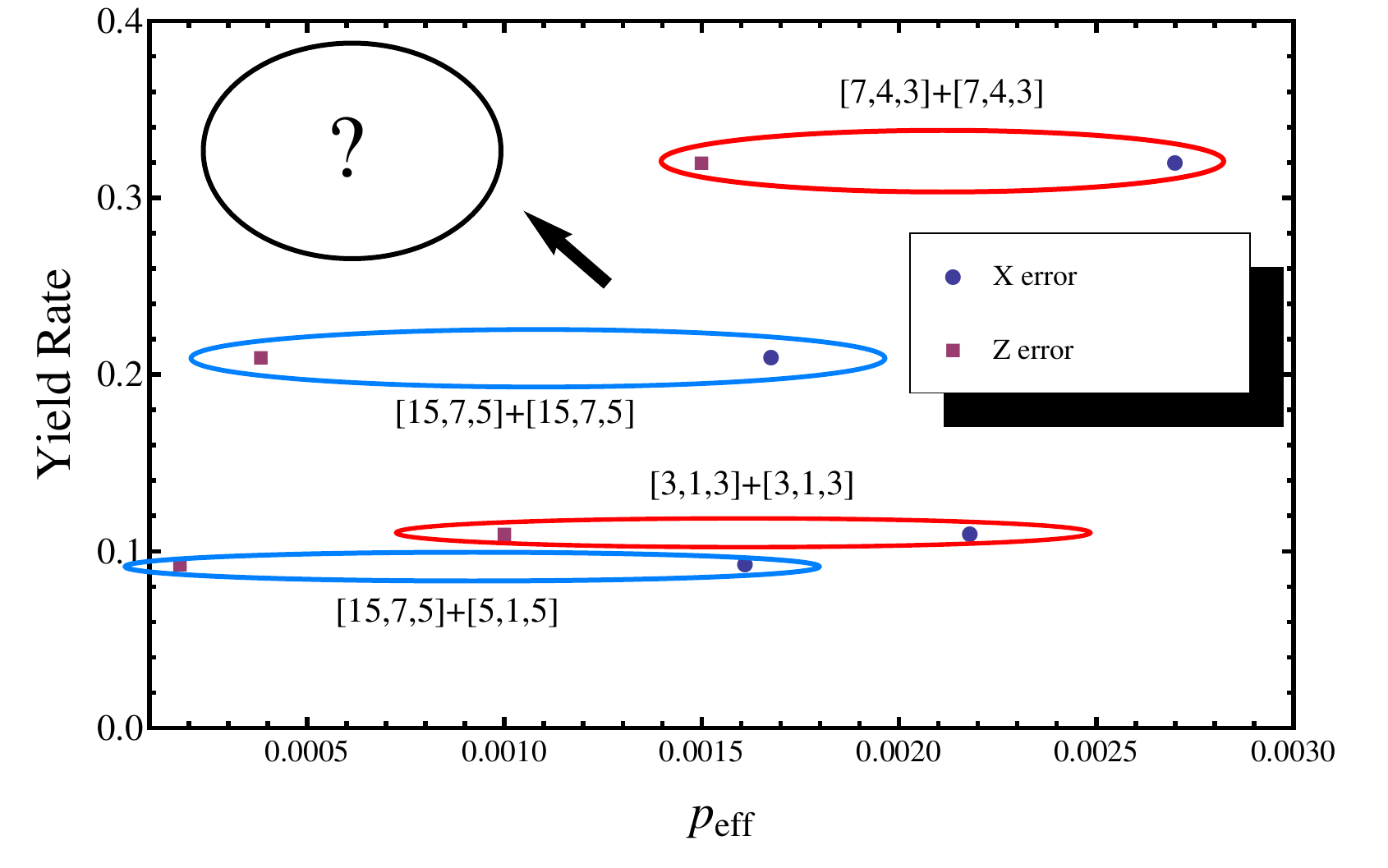}
\caption{\label{fig:yvsp}(Color online) Tradeoff between effective error rate and overall yield for different combinations of classical codes for $p=0.0001$.  }
\end{figure}

A combination of classical codes to move the circle to the top-left corner is desirable for efficient FTQC.
One can certainly use larger classical codes with higher rates, but such codes usually have a more complicated $\textsf{A}$ matrix in the parity-check matrix, and this amplifies the effective error rate, and may increase the rejection rate to an unacceptable level, unless $p$ is extremely small. The desired classical codes should have a sparse $\textsf{A}$ matrix, strong error-correcting ability, and high rate. Certain classes of LDPC codes may possess this ability. More specifically, we would want a family of codes such that the number of 1s in each column of their $\textsf{A}$ matrices is bounded by a small constant, while its error correction ability and encoding rate grow with code size. The performance of these LDPC codes is crucial for practical FTQC using large quantum block codes, and needs detailed further study.
It may be possible to use the parity-check matrix $\textsf{H}_{\text{LDPC}}$ of an LDPC code as the desired sparse matrix $\textsf{A}$, which will lead to a classical code
with parity-check matrix $\textsf{H}=[\textsf{I}\ |\ \textsf{H}_{\text{LDPC}}]$.
This is very like the type of data-syndrome codes studied in~\cite{ashikhmin2016correction}, and is a subject of our future work.


\section{Discussion and Future Work}\label{sec:discussion}
The main obstacle to implement FTQC using large block codes is the lack of an efficient and fault-tolerant preparation protocol for the required set of stabilizer ancilla states. Here, one needs not only to prepare ancillas with high fidelity, but also to suppress high-weight errors in each code block. Our new fault-tolerant distillation protocol greatly improves on the one in ~\cite{Ancilla_distillation_1}, and solves this challenge to a large extent. We use classical error-correcting codes to estimate the generalized syndrome sets, and classical error-detecting codes to pick out those ancillas with compatible estimated generalized syndromes from the output blocks, so that the quality of the output ancillas can be ensured. We proved that the output blocks are qualified if sufficiently powerful error-correcting codes and error-detecting codes are used, and the rejection rate of postselection can be very low in the low failure rate regime. Numerical simulations of fault-tolerantly preparing the $[[23,1,7]]$ quantum Golay code were carried out, which verify the validity of the protocol with small classical codes. The quality of the output ancillas is also greatly improved, with negligible extra cost introduced by postselection.

The error rate for residual $X$ errors is much higher than $Z$ errors for $|0\>_L$ state distillation, because of the $X$ error propagation in the second stage of the distillation. As a result, more $X$ errors propagate to data block than $Z$ errors.
This may suggest that we choose asymmetric CSS codes~\cite{brooks2013asymmetric} as our data block codes, which can correct more $X$ errors than $Z$ errors, or vice versa.

Error propagation during distillation can be effectively suppressed by reducing the number of 1s in each column of the $\textsf{A}$ matrices for $\cC_{c_1}$ and $\cC_{c_2}$. At the same time, $\cC_{c_1}$ and $\cC_{c_2}$  should possess strong error-correcting ability and high encoding rate (so that high distillation throughput is available). For small classical codes, these requirements somewhat conflict with each other. On the other hand, medium or large-sized LDPC codes can simultaneously fulfill these requirements, and have the potential to generate qualified ancillas with low effective error rate and high throughput.
This makes LDPC codes an important candidate for our protocol. The actual performance of LDPC codes may determine the overall resource overhead of large block code FTQC, and needs careful study.

For FTQC using large block codes, we are more interested in highly efficient quantum codes with large length, such as the $[[127,57,11]]$, or $[[255,143,15]]$ quantum BCH codes~\cite{Grassl:1999:207} concatenated with the quantum Golay code, or quantum LDPC codes \cite{tillich2014quantumLDPC,Hsieh:2011:1761}. According to Theorem~\ref{thm:main}, their large $t$ values suggest that the classical codes used in the distillation procedure should also have large distances. An extra difficulty is that these codes are usually not perfect codes, which means that the quantum decoder may not be able to take the corrupted state back to the code space if the weight of residual error is too high. So, even if the estimated generalized syndromes of the output blocks are correct, there can still exist high-weight residual errors on the output blocks after quantum error correction. This implies that we may need another layer of postselection for general quantum codes. For example, consider distillation of $|0\>_L^{\otimes 57}$ for the $[[127,57,11]]$ quantum BCH code where we just try to remove $X$ errors. The corresponding $\cC_X$ and $\cC_Z$ are the $[127,92,11]$ classical BCH code. For the estimated generalized syndromes $\tilde{g}_Z$ obtained from classical decoding of $\cC_{c_\mu}$, one can use the Berlekamp-Massey algorithm~\cite{berlekamp1968algebraicBCH, massey1969BCH} for the $[127,92,11]$ code to estimated the $X$ errors, and then follow the recipe in Sec.~\ref{sec:noiseless_prep} to remove them. Suppose the binary representation of the decoded error is $\tilde{e}$. Then, if the quantum decoding fails, in general, $\textsf{H}_Z \tilde{e}^T\neq \tilde{g}_Z^T$. If this is true, then blocks with such estimated generalized syndromes cannot be used, and should be discarded.
The cost of this extra layer of postselection needs further study.

Another potential way to further increase the efficiency of the distillation procedure is to take advantage of properties of the encoding circuits themselves~\cite{Aliferis:2007:220502}, and the permutation symmetries of the quantum codes ~\cite{paetznick_ben2011fault,chao2017fault}. For small quantum codes like the Steane code and Golay code, the encoding circuit may only generate a small subset of correlated errors, as observed in ~\cite{Aliferis:2007:220502} and ~\cite{paetznick_ben2011fault}. So it is possible to permute the qubits after encoding so that the correlated errors for differently prepared blocks are different from each other, and then verify the different prepared ancillas against each other. This has been shown to greatly reduce the resource overhead needed for ancilla verification~\cite{paetznick_ben2011fault}. Such properties for large block quantum codes remain unexplored, and could potentially increase the throughput of the distillation protocol without extra cost.

Our distillation protocol can simultaneously generate many copies of the \emph{same} stabilizer ancilla states, which makes it very suitable for FT quantum simulation algorithms
~\cite{wecker2014gate,Hastings:2015_simulation,Poulin:2015qic_simulation}. It has been suggested in this context that a fixed small set of Clifford circuits are repeatedly applied in the algorithms. Whether these Clifford circuits as a whole can be implemented through a single logical state measurement, and whether the stabilizer ancillas needed for such measurements can be prepared by our protocol, are open problems. If the answers to these questions are positive, one may take advantage of our distillation protocol to generate a fixed, small set of high quality ancilla states fault-tolerantly, with high throughput, to implement these complicated Clifford circuits in a single step. Hopefully, this will greatly accelerate the speed of quantum simulation and reduce the resource overhead.


Finally, the distillation protocol has the potential to be generalized to prepare high fidelity codeword states of quantum error correction codes for Bosonic modes fault-tolerantly. Such codes include the GKP code~\cite{Gottesman:2001:012310} and binomial code~\cite{michael2016binomial_code}, whose preparations are quite difficult due to the formidable complexity of the required operations~\cite{GKP2002preparing,GKP2010preparing,terhal2016prepGKP}.

\begin{acknowledgments}
We thank Ben Reichardt, Hui-Khoon Ng and Yingkai Ouyang for useful discussions on state verification. YCZ and CYL acknowledge useful conversations with Kung-Chuan Hsu and Chung-Chin Lu on decoding algorithms for cyclic codes. This work is partially funded by the Singapore Ministry of Education
(Tier-1 funding), through Yale-NUS Internal Grant IG14-LR001, and by NSF Grant No. CCF-1421078 and Grant No. MPS-1719778. TAB also acknowledges funding from an IBM Einstein Fellowship at the Institute for Advanced Study.
\end{acknowledgments}

\appendix

\section{Distillation of ancilla states for FTQC}\label{appendix:ancilla_states}

\subsection{Ancilla States for FTQC}\label{appendix:ancilla_state_1}
Preparing ancilla states that are stabilizer states is sufficient for universal FTQC. For example, these are the clean ancilla states required for FTQC in Ref.~\cite{brun2015teleportation}:
\begin{enumerate}
  \item $|0\>^{\otimes k}_L$ and $|+\>^{\otimes k}_L$ are used as the $X$ and $Z$ ancillas in Steane's syndrome extraction.
  \item  $|0_1,...,+_j,0_{j+1},...,0_{k}\>_L$ $($or $|+_1,...,0_{j},+_{j+1},...,+_{k}\>_L)$. This state has all logical qubits in $|0\>_L$ $(\text{or}\ |+\>_L)$, but the $j$th qubit in $|+\>_L$ $(\text{or}\ |0\>_L)$. It is the simultaneous +1 eigenstate of logical operators $\{\bar{Z}_i, \bar{X}_j, \forall i\neq j\}$ $(\text{or}\ \{\bar{X}_i, \bar{Z}_j, \forall i\neq j\})$. We denote such a state as $\|Z|\bar{X}_j\>_L$ $($or $\|X|\bar{Z}_j\>_L)$. Here the capital $Z$ $($or $X)$ before the vertical line means all other logical qubits but the $j$th one in the block are prepared in the +1 eigenstate of $\bar{Z}$ $($or $\bar{X})$. If $|0\>_L^{\otimes k}$ and $\|X|\bar{Z}_j\>$ are used as the $X$ and $Z$ ancillas, one can measure $\bar{Z}_j$ on the data block. Similarly, if $|+\>_L^{\otimes k}$ and $\|Z|\bar{X}_j\>_L$ are used as the $Z$ and $X$ ancillas, one can measure $\bar{X}_j$. It is easy to check that one can simultaneously measure $\bar{X}_{i_1},\dots, \bar{X}_{i_{k^\prime}}$ $($or $\bar{Z}_{i_1},\dots, \bar{Z}_{i_{k^\prime}})$ by using $\|Z|\bar{X}_{i_1},\dots, \bar{X}_{i_{k^\prime}}\>$ $($or $\|X|\bar{Z}_{i_1},\dots, \bar{Z}_{i_{k^\prime}}\>)$ for $k^\prime < k$ as the $X$ $(\text{or}\ Z)$ ancilla. Similarly, one can also measure $\bar{Z}_i\bar{Z_j}$ $($or $\bar{X}_i\bar{X}_j)$ by preparing $\|X|\bar{Z}_i\bar{Z}_j\>$ or $(\text{or}\ \|Z|\bar{X}_i\bar{X}_j\>)$. A logical CNOT between different logical qubits in the same data block can also be implemented in this way.


  \item One also needs the entangled Bell state between the $i$th logical qubit in the block $a$ of an $[[n_a,k_a]]$ code $\mathcal{Q}_a$ and the $j$th logical qubit in block $b$ of an $[[n_b,k_b]]$ code $\mathcal{Q}_b$ (where $\mathcal{Q}_a$ and $\mathcal{Q}_b$ can be different CSS codes):
      \beq
      \left|\Phi_{+,ij}^{(a,b)}\right\>_L = \frac{1}{\sqrt{2}}\left(|0_i\>_L^{(a)}\otimes|0_j\>_L^{(b)}
      +|1_i\>_L^{(a)}\otimes |1_j\>_L^{(b)}\right),
      \eeq
      which is the simultaneous +1 eigenstate of logical $\bar{X}_i^{(a)}\otimes\bar{X}^{(b)}_j$ and $\bar{Z}_i^{(a)}\otimes\bar{Z}^{(b)}_j$. For our purpose, we set all the other logical qubits in block $a$ and block $b$ in $|0\>_L$ state, and this state is denoted as
      \beq
      \left\|Z|\Phi_{+,ij}^{(a,b)}\right\>_L \equiv \left\|
      Z^{(a)}\left|
      \begin{subarray}{l}
      \bar{Z}_i^{(a)}\bigotimes \bar{Z}^{(b)}_j\\
            \bar{X}_i^{(a)}\bigotimes \bar{X}^{(b)}_j
      \end{subarray}
      \right|
      Z^{(b)}
      \right\rangle_L.
      \eeq
      Similarly, one can also define
      \beq
      \left\|X|\Phi_{+,ij}^{(a,b)}\right\>_L \equiv\left\|
      X^{(a)}\left|
      \begin{subarray}{l}
      \bar{Z}_i^{(a)}\bigotimes \bar{Z}^{(b)}_j\\
            \bar{X}_i^{(a)}\bigotimes \bar{X}^{(b)}_j
      \end{subarray}
      \right|
      X^{(b)}
      \right\rangle_L
      \eeq
      These logical Bell states can be used to teleport logical qubits between different blocks of different quantum codes. For example, consider a joint code block consisting of both blocks $a$ and $b$. One can set $\left\|Z|\Phi_{+,ij}^{(a,b)}\right\>_L$ as the $X$ ancilla and $|+\>^{\otimes k_a}_L\otimes |+\>^{\otimes k_b}_L$ as the $Z$ ancilla for the joint blocks, and then jointly measure $\bar{X}_i^{(a)}\otimes\bar{X}_j^{(b)}$. Similarly, by setting $|0\>^{\otimes k_a}_L\otimes |0\>^{\otimes k_b}_L$ as the $X$ ancilla and $\left\|X|\Phi_{+,ij}^{(a,b)}\right\>_L$ as the $Z$ ancilla, one can jointly measure $\bar{Z}^{(a)}_i\otimes \bar{Z}_j^{(b)}$. Together, one can do Bell measurements across the blocks and teleport one logical qubit of block $a$ to the one in block $b$. This is particularly useful in FTQC. For instance, one can teleport a logical qubit from an efficient storage code block, but not having a fault-tolerant non-Clifford gate, to a processor code block that supports a transversal logical $T$ gate (like the quantum Reed Muller code~\cite{steane1999quantumRM}), and then teleport it back after the logical $T$ is finished. This technique can to get around the Eastin-Knill theorem~\cite{eastin2009restrictions,zeng2011transversality}, and replace magic state distillation~\cite{Bravyi:2005:022316}, and (hopefully) reduce the overall resource overhead in the computation.

  \item Entangled states between the $i$th logical qubit of block $a$ and $j$th logical qubit of block $b$:
      \beq
      \begin{split}
      \left|\Omega_{ij}^{(a,b)}\right\>_L=&\frac{1}{2}\left(|0_i\>^{(a)}_L\otimes |0_j\>^{(b)}_L+|0_i\>^{(a)}_L\otimes|1_j\>^{(b)}_L\right.\\
      &\left.+|1_i\>^{(a)}_L\otimes|0_j\>^{(b)}_L-|1_i\>^{(a)}_L\otimes|1_j\>^{(b)}_L\right).
      \end{split}
      \eeq
      This ancilla is the joint +1 eigenstate of $\bar{X}^{(a)}_i\otimes\bar{Z}^{(b)}_j$ and $\bar{Z}^{(a)}_i\otimes\bar{X}^{(b)}_j$.
      One such state is
      \beq
      \left\|X^{(a)},Z^{(b)} \ \big| \ \Omega_{ij}^{(a,b)}\right\>_L \equiv \left\|
      X^{(a)}\left|
      \begin{subarray}{l}
      \bar{Z}_i^{(a)}\bigotimes \bar{X}^{(b)}_j\\
            \bar{X}_i^{(a)}\bigotimes \bar{Z}^{(b)}_j
      \end{subarray}
      \right|
      Z^{(b)}
      \right\rangle_L,
      \eeq
      where the remaining logical qubits of blocks $a$ and $b$ are in $|+\>_L$ and $|0\>_L$, respectively. In this case, if block $a$ is used as the $Z$ ancilla and block $b$ is used as the $X$ ancilla, the Steane syndrome extraction circuit will measure $\bar{X}_i\bar{Z}_j$ on the data block, and this can be used to implement a logical Hadamard gate~\cite{brun2015teleportation}.
  \item Entangled states between the $j$th logical qubit of block $a$ and the $j$th logical qubit of block $b$:
\beq
\begin{split}
\left|\Theta^{(a,b)}_{jj}\right\>_L=&\frac{1}{2}\left(|0_j\>^{(a)}_L\otimes|0_j\>^{(b)}_L+i|0_j\>^{(a)}_L\otimes|1_j\>^{(b)}_L\right.\\
&\left.+|1_j\>^{(a)}_L\otimes|0_j\>^{(b)}_L-i|1_j\>^{(a)}_L\otimes|1_j\>^{(b)}_L\right).
\end{split}
\eeq
This state is the simultaneous +1 eigenstate of $\bar{Z}_j^{(a)}\otimes\bar{Y}_j^{(b)}$ and $\bar{X}_j^{(a)}\otimes \bar{Z}_j^{(b)}$. One such state is
      \beq
      \left\|X^{(a)}, Z^{(b)}\ \big| \ \Theta_{jj}^{(a,b)}\right\>_L\equiv \left\|
      X^{(a)}\left|
      \begin{subarray}{l}
      \bar{Z}_j^{(a)}\bigotimes \bar{Y}^{(b)}_j\\
            \bar{X}_j^{(a)}\bigotimes \bar{Z}^{(b)}_j
      \end{subarray}
      \right|
      Z^{(b)}
      \right\rangle_L.
      \eeq
This state can be used to measure logical $\bar{Y}_j$ on the data block when setting block $a$ as the $Z$ ancilla and block $b$ as the $X$ ancilla in the Steane syndrome measurement. As a byproduct of this measurement of $\bar{Y}_j$, an eigenstate of $\bar{Y}_j$ is prepared in the data block simultaneously, which can be used to implement a logical phase gate.
\end{enumerate}

Note that all these stabilizer states can be produced by (imperfect) Clifford encoding circuits with CNOT gates only, together with the ability to prepare physical qubits in $|0\>$ and $|+\>$.

\subsection{Distillation circuits}\label{sec:logical ancilla}
In this subsection, we generalize the distillation for $|0\>^{\otimes k}_L$ $(\text{or}\ \ |+\>^{\otimes k}_L)$ so that all stabilizer states in Appendix.~\ref{appendix:ancilla_state_1} for FTQC can be well prepared.

$\|Z|\bar{X}_j\>_L$ and $\|X|\bar{Z}_j\>_L$--- The distillation circuit in Fig.~\ref{fig:0_x_z} can also be used to to distill $\|Z|\bar{X}_j\>_L$ and $\|X|\bar{Z}_j\>_L$. Still, there's is some difference in the decoding procedure. For $\|Z|\bar{X}_j\>_L$, we use $\cC_{c_1}$ to check eigenvalues of all $Z$ stabilizer generators and $\bar{Z}_i$ ($i\neq j$) (together form the eigenvalues of $S_1$) in the first round. Eigenvalues of the $X$ stabilizer generators and $\bar{X}_j$ (they together form the eigenvalues of $S_2$) are checked at the second round by $\cC_{c_2}$ to remove $Z$ errors. The procedure is similar for $\|X|\bar{Z}_j\>_L$.

\begin{figure}[!ht]
\centering\includegraphics[width=90mm]{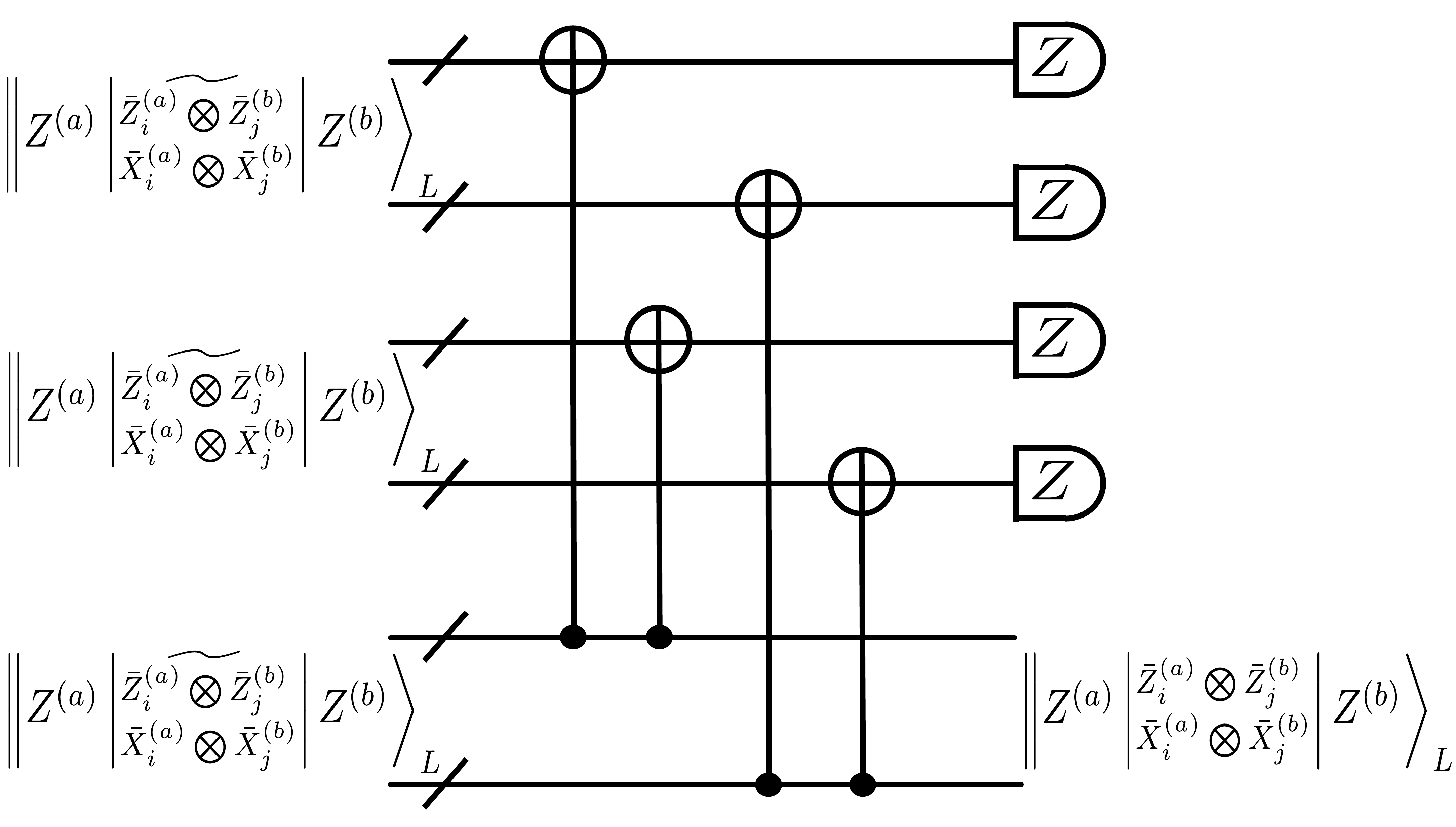}
\caption{\label{fig:zzxx_x}  The distillation circuit for one type of the $|\Phi_{+,ij}^{(a,b)}\>_L$ ancilla distillation by the classical $[3, 1, 3]$ repetition code to remove $X$ errors at the first round.}
\end{figure}

$\left\|Z \big|\Phi_{+,ij}^{(a,b)}\right\>_L$ and $\left\|X  \big|\Phi_{+,ij}^{(a,b)}\right\>_L$---
We may also need to distill states from a pair of two quantum codes $\mathcal{Q}_a$ and $\mathcal{Q}_b$. We emphasize that $\mathcal{Q}_a$ and $\mathcal{Q}_b$ can be different.
For $\left\|Z \big|\Phi_{+,ij}^{(a,b)}\right\>_L$,
the distillation circuit is the same as for $|0\>_L$. For example, we can simultaneously remove $X$ errors from both blocks  as shown in Fig.~\ref{fig:zzxx_x}. By checking the eigenvalues of the $Z$ stabilizer generators for both blocks, and logical operators $\bar{Z}^{(a)}_u$ for $u\neq i$, $\bar{Z}^{(b)}_{u'}$ for $u'\neq j$ and $\bar{Z}_i^{(a)}\otimes \bar{Z}_j^{(b)}$, one can remove $X$ errors on both blocks simultaneously. At the second round of distillation, eigenvalues of the $X$ stabilizer generators and $\bar{X}_i^{(a)}\otimes\bar{X}_j^{(b)}$ are checked to remove $Z$ errors. The process to distill
$\left\|X \big|\Phi_{+,ij}^{(a,b)}\right\>_L$ for $Z$ ancillas is similar.

\begin{figure}[!ht]
\centering\includegraphics[width=90mm]{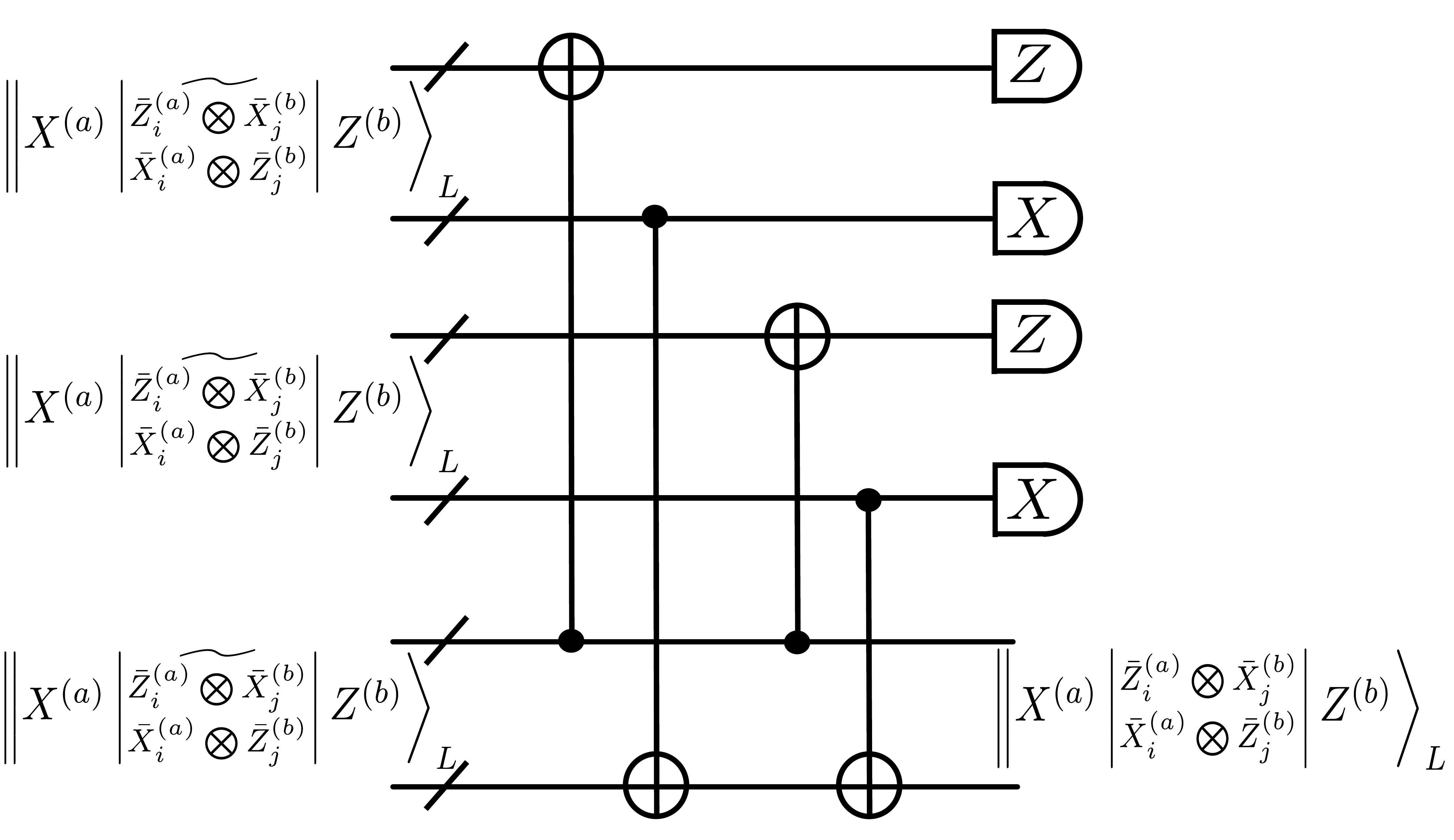}
\caption{\label{fig:xz_xz}The distillation circuit for the $|\Omega_{ij}^{(a,b)}\>_L$ ancilla using classical $[3, 1, 3]$ repetition code to remove $X$ errors from block $a$ and $Z$ errors from block $b$.}
\end{figure}

$\left\|X^{(a)},Z^{(b)} \ \big| \ \Omega_{ij}^{(a,b)}\right\>_L$--- This is another entangled state between the $i$th logical qubit of block $a$ and the $j$th logical qubit of block $b$. The rest of the logical qubits for the first block are are set to $|+\>_L$, while for the second block, they are set to $|0\>_L$. Again, we need to distill the ancilla blocks in pairs. But the situation is different from the previous case, since we need to check the eigenvalues of $\bar{X}_i^{(a)}\otimes \bar{Z}_j^{(b)}$ and $\bar{Z}_i^{(a)}\otimes \bar{X}_j^{(b)}$. The state will be destroyed if both blocks are distilled to remove $X$ or $Z$ errors simultaneously. Instead, in the first round, we distill block $a$ to remove $X$ error and block $b$ to remove $Z$ errors, as shown in Fig.~\ref{fig:xz_xz}. Then in the second round, $Z$ errors of block $a$ and $X$ errors of block $b$ are removed. The same procedure can be used to distill more general entangled states such as
\[
\left\|
      X^{(a)}\left|
      \begin{subarray}{l}
      \bar{Z}_i^{(a)}\bigotimes \bar{X}^{(b)}_j\bar{X}^{(b)}_k\\
            \bar{X}_i^{(a)}\bigotimes \bar{Z}^{(b)}_j
      \end{subarray}
      \right|
      Z^{(b)}
      \right\rangle_L,
\]
which is an entangled state between the $i$th logical qubit of block $a$ and the $j$th and $k$th logical qubits of block $b$ .

\vspace{0.4em}
$\left\|X^{(a)},Z^{(b)} \ \big| \ \Theta_{jj}^{(a,b)}\right\>_L$--- This entangled state between the $j$th logical qubit of block $a$ and the $j$th logical qubit of block $b$ is used to measure $\bar{Y}_j$ of the data block $\mathcal{Q}$ when both $\mathcal{Q}_a$ and $\mathcal{Q}_b$ are the same as $\mathcal{Q}$.
\begin{enumerate}
  \item $\mathcal{Q}$ is a symmetric CSS codes.
  \item The weights of the codewords of the underlying classical codes are $0\mod 4$.
  \item The weight of $\bar{X}_j$ is odd.
\end{enumerate}
Then, after preparing the qualified ancilla state
{\small $\left\|X^{(a)},Z^{(b)} \ \big| \ \Omega_{jj}^{(a,b)}\right\>_L$}
, we apply bitwise phase gates {\small $P=\left[
    \begin{array}{cc}
      1 & 0 \\
      0 & \upi \\
    \end{array}
  \right]$ on block $b$}, followed by $\bar{Z}_j^{(b)}$.
This procedure will give a qualified ancilla state
{\small $\left\|X^{(a)},Z^{(b)} \ \big| \ \Theta_{jj}^{(a,b)}\right\>_L$}.
This can be shown as follows: bitwise phase gates will preserve the stabilizer group with the above properties. All other logical qubits will not be affected by the bitwise phase gates since they are in the joint $+1$ eigenstate of $\bar{Z}_{i'}$ ($i'\neq j$).

\newpage
\begin{widetext}
\section{Table of Notation}\label{appendix:table}
\begin{table*}[htp]
\begin{center}
\begin{tabular}{c| c }
  \hline
  \hline
  symbol &  explanation  \\
  \hline
  $\cC_Z$, $\cC_X$ & Classical codes used to construct CSS code $\mathcal{Q}$. \rule{0pt}{2.6ex}\\[2pt]
  $\cC_{c_1}$, $\cC_{c_2}$, $\cC_{d_1}$, $\cC_{d_2}$  & Classical codes for syndrome estimation $(\cC_{c_1}$ and $\cC_{c_2})$ and compatibility $(\cC_{d_1}$ and $\cC_{d_2})$.\\[2pt]
  $G^{(j)}_{Z\,|\,i}$, $G^{(j)}_{X\,|\,i}$, $g^{(j)}_{Z\,|\,i}$, $g^{(j)}_{X\,|\,i}$  & the $i$th $Z$  ($X$) stabilizer generator of $\mathcal{Q}$ for block $j$ and its eigenvalue in the binary form.\\[2pt]
  $L^{(j)}_i$, $\ell_{i}^{(j)}$  & the $i$th stabilizer logical operator and its eigenvalue in the binary form.\\[2pt]
  $S^{(j)}_1$, $S^{(j)}_2$ & Stabilizers of block $j$ whose eigenvalues need to be estimated in the first (second) round of distillation. \\[2pt]
  $\mathcal{SC}^{(j)}_1$, $\mathcal{SC}^{(j)}_2$ & Stabilizer group of block $j$ generated by $S^{(j)}_1$ $(S^{(j)}_2)$.\\[2pt]
  $\text{SC}_{1,i}$, $\text{SC}_{2,i}$ &  The $i$th element in $\mathcal{SC}_1$ $(\mathcal{SC}_2)$ .\\[2pt]
  $SE_1$, $SE_2$ & Extended stabilizers whose eigenvalues needs to be estimated for postselection.\\[2pt]
  $\textsf{s}_\mu$, $\widetilde{\textsf{s}}_\mu$ for $\mu=1,2$  & Array of generalized syndrome bits according to $S_\mu$, and its estimate by classical decoding.\\[2pt]
  $\textsf{se}_\mu$, $\widetilde{\textsf{se}}_\mu$ for $\mu=1,2$  & Array of extended syndrome bits according to $S_\mu'\bigcup S_\mu$, and its estimate by classical decoding.\\[2pt]
  $\textsf{sc}_\mu$, $
  \widetilde{\textsf{sc}}_\mu$ for $\mu=1,2$  & Array of all syndrome bits according to $\mathcal{SC}_\mu$, and their estimates by classical decoding of $\cC_{c_\mu}$.\\[2pt]
  $\sigma_\mu$, $\mu=1,2$  & Parity matrices of generalized syndrome array obtained after bit-wise measurements. \\[2pt]
  $t$, $t_c$ & Number of correctable errors by a quantum (classical) code.  \\[2pt]
  $\mathscr{G}$, $\mathscr{R}$,   & Failure sets that occur in the quantum circuit.\\[2pt]
  $\mathscr{G}(\step)$, $\mathscr{R}(\step)$ & Failures in $\mathscr{G}$ and $\mathscr{R}$ that occur at the $\step$th time step.\\[2pt]
  $\mathscr{G}_X$, $\mathscr{G}_Z$, $\mathscr{R}_X$, $\mathscr{R}_Z$ & Collection of the $X$ $(Z)$ components of the failure set $\mathscr{G}$ ($\mathscr{R}$). \\[2pt]
  $\mathscr{G}_{\rm P}$,  $\mathscr{G}_{\rm D}$& Failures in $\mathscr{G}$ that occur in preparation and distillation stages. \\[2pt]
  $E_\mathscr{G}$ & Effective errors caused by $\mathscr{G}$ (preceding the distillation circuit).\\[2pt]
  $\text{Q}_{\mathscr{G}(\step)}$, $\text{Q}^{(j)}_{\mathscr{G}(\step)}$ &  Support of $\mathscr{G}$ at time step $\step$ and subset of $\text{Q}_{\mathscr{G}(\step)}$ on block $j$. \\[2pt]
  $\text{QE}_{\mathscr{G}}$, $\text{QE}^{(j)}_{\mathscr{G}}$ &  Support of $\mathscr{G}$ for CNOT failures in the distillation circuit.\\[2pt]
  $\text{Q}_{E_\mathscr{G}}$, $\text{Q}^{(j)}_{E_\mathscr{G}}$ &  Support of the effective errors $E_\mathscr{G}$. \\[4pt]
  $S_\mu^{(j)}|_{E_\mathscr{G}}$, $\mathcal{SC}^{(j)}_\mu|_{E_\mathscr{G}}$  & Stabilizers in $S_\mu^{(j)}$ ($\mathcal{SC}_\mu^{(j)}$) that anticommute with $E_\mathscr{G}$. \\[4pt]
  $\mathcal{I}$ & A function that extracts the indices of a set of qubits or a set of stabilizers.\\[2pt]
  $\mathcal{B}_{\mathscr{G}}$ & Index set of ancilla blocks affected by $E_\mathscr{G}$. \\[2pt]
  $\mathcal{B}_{\text{f}}$ & Index set of ancilla blocks with non-zero syndromes of a good syndrome pattern.\\[2pt]
  \hline
  \hline
\end{tabular}
\caption{\label{tab:notation} Table of notation used in the paper.}
\end{center}
\end{table*}
\end{widetext}

\newpage
\bibliographystyle{apsrev4-1}
%

\end{document}